\documentclass[aps,prx,twocolumn,amsmath,nofootinbib,longbibliography,amssymb,superscriptaddress,10pt]{revtex4-2}

\usepackage[english]{babel}
\usepackage{etoolbox}
\usepackage[shortlabels]{enumitem}
\usepackage{amsmath,amsthm,amssymb,amsfonts}
\usepackage{stmaryrd}
\usepackage{accents}
\usepackage{bm}
\usepackage{leftindex}
\usepackage[mathscr]{euscript}
\usepackage[linkcolor=blue,colorlinks=true]{hyperref}
\usepackage{graphicx}
\usepackage[caption=false]{subfig}
\graphicspath{
    {figures/}
}

\newtheorem{theorem}{Theorem}

\newtheorem{lemma}[theorem]{Lemma}

\docsvlist{A,B,C,D,E,F,G,H,I,J,K,L,M,N,O,P,Q,R,S,T,U,V,W,X,Y,Z}

\docsvlist{A,B,C,D,E,F,G,H,I,J,K,L,M,N,O,P,Q,R,S,T,U,V,W,X,Y,Z}

\docsvlist{n,e,E}

\def \bra {\langle}
\def \ket {\rangle}
\def \lr {\leftrightarrow}
\def \ts {\textsuperscript}

\begin{document}

\title{
Absence of Floating Phase in Superconductors with Time-reversal Symmetry Breaking on any Lattice
}
\author{Andrew C. Yuan}
\affiliation{Department of Physics, Stanford University, Stanford, CA 93405, USA}
\date{\today}

\begin{abstract}
Due to the interplay of multi-component order parameters (e.g., a twisted bilayer superconductor with inter-layer Josephson coupling or a frustrated ($n\ge 3$)-band superconductor), a superconductor can possess a $U(1)\times \mathbb{Z}_2$ symmetry, corresponding to the superconducting $T_c$ and time-reversal symmetry breaking transition $T_\text{TRSB}$, respectively.
It was then conjectured that in this class of Hamiltonians, 
there exists a vast parameter regime $\sO$ such that the system exhibits \textit{vestigial} TRSB, i.e., $T_\text{TRSB} > T_c$, while at the boundary $\partial \sO$, the system possesses a single phase transition $T_\text{TRSB}=T_c$.
In this paper, we provide evidence towards this conjecture by mathematically eliminating the possibility of a \textit{floating phase}, i.e., $T_\text{TRSB} < T_c$, for the strong coupling regime.
More specifically, we prove that the correlation functions of $U(1)$ spins are bounded above by that of $\mathbb{Z}_2$ spins for all temperatures and lattice structures (e.g., $\mathbb{Z}^d$ for all $d$).
In particular, this guarantees the existence of high-$T_c$ TRSB (and consequently topological) superconductivity in a large class of Hamiltonians.
Note that the same property can also be proven for a certain parameter regime ($\Delta \ge 4/5$) of the generalized XY model on any lattice structure, despite belonging to an entirely distinct class of $U(1)\times \dZ_2$ Hamiltonians.

\end{abstract}
\maketitle
\section{Introduction}

In the last several years, there has been considerable excitement generated by various theoretical \cite{wang2017topological,kivelson2020proposal,yuan2021strain,yuan2023multiband,laughlin1998magnetic,yuan2023inhomogeneity,bojesen2014phase} 
and experimental \cite{ghosh2020recent,ghosh2021thermodynamic,schemm2014observation,grinenko2021state} proposals concerning the existence of time-reversal symmetry breaking (TRSB) pairing states in various unconventional superconductors.
However, conclusive experimental evidence of these states has so far not been reported.

Recently, it was proposed that high-$T_c$ topological superconductivity can be reliably achieved via twisting two identical 2D layers of BSCCO ($d_{x^2-y^2}$-wave pairing symmetry) relative to each other \cite{can2021high,zhao2023time}. 
At relative orientation $45^\circ$, the 1\ts{st} order inter-layer Josephson coupling $-J_1 \cos \phi$ vanishes so that the 2\ts{nd} order coupling $J_2 \cos 2\phi$ (with a sign that favors an inter-layer phase difference $\phi=\pm \pi/2$) becomes the dominant inter-layer interaction.
This inter-layer term was conjectured to induce a TRSB transition ($d_{x^2-y^2} \pm i d_{xy}$ state) with a critical temperature that is on the same order as the superconducting transition, i.e., $T_\text{TRSB} \approx T_c$,
and thus generating high $T_c$ topological superconductivity
(albeit with a gap magnitude proportional to $J_2$, which is expected to be very small \cite{yuan2023inhomogeneity} in the actual system).
Similar proposals were also made for $30^\circ$-twisted bilayer graphene \cite{yao2018quasicrystalline,pezzini202030,deng2020interlayer,liu2023charge} and inherent 2-component systems \cite{maccari2022effects}, which is especially enticing since it is not limited to 2D layers and may also circumvent the issue of a small gap magnitude in the previous proposal.
Alternatively, it was proposed that in a $(n\ge 3)$-band superconductor, the inter-band 1\ts{st} order Josephson coupling can cause sufficient frustration among the inter-band phase differences so that the ground state configurations may exhibit phase differences which differ from 0 or $\pi$ and thus result in TRSB \cite{bojesen2014phase}.
Similar proposals have been argued to occur in the hole doped Ba$_{1-x}$K$_x$Fe$_2$As$_2$ pnictide superconductor \cite{maiti2013s+} and other classes of materials \cite{mukherjee2011role,lee2009pairing,platt2012mechanism,yerin2017anomalous,yerin2022magneto}.

\begin{figure}[ht]
\includegraphics[width=0.8\columnwidth]{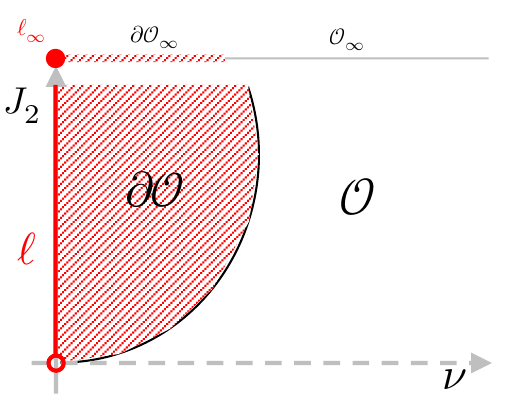}
\caption{Schematic sketch. 
Using the example of a 2-component system \cite{can2021high,maccari2022effects} with phases $\phi^\pm$, the parameter space involving $(\nu,J_2)$ describes the coupling strength $\nu$ of the current-current interaction $-\nu \nabla \phi^+ \cdot \nabla \phi^-$ and that of an inter-component $J_2 \cos 2(\phi^+ -\phi^-)$ term. See, e.g., Eq. (1) of Ref. \cite{maccari2022effects}.
The regime $\sO$ is where vestigial order $T_\text{TRSB} >T_c$ is expected to occur, while the boundary $\partial \sO$ denotes points in which there is only a single phase transition $T_\text{TRSB} =T_c$. 
In particular, it includes the critical regime $\ell \subseteq \partial \sO$, i.e., $\nu=0$, which has been studied extensively via numerics (in 2D \cite{song2022phase} and in 3D \cite{bojesen2014phase,maccari2022effects}). 
The regimes $\ell_\infty, \partial \sO_\infty, \sO_\infty$ correspond to the strong coupling limit, i.e., $J_2 \to \infty$.
We emphasize that the shape of $\partial \sO$ should not be taken too seriously, but from numerics (in 2D \cite{liu2023charge} and 3D \cite{maccari2022effects}), it's expected that for finite $J_2>0$, there exists vestigial order for large $\nu$ and a single phase transition at small $\nu$.
Similarly, the size of $\partial \sO_\infty$ should not be taken seriously, but from numerics \cite{song2022phase,bojesen2014phase}, it's expected that $\ell_\infty \subseteq \partial \sO_\infty$.
Also write $\bar{\sO}\equiv \sO \cup \partial \sO$ and similarly $\bar{\sO}_\infty$.
}
\label{fig:phase-diag}
\end{figure}

Such proposals all fall within the same class of Hamiltonians (defined explicitly in Sec. \eqref{sec:model}) which possesses a $U(1)\times \dZ_2$ symmetry, corresponding to the superconducting $T_c$ and TRSB transition $T_\text{TRSB}$, respectively.
To achieve high $T_c$ topological superconductivity, it is then imperative that the TRSB transition occurs on the same order as the superconducting transition $T_\text{TRSB}\approx T_\text{c}$.
In fact, it was conjectured that in a certain parameter regime (schematically denoted by $\sO$, see Fig. \ref{fig:phase-diag}), the system exhibits \textit{vestigial order}\footnote{
The term denotes the existence of a temperature region $T_{c}<T<T_\text{TRSB}$, in which the individual multi-component order parameters are zero (e.g., schematically, $\bra \psi_1\ket =\bra \psi_2\ket =0$), while the higher order terms are nontrivial (e.g.,$\bra \psi_1^*\psi_2\ket \ne 0$) and thus ``vestigial".
} so that $T_\text{TRSB} > T_c$, while the two transitions coincide exactly $T_\text{TRSB}= T_c$ at the boundary $\partial \sO$ (see Fig. \ref{fig:phase-diag}) \cite{can2021high,bojesen2014phase,liu2023charge,maccari2022effects}.

Within the context of mean-field theory including applied either to a Ginzburg-Landau effective field theory \cite{can2021high,bojesen2014phase,liu2023charge,maccari2022effects,stanev2010three,maiti2013s+} or an XY model on a regularized lattice \cite{yuan2023exactly}, the conjecture is relatively well-understood\footnote{
For example, at the critical regime $\ell$ in Fig. \ref{fig:phase-diag}, the 2\ts{nd} order Josephson coupling first occurs in the quartic order of Ginzburg-Landau theory and thus does not affect the critical temperature of the system (see, for example, Eq. (2) of Ref. \cite{can2021high}).
A rigorous version of this argument can be found in the Appendix of Ref. \cite{yuan2023exactly}.}.
However, it is well known that mean-field theory is reliable only when the effective number of neighboring spins is large as occurs when there are long-range interactions or in high dimensions (generally believed to be $d\ge 4$).
In low dimension ($d=2,3$) systems, fluctuations due to local nearest-neighbor interactions may play an important fact and thus it is conventional to rely on numerical results (in 2D \cite{song2022phase,liu2023charge} and in 3D \cite{bojesen2014phase,maccari2022effects}) and RG calculations (in 2D \cite{zeng2021phase,liu2023charge}) to provide us insight into the problem.
Unfortunately, in certain scenarios, the distinct approaches contradict with each other; for example, at the critical regime $\ell$, numerical evidence agrees with the conjecture $T_\text{TRSB}=T_c$, while latter RG calculations \cite{zeng2021phase} suggest a vestigial TRSB phase $T_\text{TRSB}> T_c$ (although further detailed RG calculations seem to reconcile this matter \cite{how2023absence}).
Therefore, the main goal of this paper is to obtain exact results on the short-range nearest-neighbor model.

For the majority of this paper, we shall consider a class of Hamiltonians which are physically motivated by the strong coupling regime (taking $J_2 \to \infty$ in Fig. \ref{fig:phase-diag}), and prove mathematically that the possibility of a \textit{floating phase}\footnote{The term denotes the existence of a temperature region, $T_{\text{TRSB}}<T<T_c$, in which the multi-component order parameters are individually nonzero, but not yet coupled together to have a fixed phase difference. Hence, each component is ``floating" independently.} is excluded on any lattice structure\footnote{Trivial cases where the transition temperature is $=0,\infty$ are also included} so that $T_\text{TRSB} \ge T_c$.
More specifically, we prove that the correlation functions of $U(1)$ XY spins is bounded above by that of $\dZ_2$ Ising spins for all temperatures and lattice structures (e.g., $\dZ^d$ for all $d$).
In particular, this guarantees the existence of high-$T_c$ superconductivity with TRSB and possibly topological superconductivity\footnote{Even though TRSB is necessary and in general related to topological superconductivity, it is unfortunately, not a sufficient condition. For example, the $d+ig$ state \cite{kivelson2020proposal} breaks TRS but has nodal points along the diagonals.} for a large class of Hamiltonians.
Interestingly, at the boundary $\partial \sO_\infty$ where the system is believed to have a single phase transition $T_\text{TRSB} = T_c$, the correlation inequality implies if the $U(1)$ SC transition has a diverging correlation length at the conjectured $T_\text{TRSB} = T_c$, then so must the $\dZ_2$ TRSB transition.
In a 2D system where the SC transition is believed to be BKT, this implies that the conjectured single phase transition cannot be first order.


The paper is organized as follows. Sec. \eqref{sec:model} introduces the classical statistical model under investigation as well as its motivation.
We also provide a quick sketch of subtle mathematical regularizations necessary to define the thermodynamic limit on a general lattice structure. 
Sec. \eqref{sec:cluster} introduces the \textit{(random) cluster representation}, which acts as an exact dual graphical representation closely related to the Wolff algorithm.
Sec. \eqref{sec:critical} studies the class of Hamiltonian at the critical regime $\ell_\infty$ and uses the cluster representation to prove the main statement, i.e., Theorem \eqref{thm:cluster}.
It should be noted that within this subclass of Hamiltonians, the correlation inequality can be further strengthened using an independent graphical representation, known as the \textit{(random) current representation}. 
In fact, the current representation provides important insight into the underlying reason behind why only a single phase transition is observed at the critical regime\footnote{Moreover, if we make the further arbitrary simplification of replacing the $U(1)$ XY spins with $\dZ_2$ or $\dZ_4$ clock spins, it can be proven rigorously that there exists only a single phase transition. See Appendix \eqref{app:simple} for details} $\ell_\infty$. 
However, due to its complexity and limited use for the general parameter regime, we leave the proof in Appendix \eqref{app:current}.
Sec. \eqref{sec:general} then generalizes the proof to the full class of Hamiltonians $\bar{\sO}_\infty$, i.e., Theorem \eqref{thm:cluster-general} (see also Fig. \ref{fig:nofloat}).

\subsection{A Digression to the generalized XY Model}

It is worth mentioning that the techniques developed in this paper can also be used to understand an entirely different class of $U(1)\times \dZ_2$ Hamiltonian, commonly known as the generalized XY model \cite{lee1985strings,korshunov1986phase}, i.e.,
\begin{equation}
    \label{eq:general-XY}
    H(\theta) =-\sum_{e=ij} \left[ \Delta \cos (\nabla_e \theta) +(1-\Delta) \cos (2\nabla_e \theta)\right] 
\end{equation}
Where the summation is over nearest-neighbors (edges $e=ij$), $0\le \Delta \le 1$ and $\nabla_e \theta = \theta_i-\theta_j$.
In fact, the proof nearly follows directly from Ref. \cite{dubedat2022random} and thus will only be mentioned as a digression.

On a 2D square lattice $\dZ^2$, the model has been studied extensively \cite{song2021hybrid,carpenter1989phase,nui2018correlation,hubscher2013stiffness}, all of which reaching a consistent phase diagram shown in Fig. \ref{fig:general-XY}.
From numerics \cite{song2021hybrid}, it is expected that there exists a critical $\Delta_c \approx 0.33$ such that if $\Delta \ge \Delta_c$, then only a single phase transition occurs $T_\text{nem} = T_\text{ferro}$, while for $\Delta <\Delta_c$, a split transition occurs $T_\text{ferro} < T_\text{nem}$, i.e., an analogue of a floating phase\footnote{Since there is only one component in this system, the term ``floating" is less meaningful.}.
Extension of the model to 3D has also been tentatively explored \cite{gao2022fractional}.

Though we lack comments on the split transition regime, we can show that if $4/5 \le \Delta \le 1$\footnote{This is not surprising, since the Hamiltonian does not become metastable (local min) at $\nabla_e \theta=\pi$ until $\Delta < 4/5$.}, then a similar correlation inequality can be proven and thus implying that $T_\text{ferro} \ge T_\text{nem}$ on any lattice structure.
Since in this regime, the system is expected to have a single phase transition $T_\text{ferro} = T_\text{nem}$ on $\dZ^2$, our previous argument regarding diverging correlation lengths also holds and thus the transition cannot be of first order, consistent with numerics \cite{song2021hybrid}.
We postpone the exact statement to Sec. \eqref{sec:digression}, after introducing the cluster representation in Sec. \eqref{sec:cluster}.

\begin{figure}[t]
\includegraphics[width=1\columnwidth]{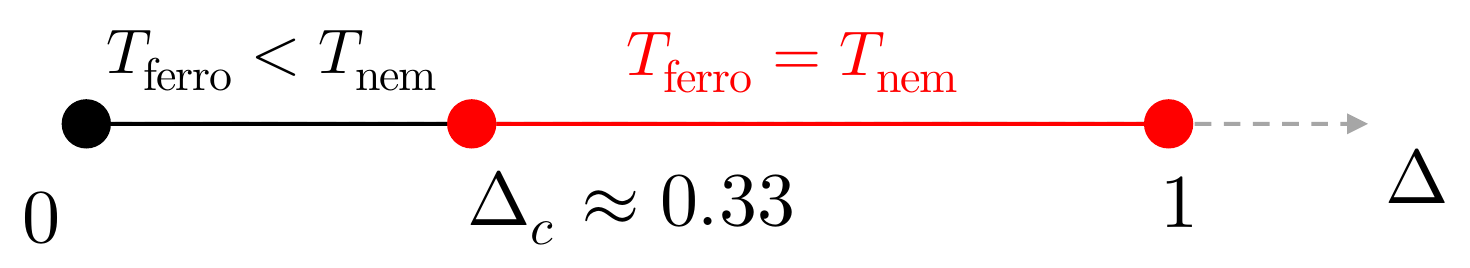}
\caption{Phase diagram of 2D generalized XY model (Eq. \eqref{eq:general-XY}) \cite{song2021hybrid,carpenter1989phase,nui2018correlation,hubscher2013stiffness}. We write $T_\text{ferro},T_\text{nem}$ instead of $T_\text{TRSB},T_c$ since physically, the $\dZ_2$ transition is into a ferromagnetic state instead of TRSB, while the $U(1)$ transition is into a nematic state \cite{carpenter1989phase}.
}
\label{fig:general-XY}
\end{figure}
\section{Model}
\label{sec:model}

For the majority of this paper, we shall consider the following classical statistical $U(1)\times \dZ_2$ Hamiltonian (critical regime $\ell_\infty$ in Fig. \ref{fig:phase-diag})
\begin{equation}
    \label{eq:U1-Z2}
    H(\sigma,\tau) = -\sum_{e =ij} \kappa_e (1 +\tau_i\tau_j) \cos \nabla_e \theta
\end{equation}
Where the summation is over all nearest-neighbors (edges $e=ij$) of an arbitrary graph, $\kappa_e > 0$ denote ferromagnetic coupling constants on each edge $e$, and $\nabla_e \theta = \theta_i -\theta_j$.
Note $\tau_i, \theta_i$ denote the $\dZ_2$ Ising spins and $U(1)$ XY spins corresponding to TRSB and SC, respectively.
For the general parameter regime $\bar{\sO}_\infty$ (see Fig. \ref{fig:phase-diag}), the Hamiltonian is further modified so that
\begin{equation}
    \label{eq:U1-Z2-general}
    H_{\lambda,\alpha}(\sigma,\tau) = H -\sum_{e =ij} \left[\lambda_e \tau_i \tau_j +\alpha_e \cos (2\nabla_e \theta)\right]
\end{equation}
Where $\lambda_e \ge \alpha_e \ge 0$ for all edges $e$\footnote{
Heuristically, $\lambda$ and $\alpha$ are correlated to $T_\text{TRSB}$ and $T_c$, respectively. 
Since the critical Hamiltonian in Eq. \eqref{eq:U1-Z2} is conjectured to have a single phase transition, it is expected that if $\lambda \ge \alpha$, then $T_\text{TRSB} \ge T_c$. 
One of the main results of the paper is to make this statement exact.}.. 


\subsection{Motivation: The Strong Coupling Limit}

Let us first consider the critical subclass $\ell_\infty$ of Hamiltonians in Eq. \eqref{eq:U1-Z2}.
When $\kappa_e=1$ is constant, the critical subclass $\ell_\infty$ describes the strong coupling regime of numerous physical motivations. 
For example, the derivation for the case of a frustrated $(n\ge 3)$-band superconductor can be found in Ref. \cite{bojesen2014phase}; our critical Hamiltonian corresponds to Eq. (8) of Ref. \cite{bojesen2014phase} at the critical point of $(K_1,K_2)=(1,0)$.
Alternatively, in the scenario of TRSB induced in twisted bilayer systems \cite{can2021high,zhao2023time}, one can model the 2D layers via identical standard XY models (so that $\kappa=1$ is the in-plane superfluid stiffness), coupled by a 2\ts{nd} order Josephson coupling $J_2 > 0$, i.e.,
\begin{equation}
    \label{eq:H-J2}
    H_{J_2} (\phi^\pm) = \sum_{s = \pm} H^\text{XY}(\phi^s ) +J_2 \sum_{i} \cos 2\phi_i
\end{equation}
Where $\phi^\pm$ denotes the phases of the XY spins in each layer and $\phi \equiv \phi^+-\phi^-$ denotes the phase difference across the junction.
In this model, it is clear that for any $J_2 >0$, the Hamiltonian is minimized when the phase difference $\phi$ chooses between $\pm \pi/2$ and thus the $\dZ_2$ Ising order corresponds to TRSB $T_\text{TRSB}$.
Moreover, if we take the strong coupling limit $J_2 \to \infty$ \cite{interchangeLimits}, the phase difference $\phi$ is forced to choose between $\pm \pi/2$ at each lattice site $i$ and thus induces an explicit $\dZ_2$ Ising order, i.e., $\phi_i  = \tau_i \times \pi/2$ where $\tau_i=\pm 1$.
On the other hand, the average phase $\theta = (\phi^+ +\phi^-)/2$ is unrestricted and thus induces a $U(1)$ XY spin $\sigma_i = e^{i\theta_i}$ corresponding to the superconducting transition $T_c$.
By change of variables\footnote{A short proof of why $\phi^\pm \to \phi,\theta$ is well-regulated can be found in the Lemma \eqref{app-lem:var-map} in the Appendix.}, the Hamiltonian in Eq. \eqref{eq:H-J2} reduces to Eq. \eqref{eq:U1-Z2} in the strong coupling limit $J_2 \to \infty$.

One may wonder why we should study the strong coupling limit $\ell_\infty$, since it ``appears" to be a singular point.
Apart from simplifying the problem, one reason is that the strong coupling limit is not singular; rather, it is continuously connected to large but finite coupling strengths\footnote{And thus warrants the notation $J_2 \to \infty$ instead of $J_2 = \infty$.} in the sense of correlation functions as argued previously \cite{interchangeLimits}.
Moreover, the authors of Ref. \cite{bojesen2014phase} numerically studied the strong coupling limit on the 3D cubic lattice and showed that the phase diagram does not change qualitatively compared to finite coupling.
Similar numerical results \cite{song2022phase,liu2023charge} were also performed in 2D, in which it was shown that the phase diagram is insensitive to the actual values of the coupling strength $J_2$ (though they did not investigate the strong coupling limit).
In higher dimension $d\ge 4$, mean-field theory predicts a similar behavior \cite{yuan2023inhomogeneity,can2021high}, and thus we believe that studying the strong coupling limit can reflect the properties of the finite coupling model, though admittedly, this is not determined definitively.

Let us now consider the general class $\bar{\sO}_\infty$ of Hamiltonians in Eq. \eqref{eq:U1-Z2-general}. 
As described in Eq. (15) of Ref. \cite{liu2023charge}, the general Hamiltonian of a twisted bilayer system can be written as
\begin{align}
    \label{eq:H-J2-general}
    H_{J_2,\lambda,\alpha} &= H_{J_2} -\lambda \sum_{e=ij} \cos (\nabla_e \phi^+ -\nabla_e \phi^-) \nonumber\\
    &\quad\quad\quad -\alpha \sum_{e=ij} \cos (\nabla_e \phi^+ + \nabla_e\phi^-)
\end{align}
Where $\phi,\theta$ are the phase difference and average phase of each layer $\phi^\pm$ as before and $H_{J_2}$ is defined in Eq. \eqref{eq:H-J2}.
Note that by expanding the $\cos$ terms up to quadratic order, the Hamiltonian obtains a current-current interaction $-\nu \nabla_e \phi^+ \cdot \nabla_e \phi^-$ where $\nu = \lambda-\alpha$ \cite{diag} (in comparison with Eq. (1) of Ref. \cite{maccari2022effects}).
We then repeat the previous construction by taking $J_2 \to \infty$ so that Eq. \eqref{eq:H-J2-general} reduces to Eq. \eqref{eq:U1-Z2-general}.

\subsection{Regularization and the Thermodynamic Limit}

We note that the Hamiltonians in Eq. \eqref{eq:U1-Z2}, \eqref{eq:U1-Z2-general} are only well-defined on a finite graph $G=(V,E)$ where $V$ denotes the finite number of lattice sites and $E$ denotes the possible edges.
To take the thermodynamic limit, the convention is to 
\begin{enumerate}[(a)]
    \item Fix an infinite graph $G_\infty = (V_\infty, E_\infty)$, e.g., the standard $\dZ^d$, with edge couplings $\kappa: E_\infty \to [0,\infty)$;
    \item Study systems on finite subgraphs $G = (V,E)$ with thermal averages $\bra \cdots \ket_G$, e.g., finite boxes $\Lambda_L =\{-L,...,L \}^d \subseteq \dZ^d$;
    \item And take the limit as $G\nearrow G_\infty$ exhausts the infinite graph (e.g., $L\to \infty$), where the existence of the limit also needs to be proved\footnote{Notice that this usually corresponds to free/open boundary conditions.}.
\end{enumerate}
In the case of the $U(1)\times \dZ_2$ Hamiltonians in Eq. \eqref{eq:U1-Z2}, \eqref{eq:U1-Z2-general}, the system is fully ferromagnetic ($\kappa_e,\lambda_e,\alpha_e \ge 0$ for all edges $e$) and thus Ginibre's inequality can be applied \cite{ginibre1970general}, so that correlation functions
\begin{equation}
    \bra \tau_0 \tau_R \ket_G, \quad \bra \cos 2(\theta_0-\theta_R)\ket_G
\end{equation}
Are monotonically increasing as the subgraph $G$ increases, and thus their thermodynamic limits $G\nearrow G_\infty$ are well-defined. 

Note that, in this paper, we shall study the high-order $U(1)$ correlations $\bra \cos 2(\theta_0-\theta_R)\ket$ rather than the conventional correlations $\bra \cos (\theta_0 -\theta_R)\ket$.
Apart from our inability to produce useful inequalities regarding the latter correlations, there is a physically relevant reason towards this perspective\footnote{Note that the Hamiltonian in Eq. \eqref{eq:U1-Z2} and \eqref{eq:U1-Z2-general} is well-defined, independent of its physical motivations, and thus should also warrant the study of the conventional correlations $\bra \cos (\theta_0 -\theta_R)\ket$. Hence, the argument provided should be regarded as heuristics.}.
Indeed, $\theta$ is physically related to the average phase, i.e., $\theta = (\phi^+ +\phi^-)/2$.
However, note that $\cos \theta$ is not invariant under rotation $\phi^+ \mapsto \phi^+ +2\pi$, and thus, in this sense, the $2\pi$-invariant higher order term $\cos (2\theta)$ is the physically relevant term to consider.

\section{Random Cluster Representation}
\label{sec:cluster}

The cluster representation dates back to Fortuin and Kasteleyn \cite{fortuin1972random}, where it was originally developed to study the Ising and more generally, the $q$-state Potts model.
The correspondence between phase transitions in spin models and percolation in the corresponding cluster representation has proven useful in numerous occasions in obtaining rigorous results \cite{duminil2017lectures,duminil2017continuity,duminil2016discontinuity,pfister1997random,aoun2023phase}.
Numerically, the correspondence has allowed the closely related Wolff algorithm to take advantage of non-local cluster spin updates, in contrast to local spin flips as in the Metropolis-Hasting algorithm \cite{binder2022monte}. 
In some sense, the Wolff algorithm provides an alternative (and possibly more familiar) perspective in constructing the cluster representation and thus we will provide a short review using the example of the standard Ising model.

\subsection{Short Review of the Wolff Algorithm}
\label{sec:cluster-Is}
\begin{figure}[ht]
\subfloat[\label{fig:cluster-spin}]{%
  \centering
  \includegraphics[width=0.4\columnwidth]{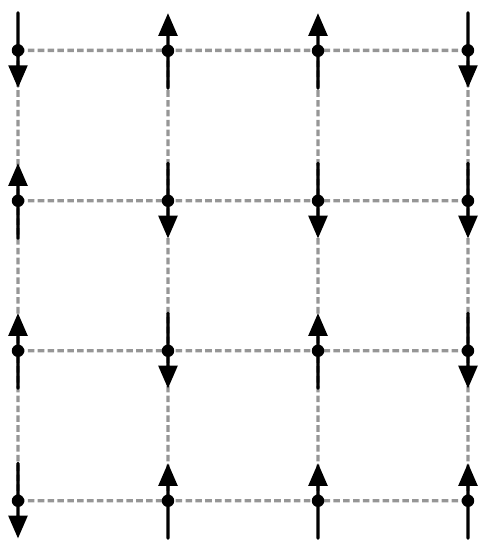}
}
\hspace{0.1\columnwidth}
\subfloat[\label{fig:cluster-spingraph}]{%
  \centering
  \includegraphics[width=0.4\columnwidth]{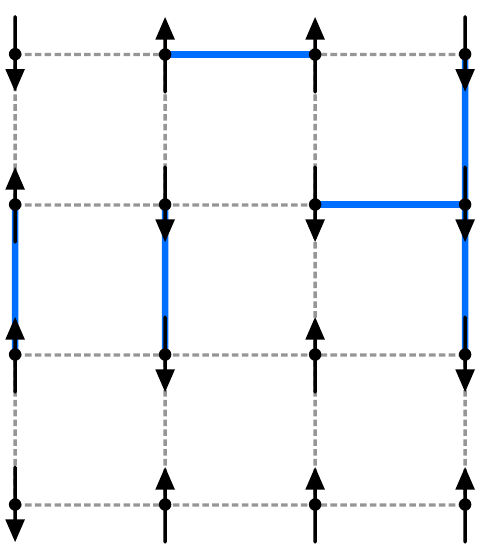}
}
\\
\subfloat[\label{fig:cluster-graph}]{%
  \centering
  \includegraphics[width=0.4\columnwidth]{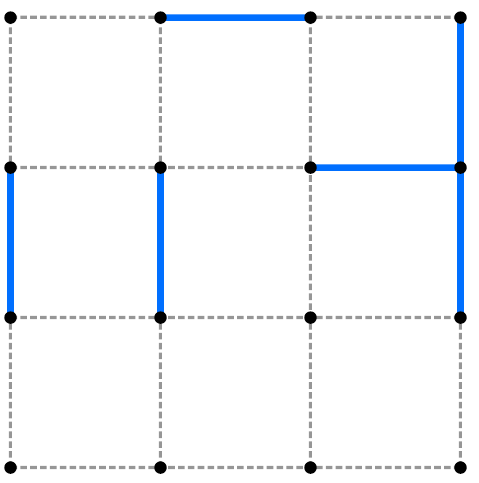}
}
\caption{Cluster Representation on $\dZ^2$. 
(a). Some Ising spin configuration chosen from the Ising Gibbs state. 
(b). A constructed subgraph $\omega$ based on the Ising spin configuration in Fig. \ref{fig:cluster-spin}. 
Note that not all ordered edges are included in $\omega$, since an edge $e$ is included with some probability. However, all disordered edges are not included (based on the edge probability defined in Eq. \eqref{eq:cluster-Is-edge}). (c). The subgraph after integrating over all possible spin configurations.}
\label{fig:cluster}
\end{figure}
The standard Ising model on a finite graph $G$, 
\begin{equation}
    H^\text{Is}(\tau) = -\sum_{e=ij} \tau_i\tau_j
\end{equation}
Defines a probability distribution over Ising spin configurations, i.e.,
\begin{equation}
    \label{eq:cluster-Is-spinprob}
    \dP^\text{Is} [ \tau] \propto \prod_{e=ij} w^\text{Is}_e(\tau), \quad w_e^\text{Is}(\tau) = e^{\beta \tau_i \tau_j}
\end{equation}
Where the product is over all edges with a weight $w_e^\text{Is}$ that depends on the spin configuration $\tau$ implicitly through its values on the endpoints $i,j$ of edge $e$.

Given a fixed spin configuration $\tau$ (see Fig. \ref{fig:cluster-spin}), the Wolff algorithm then generates a subgraph $\omega: E \to \{0,1\}$ in the following manner\footnote{In practice, the Wolff algorithm starts from a random lattice site, and constructs a cluster by adding edges inductively until the cluster stops growing. 
However, the following process is equivalent since each edge is included/excluded independently of other edges (conditionally independent with respect to a fixed spin configuration $\tau$)}. 
Let
\begin{equation}
    \label{eq:cluster-Is-edge}
    p_e^\text{Is} (\tau) = \left[ 1-\frac{w_e^\text{Is}(\tau^i)}{w_e^\text{Is}(\tau)}\right] 1\{\tau_i \tau_j = +1\}
\end{equation}
Where $\tau^i$ is the spin configuration derived from $\tau$ by flipping $\tau_i \mapsto -\tau_i$ (where $i$ is an endpoint of edge $e$) and keeping all other spins the same.
As shown in Fig. \ref{fig:cluster-spingraph}, each edge $e$ is included in the subgraph $\omega$, i.e., $\omega_e = 1$ with probability $p_e^\text{Is}(\tau)$ (otherwise, set $\omega_e =0$).
Note that the procedure for each edge $e$ is independent of each other (conditionally independent with respect to the fixed spin configuration $\tau$) and thus defines a probability distribution\footnote{
This is often referred to as the \textit{Edwards-Sokal} coupling \cite{duminil2017lectures}. 
Note that we are slightly abusing notation since $\dP^\text{Is}$ was originally defined on over spin configurations; now, it is used for $(\omega,\tau)$.
} on subgraphs $\omega$ and spin configurations $\tau$, i.e.,
\begin{align}
    \dP^\text{Is}[\omega|\tau] &= \prod_{e\in \omega} p_e^\text{Is}(\tau) \prod_{e\notin \omega} (1-p_e^\text{Is}(\tau)) \\
    \dP^\text{Is}[(\omega,\tau)] &= \dP^\text{Is}[\omega|\tau] \dP^\text{Is}[\tau]
\end{align}
By integrating over all spin configurations $\tau$, one then obtains the cluster representation, i.e., a probability distribution over subgraphs of the system (see Fig. \ref{fig:cluster-graph})
\begin{equation}
    \dP^\text{Is}[\omega] = \sum_{\tau}\dP^\text{Is}[(\omega,\tau)]
\end{equation}

In language of the Wolff algorithm and Monte Carlo \cite{binder2022monte}, the edge probabilities $p_e^\text{Is}(\tau)$ are chosen to satisfy detailed balance.
In language of probability theory, the edge probabilities are chosen so that the operation of flipping all spins in a cluster keeps the probability distribution invariant \cite{dubedat2022random,duminil2017lectures}, i.e.,

\begin{equation}
    \dP^\text{Is}[(\omega,\tau)]=\dP^\text{Is}[(\omega,\tau^{C(\omega)})]
\end{equation}
Where $C(\omega)$ is some cluster in $\omega$ that has been chosen to be flipped\footnote{To be concrete, one can define $C_v(\omega)$ as the unique cluster in $\omega$ intersecting an arbitrarily chosen site $v\in V$} and $\tau^{C(\omega)}$ is obtained from the original spin configuration $\tau$ via flipping all spins in $C(\omega)$, i.e., $\tau_i \mapsto -\tau_i$ for all $i\in C(\omega)$, and keep all remaining spins the same.

The property of cluster-flip invariance allows us to further prove that the correlation functions are in 1-1 correspondence with percolation events in the cluster representation \cite{duminil2017lectures,dubedat2022random}, i.e.,
\begin{equation}
    \label{eq:cluster-Is}
    \bra \tau_0\tau_R\ket^{\text{Is}}_{G,\beta} = \dP_{G,\beta}^\text{Is}[0\lr_\omega R]
\end{equation}
where $\{0\lr_\omega R\}$ is the event of all subgraphs $\omega$ which connect lattice sites $0, R$, and we have added the subscripts to denote dependence of inverse temperature $\beta$ and graph $G$.
Indeed, the argument is quite straightforward. Note that
\begin{align}
    \label{eq:correspondence-Is-explain-1}
    \bra \tau_0\tau_R\ket^{\text{Is}} = \dE^\text{Is} [\tau_0\tau_R (1\{0\lr_{\omega} R\}+1\{0\not\lr_{\omega} R\})]
\end{align}
If $\{0 \lr_\omega R\}$, then by the definition of the edge probability $p_e^\text{Is}$ (if an edge is included, the edge is ordered), we see that $\tau_0\tau_R =1$.
Conversely, if $\{0 \not\lr_\omega R\}$, then we can flip the spins $\tau$ in the cluster of $\omega$ containing lattice site 0, so that $\tau_0\mapsto -\tau_0$ but $\tau_R \mapsto \tau_R$. Since this operation leaves the probability invariant, we see that
\begin{equation}
    \label{eq:correspondence-Is-explain-2}
    \dE^\text{Is} [\tau_0\tau_R1\{0\not\lr_{\omega} R\}] =-\dE^\text{Is} [\tau_0\tau_R1\{0\not\lr_{\omega} R\}]=0
\end{equation}
The correlation-percolation correspondence is thus established.
\section{A Digression to the generalized XY Model}
\label{sec:digression}

In this section, let us digress from the main model defined in Eq. \eqref{eq:U1-Z2}, \eqref{eq:U1-Z2-general}, and consider the generalized XY model \cite{lee1985strings,korshunov1986phase} on an arbitrary lattice.
As discussed in Ref. \cite{song2021hybrid}, the correlation functions $\bra \cos(\theta_0-\theta_R)\ket,\bra \cos 2(\theta_0-\theta_R)\ket$ determine the phase transitions $T_\text{ferro}, T_\text{nem}$, respectively.
Hence, it is essential to extend the correlation-percolation correspondence discussed in Sec. \eqref{sec:cluster-Is} to the generalized XY model.

Indeed, the cluster representation was generalized to the standard XY model (and more generally to $O(n)$ models) \cite{chayes1998discontinuity} roughly two decades ago. 
By noticing that the sign $\xi$ of the $x$-component $\cos \theta $ of the XY spins can be used as a ``substitute" of the Ising spin in the cluster representation, the author established the correspondence between conventional correlations $\bra \cos(\theta_0 -\theta_R)\ket^\text{XY}$ and percolation (analogous to Eq. \eqref{eq:cluster-Is})\footnote{Since $\bra \cos(\theta_0 -\theta_R)\ket^\text{XY} = 2 \bra \cos \theta_0 \cos \theta_R\ket^\text{XY}$, the spin-spin correlations are schematically similar to $\sim \bra \xi_0 \xi_R \ket^\text{XY}$.}.
However, despite the straightforward generalization of the Wolff algorithm to the standard XY model, it wasn't until recently \cite{dubedat2022random} was there significant progress on establishing a similar correspondence for the higher-order correlation\footnote{
In fact, the authors were only able to extend the correspondence to $k=2$ in $\bra \cos k(\theta_0 -\theta_R) \ket^\text{XY}$. 
We refer to reader to their paper \cite{dubedat2022random} for their reasoning why higher order terms are more difficult. 
Alternatively, we provide the following argument. Note that $\bra \cos 2(\theta_0 -\theta_R) \ket^\text{XY} = 8\bra \sin \theta_0 \cos \theta_0 \sin \theta_R \cos \theta_R \ket^\text{XY} \sim \bra \xi_0 \eta_0 \xi_R \eta_R\ket^\text{XY}$ where $\xi,\eta=\pm 1$ are the signs of the $x,y$ components $\cos \theta, \sin \theta$ of the XY spins and can be treated as independent Ising spins. 
For higher order terms, there are not enough independent Ising spins that can derived from the original XY spins.
} $\bra \cos 2(\theta_0 -\theta_R) \ket^\text{XY}$.

The philosophy and techniques developed for the standard Ising and XY models can thus be straightforwardly extended to the generalized XY model.
Following Ref. \cite{dubedat2022random}, we can prove the following
\begin{theorem}
    \label{thm:general-XY}
    Consider the following generalized XY model on any finite graph $G$ with $4/5 \le \Delta_e < 1$ for all edges $e$
    \begin{equation}
        H(\theta) =-\sum_{e=ij} \left[ \Delta_e \cos (\nabla_e \theta) +(1-\Delta_e ) \cos (2\nabla_e \theta)\right] 
    \end{equation}
    Then for any temperature $\beta$ and lattice sites $0,R$ in $G$, there exists a constant $C>0$ depending only on $\beta$ and the number of edges adjacent to lattice sites $0,R$, respectively, such that
    \begin{equation}
        \bra \cos 2(\theta_0 -\theta_R) \ket_{G,\beta} \le C \bra \cos (\theta_0 -\theta_R) \ket_{G,\beta} 
    \end{equation}
\end{theorem}

We note that if the standard correlations $\bra \cos (\theta_0-\theta_R)\ket$ spins are disordered (exponentially decaying with respect to $R$) at some temperature $T$, then so must the higher order correlations $\bra \cos 2(\theta_0-\theta_R)\ket$.
Hence, Theorem \eqref{thm:general-XY} implies that $T_\text{ferro} \ge T_\text{nem}$ provided that $\Delta_e \ge 4/5$ for all edges $e$. 
\begin{proof}[Sketch of Proof]
    The Hamiltonian defines a probability distribution over spin configurations, i.e.,
    \begin{align}
        \dP[\theta] &\propto \prod_{e=ij} w_e(\theta) \\
        w_e(\theta) &\equiv  \rho_e(\cos\nabla_e \theta)\\
        &=\exp \left[ 2\beta (1-\Delta_e) \left[ \cos \nabla_e \theta +\frac{\Delta_e}{4(1-\Delta_e)}\right]^2\right] \nonumber
    \end{align}
    Since $4/5\le \Delta_e <1$, the map $x\mapsto \rho_e (x)$ is convex and strictly increasing, and thus satisfies the sufficient conditions \cite{dubedat2022random} to establish the correlation-percolation correspondence with respect to the higher order correlations $\bra \cos 2(\theta_0 -\theta_R)\ket$.
    In particular, we can define a probability distribution $\dP$ over subgraphs $\hat{\xi},\hat{\eta}:E\to \{0,1\}$ such that
    \begin{align}
        \bra \cos 2(\theta_0 -\theta_R) \ket &\le 2 \dP[0\lr_{\hat{\xi}\hat{\eta}} R] \\
        C \dP[0\lr_{\hat{\xi}} R] &\le \bra \cos (\theta_0-\theta_R)\ket
    \end{align}
    Where $C$ depends only on $\beta$ and the number of edges adjacent to lattice sites $0,R$, respectively, and $\hat{\xi}\hat{\eta}$ is the intersection of the graphs $\hat{\xi},\hat{\eta}$ \cite{dubedat2022random}. 
    Since the event $\{0\lr_{\hat{\xi}\hat{\eta}} R\}$ is included in $\{0\lr_{\hat{\xi}} R\}$, it's clear that
    \begin{equation}
        \dP[0\lr_{\hat{\xi}\hat{\eta}} R] \le \dP[0\lr_{\hat{\xi}} R]
    \end{equation}
    And thus the statement follows.
\end{proof}
\section{Critical Regime $\lambda,\alpha = 0$}
\label{sec:critical}

\subsection{Correlation Inequality}

The philosophy and techniques developed for the standard Ising and XY models can thus be straightforwardly extended to the critical Hamiltonian in Eq. \eqref{eq:U1-Z2} (see Appendix \eqref{app:cluster}).
The only difficulty lies in finding a relation (inequality) between the percolation events so that
\begin{theorem}[see Appendix \eqref{app:cluster}]
    \label{thm:cluster}
    Let the critical Hamiltonian $H$ in Eq. \eqref{eq:U1-Z2} be defined on a finite graph $G$. Then for any temperature
    \begin{equation}
        \bra \cos 2(\theta_0-\theta_R) \ket_{G,\beta} \le 2 \bra \tau_0\tau_R\ket_{G,\beta}
    \end{equation}
\end{theorem}

It should be noted that for the subclass of Hamiltonians in Eq. \eqref{eq:U1-Z2}, the correlation inequality in Theorem \eqref{thm:cluster} can be improved slightly so that the factor of 2 can be removed, i.e.,
\begin{equation}
    \bra \cos 2(\theta_0-\theta_R) \ket_{G,\beta} \le \bra \tau_0\tau_R\ket_{G,\beta}
\end{equation}
The proof is independent of the cluster representation and relies on a distinct representation called the (random) current representation\footnote{
As discussed in the introduction, the current representation provides important insight into why only a single phase transition is observed at the critical regime $\ell_\infty$. 
See Appendix \eqref{app:current}.
}.
In any case,  we see that if the $\dZ_2$ Ising spins are disordered (exponentially decaying with respect to $R$) at some temperature $T$, then so must the $U(1)$ XY spins, and thus there cannot be a floating phase, i.e.,
\begin{equation}
    T_\text{TRSB} \ge T_c
\end{equation}
\begin{proof}[Sketch of Proof]
To be concrete, the Hamiltonian in Eq. \eqref{eq:U1-Z2} defines a probability distribution over spin configurations, i.e.,
\begin{align}
    \dP[(\sigma,\tau)] &\propto \prod_{e=ij} w_e(\sigma,\tau) \\
    w_e(\sigma,\tau) &= e^{\beta \kappa_e (1+\tau_i\tau_j) \cos(\nabla_e \theta)}
\end{align}

Given a fixed spin configuration $\sigma,\tau$, we can generate a subgraph $\hat{\tau}: E\to \{0,1\}$ using the edge probabilities (see also explicit form in Eq. \eqref{eq:cluster-edge-tau-explicit})
\begin{align}
    p_e^\tau (\sigma,\tau) = \left[ 1-\frac{w_e(\sigma,\tau^i)}{w_e(\sigma,\tau)}\right] 1\{w_e(\sigma,\tau^i) < w_e(\sigma,\tau)\}
\end{align}
We then define a probability distribution $\dP[(\hat{\tau},\sigma,\tau)]$ as discussed in Sec. \eqref{sec:cluster}, and establish the correspondence between correlations $\bra \tau_0 \tau_R\ket$ and percolation $\dP [0\lr_{\hat{\tau}} R]$, analogous to Eq. \eqref{eq:cluster-Is}.

Similarly, we can generate a subgraph $\hat{\sigma}$ using the edge probabilities \cite{dubedat2022random} (see also explicit form in Eq. \eqref{eq:cluster-edge-sigma-explicit})
\begin{align}
    \label{eq:cluster-edge-sigma}
    p_e^{\sigma}(\sigma,\tau)&=\left[1 +\frac{w_e(\sigma^{\xi\eta,i},\tau)}{w_e(\sigma,\tau)}\right.\\
    &\quad\left.- \frac{w_e(\sigma^{\xi,i},\tau)}{w_e(\sigma,\tau)}-\frac{w_e(\sigma^{\eta,i},\tau)}{w_e(\sigma,\tau)} \right]\nonumber  \\
    &\quad\times  1\{w_e(\sigma^{\xi,i},\tau), w_e(\sigma^{\eta,i},\tau) < w_e(\sigma,\tau) \}    \nonumber
\end{align}
Where $\xi,\eta =\pm 1$ denote the signs of the $x,y$ components of the XY spin $\sigma=e^{i\theta}$, and $\sigma^{\xi,i}$ denotes the spin configuration derived from $\sigma$ by flipping the spin at site $i$ along the $x$-component, i.e., $\xi_i \mapsto -\xi_i$ or $\theta_i \mapsto \pi -\theta_i$, and keeping all other sites the same. 
We similarly define $\sigma^{\eta,i}$ (via $\theta_i\mapsto -\theta_i$) and $\sigma^{\xi\eta,i}$ (via $\theta_i \mapsto \theta_i +\pi$). 
We then define a probability distribution $\dP[(\hat{\sigma},\sigma,\tau)]$ as before, and establish the correspondence between correlations $\bra \cos 2(\theta_0 -\theta_R)\ket$ and percolation $\dP [0\lr_{\hat{\tau}} R]$.

Note that in our construction, we extended the probability distribution $\dP$ over spin configurations to either that of $(\hat{\tau},\sigma,\tau)$ or $(\hat{\sigma},\sigma,\tau)$. 
However, if we wish to compare the percolation events $\dP[0\lr_{\hat{\tau}} R]$ and $\dP[0\lr_{\hat{\sigma}} R]$, it is necessary to extend the probability distribution $\dP$ to that of $(\hat{\sigma},\hat{\tau}, \sigma, \tau)$, or simply that of $(\hat{\sigma},\hat{\tau})$ after integrating over all spin $(\sigma,\tau)$ configurations\footnote
{More specifically, if we did not abuse notation and use $\dP$ for simplicity, the correlation $\bra \tau_0\tau_R\ket$ would correspond to $\dP^{\tau} [0\lr_{\hat{\tau}} R]$ where $\dP^\tau$ is the probability defined over $(\hat{\tau},\sigma,\tau)$. 
Similarly, the correlation $\bra \cos 2(\theta_0 -\theta_R)\ket$ would correspond to $\dP^\sigma [0\lr_{\hat{\sigma}}R]$ where $\dP^\sigma$ is the probability defined over $(\hat{\sigma},\sigma,\tau)$. 
The relation between $\dP^\tau, \dP^\sigma$ is not yet clear.
}.

One way that has turned out to be useful is to consider the conditional probability with respect to fixed spin configuration $\sigma, \tau$.
Since each edge $e$ is constructed independently, we shall consider a fixed edge $e=ij$ so that $\hat{\tau}_e,\hat{\sigma}_e =0,1$ correspond to Bernoulli random variables. 
To extend the conditional probability to both $(\hat{\sigma},\hat{\tau})$, we must define a correlation $c_e(\sigma,\tau)$ between the two variables so that the following probability is well-defined, i.e., the following probabilities are all $\ge 0$
\begin{align}
    \label{eq:cluster-corr-1}
    \dP[\hat{\sigma}_e=1,\hat{\tau}_e =1|\sigma,\tau] &= c_e\\
    \label{eq:cluster-corr-2}
    \dP[\hat{\sigma}_e=0,\hat{\tau}_e =1|\sigma,\tau] &= p_e^\tau - c_e\\
    \label{eq:cluster-corr-3}
    \dP[\hat{\sigma}_e=1,\hat{\tau}_e =0|\sigma,\tau] &= p_e^\sigma-c_e\\
    \label{eq:cluster-corr-4}
    \dP[\hat{\sigma}_e=0,\hat{\tau}_e =0|\sigma,\tau] &= 1-p_e^\tau -p_e^\sigma +c_e
\end{align}
If such a condition is satisfied by choosing the correlation $c_e$ appropriately, then by integrating over say, $\hat{\tau}_e = 0,1$, the conditional probability of $\hat{\sigma}_e$ (with respect to the fixed spin configuration $(\sigma,\tau)$) will be exactly what was required, i.e., $p_e^\sigma$.

Indeed, the key observation is to notice that\footnote{This is the main reason why we chose to use the higher order correlation $\bra \cos 2(\theta_0 -\theta_R) \ket$, since the edge probability corresponding to the conventional correlation $\bra \cos (\theta_0 -\theta_R)\ket$ has no simple relation with $p_e^\tau$.} $p_e^\tau \ge p_e^\sigma$ regardless of the spin $(\sigma,\tau)$ configuration (see Theorem \eqref{app-thm:cluster-relation} in Appendix \eqref{app:cluster}) and thus we can choose the correlation $c_e \equiv p_e^\sigma$ so that the previous conditions in Eq. \eqref{eq:cluster-corr-1}-\eqref{eq:cluster-corr-4} are satisfied.
In particular, Eq. \eqref{eq:cluster-corr-3} is always zero and thus within this setup, we have defined a probability distribution $\dP$ over $(\hat{\sigma},\hat{\tau},\sigma,\tau)$ such that $\hat{\sigma}$ is always (with probability $=1$) a subgraph of $\hat{\tau}$.
Therefore, if lattice sites $0,R$ are connected within $\hat{\sigma}$, it must also be connected within $\hat{\tau}$, i.e.,
\begin{equation}
    \dP[0\lr_{\hat{\sigma}} R] \le \dP[0\lr_{\hat{\tau}} R]
\end{equation}
The statement then follows (the extra factor of $2$ is due to the fact that the correlation $\bra \cos 2(\theta_0-\theta_R)\ket$ is not strictly equal to $\dP[0\lr_{\hat{\sigma}} R]$. See Appendix \eqref{app:cluster} for details.)
\end{proof}

\subsection{Order of Transition}
\label{sec:impl-order}
Another interesting question is the nature (first order or higher) of the transition, provided that the conjectured single phase transition exists, i.e., $T_\text{TRSB} = T_c$.
We shall provide the discussion for the critical regime $\ell_\infty$, though the argument extends straightforwardly to the boundary $\partial \sO_\infty$.
Indeed, from mean-field theory \cite{can2021high,yuan2023inhomogeneity}, we know that the transition is of second order for all $\ell:J_2>0$ (and not just the strong coupling limit $\ell_\infty:J_2 \to \infty$) and thus should hold true for $d\ge 4$ dimensions in $\dZ^d$.
In Ref. \cite{bojesen2014phase}, the authors numerically found that the transition is first order in $d=3$ dimensions (regardless of the coupling strength), and thus brings into question whether there is a change in behavior between $d=3,4$ dimensions.
Although we do not have a definitive answer for this question, the following discussion may prove insightful.

Consider $d=2$ dimensions.
For the standard XY model, it is usually argued that the transition is continuous (higher than first order) due to the absence of spontaneous symmetry breaking by Mermin-Wagner\footnote{More rigorously, the transition for the standard XY model can be proven to be continuous in any dimension $d$ by using the Lieb-Simon-Rivasseau inequality \cite{simon1980correlation,lieb1980refinement,rivasseau1980lieb} (see also Ref. \cite{bauerschmidt2016ferromagnetic} using the Lebowitz inequality).}.
For the critical Hamiltonian with constant $\kappa_e=1$, recent 2D numerical results suggest that the $U(1)$ correlations undergo a BKT phase transition \cite{song2022phase}, and thus one would can argue that $U(1)$ correlation length diverges as $T \searrow T_c$. 
If so, the inequalities in Theorem \eqref{thm:cluster} would imply that the $\dZ_2$ correlation length also diverges (again assuming that $T_c = T_\text{TRSB}$).
Hence, the conjectured single phase transition is presumably continuous in $d=2$ dimensions.
With that in mind, the transition is continuous in $d=2$ and $d\ge 4$ dimension, what reason could cause the transition to become discontinuous in $d=3$ dimensions?


Admittedly, the previous argument is not definitive, and one may attempt to circumvent it.
For example, Mermin-Wagner by itself does not imply that the transition is BKT or continuous.
Indeed, in Ref. \cite{van2002first,van2005first,van2006first}, the authors constructed a counter example in which the system has a continuous symmetry (thus obeying Mermin-Wagner) and yet exhibited long-range order in $d=2$ dimension, i.e., even though the magnetization is $=0$, the spin-spin correlation functions do not decay to zero.
More specifically, they proved that the counterexample exhibited a first order phase transition in $d=2$.
However, we argue the above counter example does not apply to our system: 
\begin{enumerate}[(a)]
    \item The constructed example is quite unphysical since they require to take a parameter $p \to \infty$, in which the phase transition changes from being higher order at $p=1$ to first order as $p\to \infty$. 
    In comparison, although our model corresponds to the strong coupling limit $J_2 \to \infty$, it has been consistently shown (in 2D \cite{song2022phase} and 3D \cite{maccari2022effects,bojesen2014phase}) that the transition is qualitatively independent of the coupling strength.
    \item 2D numerics \cite{song2022phase} suggest that at arbitrary finite coupling $J_2 >0$ (though they did not go to the strong coupling limit), the $U(1)$ spin-spin correlations decay algebraically below the critical temperature and thus exhibit a BKT transition. 
    This is in contrast to the counter example constructed in Ref. \cite{van2002first,van2005first,van2006first}.
\end{enumerate}
\section{General Regime $\lambda \ge \alpha \ge  0$}
\label{sec:general}

\begin{figure}[ht]
\includegraphics[width=0.8\columnwidth]{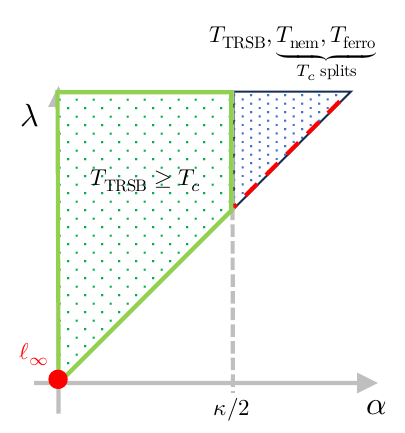}
\caption{Absence of Floating Phase. The $x,y$ axes denote the $\alpha, \lambda$ coupling described in Eq. \eqref{eq:U1-Z2-general} (where we have omitted the edge subscript).
The red dot denotes the critical Hamiltonian $\lambda=\alpha=0$ defined in Eq. \eqref{eq:U1-Z2} while the diagonal lines describes $\lambda = \alpha$.
Within the green region, there exists no floating phase on any lattice structure. 
Within the blue region, if one compares Eq. \eqref{eq:U1-Z2-general} with the generalized XY model in Eq. \eqref{eq:general-XY}, the SC transition $T_c$ may further split into a nematic $T_\text{nem}$ and $T_\text{ferro}$ transition and thus the order of transitions among $T_\text{TRSB},T_\text{nem},T_\text{ferro}$ is unclear.
}
\label{fig:nofloat}
\end{figure}

Using similar techniques as discussed previously in Sec. \eqref{sec:critical}, we can prove the following.
\begin{theorem}[see Appendix \eqref{app:cluster-general}]
    \label{thm:cluster-general}
    Consider the general Hamiltonian $H_{\lambda,\alpha}$ in Eq. \eqref{eq:U1-Z2-general} on a finite graph $G$, such that $\lambda_e \ge \alpha_e \ge 0$ and $\alpha_e/\kappa_e \le 1/2$ for all edges $e$. Then for any temperature
    \begin{equation}
        \bra \cos 2(\theta_0-\theta_R) \ket_{G,\beta} \le 2 \bra \tau_0\tau_R\ket_{G,\beta}
    \end{equation}
\end{theorem}

Before sketching the proof, let us discuss the implications. 
By Theorem \eqref{thm:cluster-general}, we see that if $\lambda_e \ge \alpha_e \ge 0$ and $\alpha_e \le \kappa_e/2$ for all edges $e$ on the lattice, then the system cannot have a floating phase.
As shown in Fig. \ref{fig:nofloat}, the origin $\lambda=\alpha=0$ corresponds to the critical regime $\ell_\infty$ \cite{diag} in Eq. \eqref{eq:U1-Z2}.
Since the current-current coupling (see Fig. \ref{fig:phase-diag}) is characterized by $\nu =\lambda-\alpha$ (up to quadratic orders as argued in Sec. \eqref{sec:model}), the region $\lambda \ge \alpha$ corresponds to the general parameter regime $\bar{\sO}_\infty$ in Fig. \ref{fig:phase-diag}.

It should be noted that the restriction $\alpha_e \le \kappa_e/2$ is not merely technical. 
Indeed, if we compare the considered Hamiltonian Eq. \eqref{eq:U1-Z2-general} with the generalized XY model in Eq. \eqref{eq:general-XY}, we that if $\alpha_e > \kappa_e/2$, the system becomes metastable at $\nabla_e \theta = \pi$ (in the sector $\tau_i\tau_j = +1$) and thus the SC transition $T_c$ can possibly split into a nematic $T_\text{nem}$ and a ferromagnetic $T_\text{ferro}$ transition, each characterized by the correlation functions $\bra \cos (\theta_0-\theta_R)\ket, \bra \cos 2(\theta_0-\theta_R)\ket$, respectively.
The order between the 3 possible transitions $T_\text{TRSB},T_\text{nem},T_\text{ferro}$ is then unclear.
However, in practice, the extra restriction is unimportant, since the in-plane superfluid stiffness $\kappa$ is usually much larger than any inter-component coupling $\lambda,\alpha$. 
Therefore, for a large class of Hamiltonians, high-$T_c$ TRSB is guaranteed.

\begin{proof}[Sketch of Proof]
The Hamiltonian in Eq. \eqref{eq:U1-Z2-general} defines a probability distribution over spin configurations, i.e.,
\begin{align}
    \dP[(\sigma,\tau)] &\propto \prod_{e=ij} w_e(\sigma,\tau) \\
    w_e(\sigma,\tau) &= e^{\beta [\kappa_e (1+\tau_i\tau_j) \cos(\nabla_e\theta)+\lambda_e \tau_i\tau_j +\alpha_e \cos (2\nabla_e \theta)]}
\end{align}

Given a fixed spin configuration $\sigma,\tau$, we can generate a subgraph $\hat{\tau}: E\to \{0,1\}$ using the edge probabilities (see also explicit form in Eq. \eqref{eq:cluster-edge-tau-explicit-general})
\begin{align}
    p_e^\tau (\sigma,\tau) = \left[ 1-\frac{w_e(\sigma,\tau^i)}{w_e(\sigma,\tau)}\right] 1\{w_e(\sigma,\tau^i) < w_e(\sigma,\tau)\}
\end{align}
We can then define a probability distribution $\dP[(\hat{\tau},\sigma,\tau)]$ as previously discussed in Sec. \eqref{sec:cluster}, and establish the correspondence between correlations $\bra \tau_0 \tau_R\ket$ and percolation $\dP [0\lr_{\hat{\tau}} R]$, analogous to Eq. \eqref{eq:cluster-Is}.

Similarly, we can generate a subgraph $\hat{\sigma}$ using the edge probabilities \cite{dubedat2022random} (see also explicit form in Eq. \eqref{eq:cluster-edge-sigma-explicit-general})
\begin{align}
    p_e^{\sigma}(\sigma,\tau)&=\left[1 +\frac{w_e(\sigma^{\xi\eta,i},\tau)}{w_e(\sigma,\tau)}\right.\\
    &\quad\left.- \frac{w_e(\sigma^{\xi,i},\tau)}{w_e(\sigma,\tau)}-\frac{w_e(\sigma^{\eta,i},\tau)}{w_e(\sigma,\tau)} \right]\nonumber  \\
    &\quad\times  1\{w_e(\sigma^{\xi,i},\tau), w_e(\sigma^{\eta,i},\tau) < w_e(\sigma,\tau) \} \nonumber\\
    &\quad\times 1\{\tau_i \tau_j = +1\} \nonumber
\end{align}
Note that compared to Eq. \eqref{eq:cluster-edge-sigma}, we require an extra condition $\tau_i\tau_j=+1$.
The reasoning is of 2-fold: (1). Without the extra condition, it becomes possible for $\tau_i \tau_j=-1$ so that $p_e^\sigma <0$ and thus inducing a sign problem (probabilities cannot be negative). (2). To establish the correspondence between correlations $\bra \cos 2(\theta_0 -\theta_R)\ket$ and percolation $\dP [0\lr_{\hat{\tau}} R]$, we only require reflecting the XY spins (analogous to the argument in Eq. \eqref{eq:correspondence-Is-explain-1}, \eqref{eq:correspondence-Is-explain-2}) and thus we are free to choose a sector of $\tau_i \tau_j =\pm 1$ (See Appendix \eqref{app:cluster-U1-general}).

Note that once we choose the section $\tau_i\tau_j=+1$, the edge probability reduces to
\begin{align}
    p_e^{\sigma}(\sigma,\tau)&=\left[1 +\frac{\tilde{w}_e(\sigma^{\xi\eta,i})}{\tilde{w}_e(\sigma)} \right.\\
    &\quad\left. -\frac{\tilde{w}_e(\sigma^{\xi,i})}{\tilde{w}_e(\sigma)}-\frac{\tilde{w}_e(\sigma^{\eta,i})}{\tilde{w}_e(\sigma)} \right] \nonumber\\
    &\quad\times  1\{\tilde{w}_e(\sigma^{\xi,i}), \tilde{w}_e(\sigma^{\eta,i}) < \tilde{w}_e(\sigma), \tau_i\tau_j= +1 \} \nonumber
\end{align}
Where the effective edge weight $\tilde{w}_e$ now satisfy the convexity and monotonically increasing condition described in Ref. \cite{dubedat2022random} (also see Appendix \eqref{app:cluster-U1-general}) provided that $\alpha_e/\kappa_e \le 1/2$.
As in the proof of Theorem \eqref{thm:cluster}, we note that $p_e^\tau \ge p_e^\sigma$ for all possible spin $(\sigma,\tau)$ configurations (see Theorem \eqref{app-thm:cluster-relation-general} in Appendix \eqref{app:cluster-general}) and thus the statement follows.

\end{proof}

\section{Summary and Discussion}

As discussed in the main text, we have rigorously proven that the class of $U(1)\times \dZ_2$ Hamiltonian in Eq. \eqref{eq:U1-Z2-general} does not exhibit a floating phase on any lattice structure, i.e., $T_\text{TRSB} \ge T_c$ (which includes the boundary cases where the transition temperature is possibly $=0,\infty$).
The model is physically motivated by the either twisted bilayer systems with dominant 2\ts{nd} order inter-layer Josephson coupling $J_2$ or $(n\ge 3)$-band superconductors with frustrated 1\ts{st} order inter-band Josephson coupling.
In fact, it corresponds to the strong coupling regime ($J_2 \to \infty$), which we have shown to be continuously connected to finite but large $J_2$ values \cite{interchangeLimits}.
From numerical simulations in 2D \cite{song2022phase,liu2023charge}, 3D \cite{maccari2022effects,bojesen2014phase} and exact mean-field understanding (for $d\ge 4$) \cite{yuan2023exactly,can2021high}, it is believed that the qualitative properties of the systems are insensitive to the coupling strength and thus our result on the strong coupling limit sheds light onto the phase diagram (see Fig. \ref{fig:phase-diag}).

The technique developed regarding the cluster representation in this paper may also be used to quickly check if other classes of Hamiltonians possess the same property, i.e., $T_\text{TRSB} \ge T_c$. 
Indeed, the proof (ignoring technical details) ultimately relies on the relation between the edge weights corresponding to the TRSB $p^\text{Is}_e$ and SC transition $p^\text{XY}_e$, i.e., $p^\text{Is}_e \ge p^\text{XY}_e$, as constructed via the Wolff algorithm.
Therefore, we believe that our proof can be useful when consider other types of interactions in multi-component systems.
It's also worth mentioning that the cluster representation can be easily extended to finite coupling $J_2 <\infty$ so that a correlation-percolation corresponce can be established (analogous to Eq. \eqref{eq:cluster-Is}).
However, the difficulty there is that the relation between the percolation events corresponding to the SC and TRSB transition is less clear, i.e., the analogous $p^\tau_e \ge p_e^\sigma$ used in Theorem \eqref{thm:cluster} and \eqref{thm:cluster-general} is no longer true \cite{finitecoupling}.
Hence, further developments are necessary to determine the order of transitions.

As a tangent, we have also shown that generalized XY model \cite{lee1985strings,korshunov1986phase}, though falling into an entirely distinct class of $U(1)\times \dZ_2$ Hamiltonian, possesses a similar property, i.e., $T_\text{ferro} \ge T_\text{nem}$, on any lattice structure provided that $\Delta \ge 4/5$. 
In this regime, it is believed that there is only a single phase transition $T_\text{ferro} = T_\text{nem}$ as least for the 2D square lattice \cite{song2021hybrid,carpenter1989phase}. 
This is expected since the Hamiltonian in Eq. \eqref{eq:general-XY} does not develop a metastable (local min) state at $\nabla_e \theta =\pi$ unless $\Delta <4/5$.

\section{Acknowledegements}

I am grateful for Steve A. Kivelson's support and generosity during this project and also for providing extensive comments and suggestions on the draft.  This work was supported, in part, by NSF Grant No. DMR-2000987 at Stanford University. 
\bibliography{main.bbl}

\begin{thebibliography}{64}%
\makeatletter
\providecommand \@ifxundefined [1]{%
 \@ifx{#1\undefined}
}%
\providecommand \@ifnum [1]{%
 \ifnum #1\expandafter \@firstoftwo
 \else \expandafter \@secondoftwo
 \fi
}%
\providecommand \@ifx [1]{%
 \ifx #1\expandafter \@firstoftwo
 \else \expandafter \@secondoftwo
 \fi
}%
\providecommand \natexlab [1]{#1}%
\providecommand \enquote  [1]{``#1''}%
\providecommand \bibnamefont  [1]{#1}%
\providecommand \bibfnamefont [1]{#1}%
\providecommand \citenamefont [1]{#1}%
\providecommand \href@noop [0]{\@secondoftwo}%
\providecommand \href [0]{\begingroup \@sanitize@url \@href}%
\providecommand \@href[1]{\@@startlink{#1}\@@href}%
\providecommand \@@href[1]{\endgroup#1\@@endlink}%
\providecommand \@sanitize@url [0]{\catcode `\\12\catcode `\$12\catcode
  `\&12\catcode `\#12\catcode `\^12\catcode `\_12\catcode `\%12\relax}%
\providecommand \@@startlink[1]{}%
\providecommand \@@endlink[0]{}%
\providecommand \url  [0]{\begingroup\@sanitize@url \@url }%
\providecommand \@url [1]{\endgroup\@href {#1}{\urlprefix }}%
\providecommand \urlprefix  [0]{URL }%
\providecommand \Eprint [0]{\href }%
\providecommand \doibase [0]{https://doi.org/}%
\providecommand \selectlanguage [0]{\@gobble}%
\providecommand \bibinfo  [0]{\@secondoftwo}%
\providecommand \bibfield  [0]{\@secondoftwo}%
\providecommand \translation [1]{[#1]}%
\providecommand \BibitemOpen [0]{}%
\providecommand \bibitemStop [0]{}%
\providecommand \bibitemNoStop [0]{.\EOS\space}%
\providecommand \EOS [0]{\spacefactor3000\relax}%
\providecommand \BibitemShut  [1]{\csname bibitem#1\endcsname}%
\let\auto@bib@innerbib\@empty
\bibitem [{\citenamefont {Wang}\ and\ \citenamefont
  {Fu}(2017)}]{wang2017topological}%
  \BibitemOpen
  \bibfield  {author} {\bibinfo {author} {\bibfnamefont {Y.}~\bibnamefont
  {Wang}}\ and\ \bibinfo {author} {\bibfnamefont {L.}~\bibnamefont {Fu}},\
  }\bibfield  {title} {\bibinfo {title} {Topological phase transitions in
  multicomponent superconductors},\ }\href@noop {} {\bibfield  {journal}
  {\bibinfo  {journal} {Physical review letters}\ }\textbf {\bibinfo {volume}
  {119}},\ \bibinfo {pages} {187003} (\bibinfo {year} {2017})}\BibitemShut
  {NoStop}%
\bibitem [{\citenamefont {Kivelson}\ \emph {et~al.}(2020)\citenamefont
  {Kivelson}, \citenamefont {Yuan}, \citenamefont {Ramshaw},\ and\
  \citenamefont {Thomale}}]{kivelson2020proposal}%
  \BibitemOpen
  \bibfield  {author} {\bibinfo {author} {\bibfnamefont {S.~A.}\ \bibnamefont
  {Kivelson}}, \bibinfo {author} {\bibfnamefont {A.~C.}\ \bibnamefont {Yuan}},
  \bibinfo {author} {\bibfnamefont {B.}~\bibnamefont {Ramshaw}},\ and\ \bibinfo
  {author} {\bibfnamefont {R.}~\bibnamefont {Thomale}},\ }\bibfield  {title}
  {\bibinfo {title} {{A proposal for reconciling diverse experiments on the
  superconducting state in Sr$_2$RuO$_4$}},\ }\href@noop {} {\bibfield
  {journal} {\bibinfo  {journal} {npj Quantum Materials}\ }\textbf {\bibinfo
  {volume} {5}},\ \bibinfo {pages} {43} (\bibinfo {year} {2020})}\BibitemShut
  {NoStop}%
\bibitem [{\citenamefont {Yuan}\ \emph {et~al.}(2021)\citenamefont {Yuan},
  \citenamefont {Berg},\ and\ \citenamefont {Kivelson}}]{yuan2021strain}%
  \BibitemOpen
  \bibfield  {author} {\bibinfo {author} {\bibfnamefont {A.~C.}\ \bibnamefont
  {Yuan}}, \bibinfo {author} {\bibfnamefont {E.}~\bibnamefont {Berg}},\ and\
  \bibinfo {author} {\bibfnamefont {S.~A.}\ \bibnamefont {Kivelson}},\
  }\bibfield  {title} {\bibinfo {title} {{Strain-induced time reversal breaking
  and half quantum vortices near a putative superconducting tetracritical point
  in Sr$_2$RuO$_4$}},\ }\href@noop {} {\bibfield  {journal} {\bibinfo
  {journal} {Physical Review B}\ }\textbf {\bibinfo {volume} {104}},\ \bibinfo
  {pages} {054518} (\bibinfo {year} {2021})}\BibitemShut {NoStop}%
\bibitem [{\citenamefont {Yuan}\ \emph
  {et~al.}(2023{\natexlab{a}})\citenamefont {Yuan}, \citenamefont {Berg},\ and\
  \citenamefont {Kivelson}}]{yuan2023multiband}%
  \BibitemOpen
  \bibfield  {author} {\bibinfo {author} {\bibfnamefont {A.~C.}\ \bibnamefont
  {Yuan}}, \bibinfo {author} {\bibfnamefont {E.}~\bibnamefont {Berg}},\ and\
  \bibinfo {author} {\bibfnamefont {S.~A.}\ \bibnamefont {Kivelson}},\
  }\bibfield  {title} {\bibinfo {title} {{Multiband mean-field theory of the
  $d+ i g$ superconductivity scenario in Sr$_2$RuO$_4$}},\ }\href@noop {}
  {\bibfield  {journal} {\bibinfo  {journal} {Physical Review B}\ }\textbf
  {\bibinfo {volume} {108}},\ \bibinfo {pages} {014502} (\bibinfo {year}
  {2023}{\natexlab{a}})}\BibitemShut {NoStop}%
\bibitem [{\citenamefont {Laughlin}(1998)}]{laughlin1998magnetic}%
  \BibitemOpen
  \bibfield  {author} {\bibinfo {author} {\bibfnamefont {R.}~\bibnamefont
  {Laughlin}},\ }\bibfield  {title} {\bibinfo {title} {{Magnetic induction of
  $d_{x^2- y^2}+ id_{xy}$ order in high-T$_c$ superconductors}},\ }\href@noop
  {} {\bibfield  {journal} {\bibinfo  {journal} {Physical review letters}\
  }\textbf {\bibinfo {volume} {80}},\ \bibinfo {pages} {5188} (\bibinfo {year}
  {1998})}\BibitemShut {NoStop}%
\bibitem [{\citenamefont {Yuan}\ \emph
  {et~al.}(2023{\natexlab{b}})\citenamefont {Yuan}, \citenamefont {Vituri},
  \citenamefont {Berg}, \citenamefont {Spivak},\ and\ \citenamefont
  {Kivelson}}]{yuan2023inhomogeneity}%
  \BibitemOpen
  \bibfield  {author} {\bibinfo {author} {\bibfnamefont {A.~C.}\ \bibnamefont
  {Yuan}}, \bibinfo {author} {\bibfnamefont {Y.}~\bibnamefont {Vituri}},
  \bibinfo {author} {\bibfnamefont {E.}~\bibnamefont {Berg}}, \bibinfo {author}
  {\bibfnamefont {B.}~\bibnamefont {Spivak}},\ and\ \bibinfo {author}
  {\bibfnamefont {S.~A.}\ \bibnamefont {Kivelson}},\ }\bibfield  {title}
  {\bibinfo {title} {Inhomogeneity-induced time-reversal symmetry breaking in
  cuprate twist-junctions},\ }\href@noop {} {\bibfield  {journal} {\bibinfo
  {journal} {arXiv preprint arXiv:2305.15472}\ } (\bibinfo {year}
  {2023}{\natexlab{b}})}\BibitemShut {NoStop}%
\bibitem [{\citenamefont {Bojesen}\ \emph {et~al.}(2014)\citenamefont
  {Bojesen}, \citenamefont {Babaev},\ and\ \citenamefont
  {Sudb{\o}}}]{bojesen2014phase}%
  \BibitemOpen
  \bibfield  {author} {\bibinfo {author} {\bibfnamefont {T.~A.}\ \bibnamefont
  {Bojesen}}, \bibinfo {author} {\bibfnamefont {E.}~\bibnamefont {Babaev}},\
  and\ \bibinfo {author} {\bibfnamefont {A.}~\bibnamefont {Sudb{\o}}},\
  }\bibfield  {title} {\bibinfo {title} {Phase transitions and anomalous normal
  state in superconductors with broken time-reversal symmetry},\ }\href@noop {}
  {\bibfield  {journal} {\bibinfo  {journal} {Physical Review B}\ }\textbf
  {\bibinfo {volume} {89}},\ \bibinfo {pages} {104509} (\bibinfo {year}
  {2014})}\BibitemShut {NoStop}%
\bibitem [{\citenamefont {Ghosh}\ \emph {et~al.}(2020)\citenamefont {Ghosh},
  \citenamefont {Smidman}, \citenamefont {Shang}, \citenamefont {Annett},
  \citenamefont {Hillier}, \citenamefont {Quintanilla},\ and\ \citenamefont
  {Yuan}}]{ghosh2020recent}%
  \BibitemOpen
  \bibfield  {author} {\bibinfo {author} {\bibfnamefont {S.~K.}\ \bibnamefont
  {Ghosh}}, \bibinfo {author} {\bibfnamefont {M.}~\bibnamefont {Smidman}},
  \bibinfo {author} {\bibfnamefont {T.}~\bibnamefont {Shang}}, \bibinfo
  {author} {\bibfnamefont {J.~F.}\ \bibnamefont {Annett}}, \bibinfo {author}
  {\bibfnamefont {A.~D.}\ \bibnamefont {Hillier}}, \bibinfo {author}
  {\bibfnamefont {J.}~\bibnamefont {Quintanilla}},\ and\ \bibinfo {author}
  {\bibfnamefont {H.}~\bibnamefont {Yuan}},\ }\bibfield  {title} {\bibinfo
  {title} {{Recent progress on superconductors with time-reversal symmetry
  breaking}},\ }\href@noop {} {\bibfield  {journal} {\bibinfo  {journal}
  {Journal of Physics: Condensed Matter}\ }\textbf {\bibinfo {volume} {33}},\
  \bibinfo {pages} {033001} (\bibinfo {year} {2020})}\BibitemShut {NoStop}%
\bibitem [{\citenamefont {Ghosh}\ \emph {et~al.}(2021)\citenamefont {Ghosh},
  \citenamefont {Shekhter}, \citenamefont {Jerzembeck}, \citenamefont
  {Kikugawa}, \citenamefont {Sokolov}, \citenamefont {Brando}, \citenamefont
  {Mackenzie}, \citenamefont {Hicks},\ and\ \citenamefont
  {Ramshaw}}]{ghosh2021thermodynamic}%
  \BibitemOpen
  \bibfield  {author} {\bibinfo {author} {\bibfnamefont {S.}~\bibnamefont
  {Ghosh}}, \bibinfo {author} {\bibfnamefont {A.}~\bibnamefont {Shekhter}},
  \bibinfo {author} {\bibfnamefont {F.}~\bibnamefont {Jerzembeck}}, \bibinfo
  {author} {\bibfnamefont {N.}~\bibnamefont {Kikugawa}}, \bibinfo {author}
  {\bibfnamefont {D.~A.}\ \bibnamefont {Sokolov}}, \bibinfo {author}
  {\bibfnamefont {M.}~\bibnamefont {Brando}}, \bibinfo {author} {\bibfnamefont
  {A.}~\bibnamefont {Mackenzie}}, \bibinfo {author} {\bibfnamefont {C.~W.}\
  \bibnamefont {Hicks}},\ and\ \bibinfo {author} {\bibfnamefont
  {B.}~\bibnamefont {Ramshaw}},\ }\bibfield  {title} {\bibinfo {title}
  {{Thermodynamic evidence for a two-component superconducting order parameter
  in Sr$_2$RuO$_4$}},\ }\href@noop {} {\bibfield  {journal} {\bibinfo
  {journal} {Nature Physics}\ }\textbf {\bibinfo {volume} {17}},\ \bibinfo
  {pages} {199} (\bibinfo {year} {2021})}\BibitemShut {NoStop}%
\bibitem [{\citenamefont {Schemm}\ \emph {et~al.}(2014)\citenamefont {Schemm},
  \citenamefont {Gannon}, \citenamefont {Wishne}, \citenamefont {Halperin},\
  and\ \citenamefont {Kapitulnik}}]{schemm2014observation}%
  \BibitemOpen
  \bibfield  {author} {\bibinfo {author} {\bibfnamefont {E.}~\bibnamefont
  {Schemm}}, \bibinfo {author} {\bibfnamefont {W.}~\bibnamefont {Gannon}},
  \bibinfo {author} {\bibfnamefont {C.}~\bibnamefont {Wishne}}, \bibinfo
  {author} {\bibfnamefont {W.}~\bibnamefont {Halperin}},\ and\ \bibinfo
  {author} {\bibfnamefont {A.}~\bibnamefont {Kapitulnik}},\ }\bibfield  {title}
  {\bibinfo {title} {{Observation of broken time-reversal symmetry in the
  heavy-fermion superconductor UPt$_3$}},\ }\href@noop {} {\bibfield  {journal}
  {\bibinfo  {journal} {Science}\ }\textbf {\bibinfo {volume} {345}},\ \bibinfo
  {pages} {190} (\bibinfo {year} {2014})}\BibitemShut {NoStop}%
\bibitem [{\citenamefont {Grinenko}\ \emph {et~al.}(2021)\citenamefont
  {Grinenko}, \citenamefont {Weston}, \citenamefont {Caglieris}, \citenamefont
  {Wuttke}, \citenamefont {Hess}, \citenamefont {Gottschall}, \citenamefont
  {Maccari}, \citenamefont {Gorbunov}, \citenamefont {Zherlitsyn},
  \citenamefont {Wosnitza} \emph {et~al.}}]{grinenko2021state}%
  \BibitemOpen
  \bibfield  {author} {\bibinfo {author} {\bibfnamefont {V.}~\bibnamefont
  {Grinenko}}, \bibinfo {author} {\bibfnamefont {D.}~\bibnamefont {Weston}},
  \bibinfo {author} {\bibfnamefont {F.}~\bibnamefont {Caglieris}}, \bibinfo
  {author} {\bibfnamefont {C.}~\bibnamefont {Wuttke}}, \bibinfo {author}
  {\bibfnamefont {C.}~\bibnamefont {Hess}}, \bibinfo {author} {\bibfnamefont
  {T.}~\bibnamefont {Gottschall}}, \bibinfo {author} {\bibfnamefont
  {I.}~\bibnamefont {Maccari}}, \bibinfo {author} {\bibfnamefont
  {D.}~\bibnamefont {Gorbunov}}, \bibinfo {author} {\bibfnamefont
  {S.}~\bibnamefont {Zherlitsyn}}, \bibinfo {author} {\bibfnamefont
  {J.}~\bibnamefont {Wosnitza}}, \emph {et~al.},\ }\bibfield  {title} {\bibinfo
  {title} {State with spontaneously broken time-reversal symmetry above the
  superconducting phase transition},\ }\href@noop {} {\bibfield  {journal}
  {\bibinfo  {journal} {Nature Physics}\ }\textbf {\bibinfo {volume} {17}},\
  \bibinfo {pages} {1254} (\bibinfo {year} {2021})}\BibitemShut {NoStop}%
\bibitem [{\citenamefont {Can}\ \emph {et~al.}(2021)\citenamefont {Can},
  \citenamefont {Tummuru}, \citenamefont {Day}, \citenamefont {Elfimov},
  \citenamefont {Damascelli},\ and\ \citenamefont {Franz}}]{can2021high}%
  \BibitemOpen
  \bibfield  {author} {\bibinfo {author} {\bibfnamefont {O.}~\bibnamefont
  {Can}}, \bibinfo {author} {\bibfnamefont {T.}~\bibnamefont {Tummuru}},
  \bibinfo {author} {\bibfnamefont {R.~P.}\ \bibnamefont {Day}}, \bibinfo
  {author} {\bibfnamefont {I.}~\bibnamefont {Elfimov}}, \bibinfo {author}
  {\bibfnamefont {A.}~\bibnamefont {Damascelli}},\ and\ \bibinfo {author}
  {\bibfnamefont {M.}~\bibnamefont {Franz}},\ }\bibfield  {title} {\bibinfo
  {title} {High-temperature topological superconductivity in twisted
  double-layer copper oxides},\ }\href@noop {} {\bibfield  {journal} {\bibinfo
  {journal} {Nature Physics}\ }\textbf {\bibinfo {volume} {17}},\ \bibinfo
  {pages} {519} (\bibinfo {year} {2021})}\BibitemShut {NoStop}%
\bibitem [{\citenamefont {Zhao}\ \emph {et~al.}(2023)\citenamefont {Zhao},
  \citenamefont {Cui}, \citenamefont {Volkov}, \citenamefont {Yoo},
  \citenamefont {Lee}, \citenamefont {Gardener}, \citenamefont {Akey},
  \citenamefont {Engelke}, \citenamefont {Ronen}, \citenamefont {Zhong} \emph
  {et~al.}}]{zhao2023time}%
  \BibitemOpen
  \bibfield  {author} {\bibinfo {author} {\bibfnamefont {S.~F.}\ \bibnamefont
  {Zhao}}, \bibinfo {author} {\bibfnamefont {X.}~\bibnamefont {Cui}}, \bibinfo
  {author} {\bibfnamefont {P.~A.}\ \bibnamefont {Volkov}}, \bibinfo {author}
  {\bibfnamefont {H.}~\bibnamefont {Yoo}}, \bibinfo {author} {\bibfnamefont
  {S.}~\bibnamefont {Lee}}, \bibinfo {author} {\bibfnamefont {J.~A.}\
  \bibnamefont {Gardener}}, \bibinfo {author} {\bibfnamefont {A.~J.}\
  \bibnamefont {Akey}}, \bibinfo {author} {\bibfnamefont {R.}~\bibnamefont
  {Engelke}}, \bibinfo {author} {\bibfnamefont {Y.}~\bibnamefont {Ronen}},
  \bibinfo {author} {\bibfnamefont {R.}~\bibnamefont {Zhong}}, \emph {et~al.},\
  }\bibfield  {title} {\bibinfo {title} {Time-reversal symmetry breaking
  superconductivity between twisted cuprate superconductors},\ }\href@noop {}
  {\bibfield  {journal} {\bibinfo  {journal} {Science}\ }\textbf {\bibinfo
  {volume} {382}},\ \bibinfo {pages} {1422} (\bibinfo {year}
  {2023})}\BibitemShut {NoStop}%
\bibitem [{\citenamefont {Yao}\ \emph {et~al.}(2018)\citenamefont {Yao},
  \citenamefont {Wang}, \citenamefont {Bao}, \citenamefont {Zhang},
  \citenamefont {Zhang}, \citenamefont {Bao}, \citenamefont {Chan},
  \citenamefont {Chen}, \citenamefont {Avila}, \citenamefont {Asensio} \emph
  {et~al.}}]{yao2018quasicrystalline}%
  \BibitemOpen
  \bibfield  {author} {\bibinfo {author} {\bibfnamefont {W.}~\bibnamefont
  {Yao}}, \bibinfo {author} {\bibfnamefont {E.}~\bibnamefont {Wang}}, \bibinfo
  {author} {\bibfnamefont {C.}~\bibnamefont {Bao}}, \bibinfo {author}
  {\bibfnamefont {Y.}~\bibnamefont {Zhang}}, \bibinfo {author} {\bibfnamefont
  {K.}~\bibnamefont {Zhang}}, \bibinfo {author} {\bibfnamefont
  {K.}~\bibnamefont {Bao}}, \bibinfo {author} {\bibfnamefont {C.~K.}\
  \bibnamefont {Chan}}, \bibinfo {author} {\bibfnamefont {C.}~\bibnamefont
  {Chen}}, \bibinfo {author} {\bibfnamefont {J.}~\bibnamefont {Avila}},
  \bibinfo {author} {\bibfnamefont {M.~C.}\ \bibnamefont {Asensio}}, \emph
  {et~al.},\ }\bibfield  {title} {\bibinfo {title} {Quasicrystalline 30 twisted
  bilayer graphene as an incommensurate superlattice with strong interlayer
  coupling},\ }\href@noop {} {\bibfield  {journal} {\bibinfo  {journal}
  {Proceedings of the National Academy of Sciences}\ }\textbf {\bibinfo
  {volume} {115}},\ \bibinfo {pages} {6928} (\bibinfo {year}
  {2018})}\BibitemShut {NoStop}%
\bibitem [{\citenamefont {Pezzini}\ \emph {et~al.}(2020)\citenamefont
  {Pezzini}, \citenamefont {Miseikis}, \citenamefont {Piccinini}, \citenamefont
  {Forti}, \citenamefont {Pace}, \citenamefont {Engelke}, \citenamefont
  {Rossella}, \citenamefont {Watanabe}, \citenamefont {Taniguchi},
  \citenamefont {Kim} \emph {et~al.}}]{pezzini202030}%
  \BibitemOpen
  \bibfield  {author} {\bibinfo {author} {\bibfnamefont {S.}~\bibnamefont
  {Pezzini}}, \bibinfo {author} {\bibfnamefont {V.}~\bibnamefont {Miseikis}},
  \bibinfo {author} {\bibfnamefont {G.}~\bibnamefont {Piccinini}}, \bibinfo
  {author} {\bibfnamefont {S.}~\bibnamefont {Forti}}, \bibinfo {author}
  {\bibfnamefont {S.}~\bibnamefont {Pace}}, \bibinfo {author} {\bibfnamefont
  {R.}~\bibnamefont {Engelke}}, \bibinfo {author} {\bibfnamefont
  {F.}~\bibnamefont {Rossella}}, \bibinfo {author} {\bibfnamefont
  {K.}~\bibnamefont {Watanabe}}, \bibinfo {author} {\bibfnamefont
  {T.}~\bibnamefont {Taniguchi}}, \bibinfo {author} {\bibfnamefont
  {P.}~\bibnamefont {Kim}}, \emph {et~al.},\ }\bibfield  {title} {\bibinfo
  {title} {30-twisted bilayer graphene quasicrystals from chemical vapor
  deposition},\ }\href@noop {} {\bibfield  {journal} {\bibinfo  {journal} {Nano
  letters}\ }\textbf {\bibinfo {volume} {20}},\ \bibinfo {pages} {3313}
  (\bibinfo {year} {2020})}\BibitemShut {NoStop}%
\bibitem [{\citenamefont {Deng}\ \emph {et~al.}(2020)\citenamefont {Deng},
  \citenamefont {Wang}, \citenamefont {Li}, \citenamefont {Li}, \citenamefont
  {Wang}, \citenamefont {Tang}, \citenamefont {Fu}, \citenamefont {Tian},
  \citenamefont {Gao}, \citenamefont {Xue} \emph
  {et~al.}}]{deng2020interlayer}%
  \BibitemOpen
  \bibfield  {author} {\bibinfo {author} {\bibfnamefont {B.}~\bibnamefont
  {Deng}}, \bibinfo {author} {\bibfnamefont {B.}~\bibnamefont {Wang}}, \bibinfo
  {author} {\bibfnamefont {N.}~\bibnamefont {Li}}, \bibinfo {author}
  {\bibfnamefont {R.}~\bibnamefont {Li}}, \bibinfo {author} {\bibfnamefont
  {Y.}~\bibnamefont {Wang}}, \bibinfo {author} {\bibfnamefont {J.}~\bibnamefont
  {Tang}}, \bibinfo {author} {\bibfnamefont {Q.}~\bibnamefont {Fu}}, \bibinfo
  {author} {\bibfnamefont {Z.}~\bibnamefont {Tian}}, \bibinfo {author}
  {\bibfnamefont {P.}~\bibnamefont {Gao}}, \bibinfo {author} {\bibfnamefont
  {J.}~\bibnamefont {Xue}}, \emph {et~al.},\ }\bibfield  {title} {\bibinfo
  {title} {Interlayer decoupling in 30 twisted bilayer graphene quasicrystal},\
  }\href@noop {} {\bibfield  {journal} {\bibinfo  {journal} {ACS nano}\
  }\textbf {\bibinfo {volume} {14}},\ \bibinfo {pages} {1656} (\bibinfo {year}
  {2020})}\BibitemShut {NoStop}%
\bibitem [{\citenamefont {Liu}\ \emph {et~al.}(2023)\citenamefont {Liu},
  \citenamefont {Zhou}, \citenamefont {Wu},\ and\ \citenamefont
  {Yang}}]{liu2023charge}%
  \BibitemOpen
  \bibfield  {author} {\bibinfo {author} {\bibfnamefont {Y.-B.}\ \bibnamefont
  {Liu}}, \bibinfo {author} {\bibfnamefont {J.}~\bibnamefont {Zhou}}, \bibinfo
  {author} {\bibfnamefont {C.}~\bibnamefont {Wu}},\ and\ \bibinfo {author}
  {\bibfnamefont {F.}~\bibnamefont {Yang}},\ }\bibfield  {title} {\bibinfo
  {title} {Charge-4e superconductivity and chiral metal in 45°-twisted bilayer
  cuprates and related bilayers},\ }\href@noop {} {\bibfield  {journal}
  {\bibinfo  {journal} {Nature Communications}\ }\textbf {\bibinfo {volume}
  {14}},\ \bibinfo {pages} {7926} (\bibinfo {year} {2023})}\BibitemShut
  {NoStop}%
\bibitem [{\citenamefont {Maccari}\ and\ \citenamefont
  {Babaev}(2022)}]{maccari2022effects}%
  \BibitemOpen
  \bibfield  {author} {\bibinfo {author} {\bibfnamefont {I.}~\bibnamefont
  {Maccari}}\ and\ \bibinfo {author} {\bibfnamefont {E.}~\bibnamefont
  {Babaev}},\ }\bibfield  {title} {\bibinfo {title} {Effects of intercomponent
  couplings on the appearance of time-reversal symmetry breaking
  fermion-quadrupling states in two-component london models},\ }\href@noop {}
  {\bibfield  {journal} {\bibinfo  {journal} {Physical Review B}\ }\textbf
  {\bibinfo {volume} {105}},\ \bibinfo {pages} {214520} (\bibinfo {year}
  {2022})}\BibitemShut {NoStop}%
\bibitem [{\citenamefont {Maiti}\ and\ \citenamefont
  {Chubukov}(2013)}]{maiti2013s+}%
  \BibitemOpen
  \bibfield  {author} {\bibinfo {author} {\bibfnamefont {S.}~\bibnamefont
  {Maiti}}\ and\ \bibinfo {author} {\bibfnamefont {A.~V.}\ \bibnamefont
  {Chubukov}},\ }\bibfield  {title} {\bibinfo {title} {s+ i s state with broken
  time-reversal symmetry in fe-based superconductors},\ }\href@noop {}
  {\bibfield  {journal} {\bibinfo  {journal} {Physical Review B}\ }\textbf
  {\bibinfo {volume} {87}},\ \bibinfo {pages} {144511} (\bibinfo {year}
  {2013})}\BibitemShut {NoStop}%
\bibitem [{\citenamefont {Mukherjee}\ and\ \citenamefont
  {Agterberg}(2011)}]{mukherjee2011role}%
  \BibitemOpen
  \bibfield  {author} {\bibinfo {author} {\bibfnamefont {S.}~\bibnamefont
  {Mukherjee}}\ and\ \bibinfo {author} {\bibfnamefont {D.}~\bibnamefont
  {Agterberg}},\ }\bibfield  {title} {\bibinfo {title} {Role of d-wave pairing
  in a 15 superconductors},\ }\href@noop {} {\bibfield  {journal} {\bibinfo
  {journal} {Physical Review B}\ }\textbf {\bibinfo {volume} {84}},\ \bibinfo
  {pages} {134520} (\bibinfo {year} {2011})}\BibitemShut {NoStop}%
\bibitem [{\citenamefont {Lee}\ \emph {et~al.}(2009)\citenamefont {Lee},
  \citenamefont {Zhang},\ and\ \citenamefont {Wu}}]{lee2009pairing}%
  \BibitemOpen
  \bibfield  {author} {\bibinfo {author} {\bibfnamefont {W.-C.}\ \bibnamefont
  {Lee}}, \bibinfo {author} {\bibfnamefont {S.-C.}\ \bibnamefont {Zhang}},\
  and\ \bibinfo {author} {\bibfnamefont {C.}~\bibnamefont {Wu}},\ }\bibfield
  {title} {\bibinfo {title} {Pairing state with a time-reversal symmetry
  breaking in feas-based superconductors},\ }\href@noop {} {\bibfield
  {journal} {\bibinfo  {journal} {Physical review letters}\ }\textbf {\bibinfo
  {volume} {102}},\ \bibinfo {pages} {217002} (\bibinfo {year}
  {2009})}\BibitemShut {NoStop}%
\bibitem [{\citenamefont {Platt}\ \emph {et~al.}(2012)\citenamefont {Platt},
  \citenamefont {Thomale}, \citenamefont {Honerkamp}, \citenamefont {Zhang},\
  and\ \citenamefont {Hanke}}]{platt2012mechanism}%
  \BibitemOpen
  \bibfield  {author} {\bibinfo {author} {\bibfnamefont {C.}~\bibnamefont
  {Platt}}, \bibinfo {author} {\bibfnamefont {R.}~\bibnamefont {Thomale}},
  \bibinfo {author} {\bibfnamefont {C.}~\bibnamefont {Honerkamp}}, \bibinfo
  {author} {\bibfnamefont {S.-C.}\ \bibnamefont {Zhang}},\ and\ \bibinfo
  {author} {\bibfnamefont {W.}~\bibnamefont {Hanke}},\ }\bibfield  {title}
  {\bibinfo {title} {Mechanism for a pairing state with time-reversal symmetry
  breaking in iron-based superconductors},\ }\href@noop {} {\bibfield
  {journal} {\bibinfo  {journal} {Physical Review B}\ }\textbf {\bibinfo
  {volume} {85}},\ \bibinfo {pages} {180502} (\bibinfo {year}
  {2012})}\BibitemShut {NoStop}%
\bibitem [{\citenamefont {Yerin}\ \emph {et~al.}(2017)\citenamefont {Yerin},
  \citenamefont {Omelyanchouk}, \citenamefont {Drechsler}, \citenamefont
  {Efremov},\ and\ \citenamefont {van~den Brink}}]{yerin2017anomalous}%
  \BibitemOpen
  \bibfield  {author} {\bibinfo {author} {\bibfnamefont {Y.}~\bibnamefont
  {Yerin}}, \bibinfo {author} {\bibfnamefont {A.}~\bibnamefont {Omelyanchouk}},
  \bibinfo {author} {\bibfnamefont {S.-L.}\ \bibnamefont {Drechsler}}, \bibinfo
  {author} {\bibfnamefont {D.~V.}\ \bibnamefont {Efremov}},\ and\ \bibinfo
  {author} {\bibfnamefont {J.}~\bibnamefont {van~den Brink}},\ }\bibfield
  {title} {\bibinfo {title} {Anomalous diamagnetic response in multiband
  superconductors with broken time-reversal symmetry},\ }\href@noop {}
  {\bibfield  {journal} {\bibinfo  {journal} {Physical Review B}\ }\textbf
  {\bibinfo {volume} {96}},\ \bibinfo {pages} {144513} (\bibinfo {year}
  {2017})}\BibitemShut {NoStop}%
\bibitem [{\citenamefont {Yerin}\ \emph {et~al.}(2022)\citenamefont {Yerin},
  \citenamefont {Drechsler}, \citenamefont {Cuoco},\ and\ \citenamefont
  {Petrillo}}]{yerin2022magneto}%
  \BibitemOpen
  \bibfield  {author} {\bibinfo {author} {\bibfnamefont {Y.}~\bibnamefont
  {Yerin}}, \bibinfo {author} {\bibfnamefont {S.-L.}\ \bibnamefont
  {Drechsler}}, \bibinfo {author} {\bibfnamefont {M.}~\bibnamefont {Cuoco}},\
  and\ \bibinfo {author} {\bibfnamefont {C.}~\bibnamefont {Petrillo}},\
  }\bibfield  {title} {\bibinfo {title} {Magneto-topological transitions in
  multicomponent superconductors},\ }\href@noop {} {\bibfield  {journal}
  {\bibinfo  {journal} {Physical Review B}\ }\textbf {\bibinfo {volume}
  {106}},\ \bibinfo {pages} {054517} (\bibinfo {year} {2022})}\BibitemShut
  {NoStop}%
\bibitem [{\citenamefont {Song}\ and\ \citenamefont
  {Zhang}(2022)}]{song2022phase}%
  \BibitemOpen
  \bibfield  {author} {\bibinfo {author} {\bibfnamefont {F.-F.}\ \bibnamefont
  {Song}}\ and\ \bibinfo {author} {\bibfnamefont {G.-M.}\ \bibnamefont
  {Zhang}},\ }\bibfield  {title} {\bibinfo {title} {Phase coherence of pairs of
  cooper pairs as quasi-long-range order of half-vortex pairs in a
  two-dimensional bilayer system},\ }\href@noop {} {\bibfield  {journal}
  {\bibinfo  {journal} {Physical Review Letters}\ }\textbf {\bibinfo {volume}
  {128}},\ \bibinfo {pages} {195301} (\bibinfo {year} {2022})}\BibitemShut
  {NoStop}%
\bibitem [{\citenamefont {Stanev}\ and\ \citenamefont
  {Te{\v{s}}anovi{\'c}}(2010)}]{stanev2010three}%
  \BibitemOpen
  \bibfield  {author} {\bibinfo {author} {\bibfnamefont {V.}~\bibnamefont
  {Stanev}}\ and\ \bibinfo {author} {\bibfnamefont {Z.}~\bibnamefont
  {Te{\v{s}}anovi{\'c}}},\ }\bibfield  {title} {\bibinfo {title} {Three-band
  superconductivity and the order parameter that breaks time-reversal
  symmetry},\ }\href@noop {} {\bibfield  {journal} {\bibinfo  {journal}
  {Physical review B}\ }\textbf {\bibinfo {volume} {81}},\ \bibinfo {pages}
  {134522} (\bibinfo {year} {2010})}\BibitemShut {NoStop}%
\bibitem [{\citenamefont {Yuan}(2023)}]{yuan2023exactly}%
  \BibitemOpen
  \bibfield  {author} {\bibinfo {author} {\bibfnamefont {A.~C.}\ \bibnamefont
  {Yuan}},\ }\bibfield  {title} {\bibinfo {title} {Exactly solvable model of
  randomly coupled twisted superconducting bilayers},\ }\href
  {https://doi.org/10.1103/PhysRevB.108.184515} {\bibfield  {journal} {\bibinfo
   {journal} {Phys. Rev. B}\ }\textbf {\bibinfo {volume} {108}},\ \bibinfo
  {pages} {184515} (\bibinfo {year} {2023})}\BibitemShut {NoStop}%
\bibitem [{\citenamefont {Zeng}\ \emph {et~al.}(2021)\citenamefont {Zeng},
  \citenamefont {Hu}, \citenamefont {Hu}, \citenamefont {You},\ and\
  \citenamefont {Wu}}]{zeng2021phase}%
  \BibitemOpen
  \bibfield  {author} {\bibinfo {author} {\bibfnamefont {M.}~\bibnamefont
  {Zeng}}, \bibinfo {author} {\bibfnamefont {L.-H.}\ \bibnamefont {Hu}},
  \bibinfo {author} {\bibfnamefont {H.-Y.}\ \bibnamefont {Hu}}, \bibinfo
  {author} {\bibfnamefont {Y.-Z.}\ \bibnamefont {You}},\ and\ \bibinfo {author}
  {\bibfnamefont {C.}~\bibnamefont {Wu}},\ }\bibfield  {title} {\bibinfo
  {title} {Phase-fluctuation induced time-reversal symmetry breaking normal
  state},\ }\href@noop {} {\bibfield  {journal} {\bibinfo  {journal} {arXiv
  preprint arXiv:2102.06158}\ } (\bibinfo {year} {2021})}\BibitemShut {NoStop}%
\bibitem [{\citenamefont {How}\ and\ \citenamefont
  {Yip}(2023)}]{how2023absence}%
  \BibitemOpen
  \bibfield  {author} {\bibinfo {author} {\bibfnamefont {P.~T.}\ \bibnamefont
  {How}}\ and\ \bibinfo {author} {\bibfnamefont {S.~K.}\ \bibnamefont {Yip}},\
  }\bibfield  {title} {\bibinfo {title} {Absence of ginzburg-landau mechanism
  for vestigial order in the normal phase above a two-component
  superconductor},\ }\href@noop {} {\bibfield  {journal} {\bibinfo  {journal}
  {Physical Review B}\ }\textbf {\bibinfo {volume} {107}},\ \bibinfo {pages}
  {104514} (\bibinfo {year} {2023})}\BibitemShut {NoStop}%
\bibitem [{\citenamefont {Lee}\ and\ \citenamefont
  {Grinstein}(1985)}]{lee1985strings}%
  \BibitemOpen
  \bibfield  {author} {\bibinfo {author} {\bibfnamefont {D.}~\bibnamefont
  {Lee}}\ and\ \bibinfo {author} {\bibfnamefont {G.}~\bibnamefont
  {Grinstein}},\ }\bibfield  {title} {\bibinfo {title} {Strings in
  two-dimensional classical xy models},\ }\href@noop {} {\bibfield  {journal}
  {\bibinfo  {journal} {Physical review letters}\ }\textbf {\bibinfo {volume}
  {55}},\ \bibinfo {pages} {541} (\bibinfo {year} {1985})}\BibitemShut
  {NoStop}%
\bibitem [{\citenamefont {Korshunov}(1986)}]{korshunov1986phase}%
  \BibitemOpen
  \bibfield  {author} {\bibinfo {author} {\bibfnamefont {S.}~\bibnamefont
  {Korshunov}},\ }\bibfield  {title} {\bibinfo {title} {Phase diagram of the
  modified xy model},\ }\href@noop {} {\bibfield  {journal} {\bibinfo
  {journal} {Journal of Physics C: Solid State Physics}\ }\textbf {\bibinfo
  {volume} {19}},\ \bibinfo {pages} {4427} (\bibinfo {year}
  {1986})}\BibitemShut {NoStop}%
\bibitem [{\citenamefont {Dub{\'e}dat}\ and\ \citenamefont
  {Falconet}(2022)}]{dubedat2022random}%
  \BibitemOpen
  \bibfield  {author} {\bibinfo {author} {\bibfnamefont {J.}~\bibnamefont
  {Dub{\'e}dat}}\ and\ \bibinfo {author} {\bibfnamefont {H.}~\bibnamefont
  {Falconet}},\ }\bibfield  {title} {\bibinfo {title} {Random clusters in the
  villain and xy models},\ }\href@noop {} {\bibfield  {journal} {\bibinfo
  {journal} {arXiv preprint arXiv:2210.03620}\ } (\bibinfo {year}
  {2022})}\BibitemShut {NoStop}%
\bibitem [{\citenamefont {Song}\ and\ \citenamefont
  {Zhang}(2021)}]{song2021hybrid}%
  \BibitemOpen
  \bibfield  {author} {\bibinfo {author} {\bibfnamefont {F.-F.}\ \bibnamefont
  {Song}}\ and\ \bibinfo {author} {\bibfnamefont {G.-M.}\ \bibnamefont
  {Zhang}},\ }\bibfield  {title} {\bibinfo {title} {Hybrid
  berezinskii-kosterlitz-thouless and ising topological phase transition in the
  generalized two-dimensional xy model using tensor networks},\ }\href@noop {}
  {\bibfield  {journal} {\bibinfo  {journal} {Physical Review B}\ }\textbf
  {\bibinfo {volume} {103}},\ \bibinfo {pages} {024518} (\bibinfo {year}
  {2021})}\BibitemShut {NoStop}%
\bibitem [{\citenamefont {Carpenter}\ and\ \citenamefont
  {Chalker}(1989)}]{carpenter1989phase}%
  \BibitemOpen
  \bibfield  {author} {\bibinfo {author} {\bibfnamefont {D.}~\bibnamefont
  {Carpenter}}\ and\ \bibinfo {author} {\bibfnamefont {J.}~\bibnamefont
  {Chalker}},\ }\bibfield  {title} {\bibinfo {title} {The phase diagram of a
  generalised xy model},\ }\href@noop {} {\bibfield  {journal} {\bibinfo
  {journal} {Journal of Physics: Condensed Matter}\ }\textbf {\bibinfo {volume}
  {1}},\ \bibinfo {pages} {4907} (\bibinfo {year} {1989})}\BibitemShut
  {NoStop}%
\bibitem [{\citenamefont {Nui}\ \emph {et~al.}(2018)\citenamefont {Nui},
  \citenamefont {Tuan}, \citenamefont {Kien}, \citenamefont {Huy},
  \citenamefont {Dang},\ and\ \citenamefont {Viet}}]{nui2018correlation}%
  \BibitemOpen
  \bibfield  {author} {\bibinfo {author} {\bibfnamefont {D.~X.}\ \bibnamefont
  {Nui}}, \bibinfo {author} {\bibfnamefont {L.}~\bibnamefont {Tuan}}, \bibinfo
  {author} {\bibfnamefont {N.~D.~T.}\ \bibnamefont {Kien}}, \bibinfo {author}
  {\bibfnamefont {P.~T.}\ \bibnamefont {Huy}}, \bibinfo {author} {\bibfnamefont
  {H.~T.}\ \bibnamefont {Dang}},\ and\ \bibinfo {author} {\bibfnamefont
  {D.~X.}\ \bibnamefont {Viet}},\ }\bibfield  {title} {\bibinfo {title}
  {Correlation length in a generalized two-dimensional xy model},\ }\href@noop
  {} {\bibfield  {journal} {\bibinfo  {journal} {Physical Review B}\ }\textbf
  {\bibinfo {volume} {98}},\ \bibinfo {pages} {144421} (\bibinfo {year}
  {2018})}\BibitemShut {NoStop}%
\bibitem [{\citenamefont {H{\"u}bscher}\ and\ \citenamefont
  {Wessel}(2013)}]{hubscher2013stiffness}%
  \BibitemOpen
  \bibfield  {author} {\bibinfo {author} {\bibfnamefont {D.~M.}\ \bibnamefont
  {H{\"u}bscher}}\ and\ \bibinfo {author} {\bibfnamefont {S.}~\bibnamefont
  {Wessel}},\ }\bibfield  {title} {\bibinfo {title} {Stiffness jump in the
  generalized x y model on the square lattice},\ }\href@noop {} {\bibfield
  {journal} {\bibinfo  {journal} {Physical Review E}\ }\textbf {\bibinfo
  {volume} {87}},\ \bibinfo {pages} {062112} (\bibinfo {year}
  {2013})}\BibitemShut {NoStop}%
\bibitem [{\citenamefont {Gao}\ \emph {et~al.}(2022)\citenamefont {Gao},
  \citenamefont {Huang},\ and\ \citenamefont {Lee}}]{gao2022fractional}%
  \BibitemOpen
  \bibfield  {author} {\bibinfo {author} {\bibfnamefont {Z.-Q.}\ \bibnamefont
  {Gao}}, \bibinfo {author} {\bibfnamefont {Y.-T.}\ \bibnamefont {Huang}},\
  and\ \bibinfo {author} {\bibfnamefont {D.-H.}\ \bibnamefont {Lee}},\
  }\bibfield  {title} {\bibinfo {title} {Fractional vortices, z 2 gauge theory,
  and the confinement-deconfinement transition},\ }\href@noop {} {\bibfield
  {journal} {\bibinfo  {journal} {Physical Review B}\ }\textbf {\bibinfo
  {volume} {106}},\ \bibinfo {pages} {L121105} (\bibinfo {year}
  {2022})}\BibitemShut {NoStop}%
\bibitem [{int()}]{interchangeLimits}%
  \BibitemOpen
  \href@noop {} {}\bibinfo {note} {The correct order of limits is to take the
  thermodynamic limit $L\to \infty$ first and then the coupling strength $J_2
  \to \infty$. However, we seem to have done the opposite in defining the
  $U(1)\times \dZ_2$ Hamiltonian. The reason is that the correlation functions,
  e.g., $\bra \cos 2(\theta_0 -\theta_R) \ket_{L,J_2}$ where
  $\theta=(\phi^++\phi^-)/2$ is the average $U(1)$ phase, can be shown to be
  monotonically increasing with respect to system size $L$ and coupling
  strength $J_2$ \cite{ginibre1970general}, and thus the limits can be replaced
  be taking supremums (maxes) $\sup_{J_2} \sup_{L} \bra \cos 2(\theta_0
  -\theta_R) \ket_{L,J_2}$ and it is clear that supremums can be interchanged.
  Similar statements hold for the $\dZ_2$ correlation $\bra \sin \phi_0 \sin
  \phi_R\ket_{L,J_2}$ where $\phi=\phi^+-\phi^-$ is the phase
  difference.}\BibitemShut {Stop}%
\bibitem [{dia()}]{diag}%
  \BibitemOpen
  \href@noop {} {}\bibinfo {note} {Note that in the case where $\lambda =\alpha
  >0$, up to quadratic orders, the current-current coupling remains zero and
  thus should also correspond to the critical regime $\ell$. Indeed, numerics
  in 2D \cite{liu2023charge} confirm that there exists only a single phase
  transition at finite $J_2$. However, at higher orders (or by comparing Eq.
  \eqref{eq:U1-Z2} and \eqref{eq:U1-Z2-general}), the two are clearly
  different. Hence, we abuse our notation slightly and only call the class
  $\lambda=\alpha=0$ the critical regime.}\BibitemShut {Stop}%
\bibitem [{\citenamefont {Ginibre}(1970)}]{ginibre1970general}%
  \BibitemOpen
  \bibfield  {author} {\bibinfo {author} {\bibfnamefont {J.}~\bibnamefont
  {Ginibre}},\ }\bibfield  {title} {\bibinfo {title} {General formulation of
  griffiths' inequalities},\ }\href@noop {} {\bibfield  {journal} {\bibinfo
  {journal} {Communications in mathematical physics}\ }\textbf {\bibinfo
  {volume} {16}},\ \bibinfo {pages} {310} (\bibinfo {year} {1970})}\BibitemShut
  {NoStop}%
\bibitem [{\citenamefont {Fortuin}\ and\ \citenamefont
  {Kasteleyn}(1972)}]{fortuin1972random}%
  \BibitemOpen
  \bibfield  {author} {\bibinfo {author} {\bibfnamefont {C.~M.}\ \bibnamefont
  {Fortuin}}\ and\ \bibinfo {author} {\bibfnamefont {P.~W.}\ \bibnamefont
  {Kasteleyn}},\ }\bibfield  {title} {\bibinfo {title} {On the random-cluster
  model: I. introduction and relation to other models},\ }\href@noop {}
  {\bibfield  {journal} {\bibinfo  {journal} {Physica}\ }\textbf {\bibinfo
  {volume} {57}},\ \bibinfo {pages} {536} (\bibinfo {year} {1972})}\BibitemShut
  {NoStop}%
\bibitem [{\citenamefont {Duminil-Copin}(2017)}]{duminil2017lectures}%
  \BibitemOpen
  \bibfield  {author} {\bibinfo {author} {\bibfnamefont {H.}~\bibnamefont
  {Duminil-Copin}},\ }\bibfield  {title} {\bibinfo {title} {Lectures on the
  ising and potts models on the hypercubic lattice},\ }in\ \href@noop {} {\emph
  {\bibinfo {booktitle} {PIMS-CRM Summer School in Probability}}}\ (\bibinfo
  {publisher} {Springer},\ \bibinfo {year} {2017})\ pp.\ \bibinfo {pages}
  {35--161}\BibitemShut {NoStop}%
\bibitem [{\citenamefont {Duminil-Copin}\ \emph {et~al.}(2017)\citenamefont
  {Duminil-Copin}, \citenamefont {Sidoravicius},\ and\ \citenamefont
  {Tassion}}]{duminil2017continuity}%
  \BibitemOpen
  \bibfield  {author} {\bibinfo {author} {\bibfnamefont {H.}~\bibnamefont
  {Duminil-Copin}}, \bibinfo {author} {\bibfnamefont {V.}~\bibnamefont
  {Sidoravicius}},\ and\ \bibinfo {author} {\bibfnamefont {V.}~\bibnamefont
  {Tassion}},\ }\bibfield  {title} {\bibinfo {title} {{Continuity of the phase
  transition for planar random-cluster and Potts models with $1\le q\le 4$}},\
  }\href@noop {} {\bibfield  {journal} {\bibinfo  {journal} {Communications in
  Mathematical Physics}\ }\textbf {\bibinfo {volume} {349}},\ \bibinfo {pages}
  {47} (\bibinfo {year} {2017})}\BibitemShut {NoStop}%
\bibitem [{\citenamefont {Duminil-Copin}\ \emph {et~al.}(2016)\citenamefont
  {Duminil-Copin}, \citenamefont {Gagnebin}, \citenamefont {Harel},
  \citenamefont {Manolescu},\ and\ \citenamefont
  {Tassion}}]{duminil2016discontinuity}%
  \BibitemOpen
  \bibfield  {author} {\bibinfo {author} {\bibfnamefont {H.}~\bibnamefont
  {Duminil-Copin}}, \bibinfo {author} {\bibfnamefont {M.}~\bibnamefont
  {Gagnebin}}, \bibinfo {author} {\bibfnamefont {M.}~\bibnamefont {Harel}},
  \bibinfo {author} {\bibfnamefont {I.}~\bibnamefont {Manolescu}},\ and\
  \bibinfo {author} {\bibfnamefont {V.}~\bibnamefont {Tassion}},\ }\bibfield
  {title} {\bibinfo {title} {Discontinuity of the phase transition for the
  planar random-cluster and potts models with $ q> 4$},\ }\href@noop {}
  {\bibfield  {journal} {\bibinfo  {journal} {arXiv preprint arXiv:1611.09877}\
  } (\bibinfo {year} {2016})}\BibitemShut {NoStop}%
\bibitem [{\citenamefont {Pfister}\ and\ \citenamefont
  {Velenik}(1997)}]{pfister1997random}%
  \BibitemOpen
  \bibfield  {author} {\bibinfo {author} {\bibfnamefont {C.~E.}\ \bibnamefont
  {Pfister}}\ and\ \bibinfo {author} {\bibfnamefont {Y.}~\bibnamefont
  {Velenik}},\ }\bibfield  {title} {\bibinfo {title} {Random-cluster
  representation of the ashkin-teller model},\ }\href@noop {} {\bibfield
  {journal} {\bibinfo  {journal} {Journal of statistical physics}\ }\textbf
  {\bibinfo {volume} {88}},\ \bibinfo {pages} {1295} (\bibinfo {year}
  {1997})}\BibitemShut {NoStop}%
\bibitem [{\citenamefont {Aoun}\ \emph {et~al.}(2023)\citenamefont {Aoun},
  \citenamefont {Dober},\ and\ \citenamefont {Glazman}}]{aoun2023phase}%
  \BibitemOpen
  \bibfield  {author} {\bibinfo {author} {\bibfnamefont {Y.}~\bibnamefont
  {Aoun}}, \bibinfo {author} {\bibfnamefont {M.}~\bibnamefont {Dober}},\ and\
  \bibinfo {author} {\bibfnamefont {A.}~\bibnamefont {Glazman}},\ }\bibfield
  {title} {\bibinfo {title} {Phase diagram of the ashkin-teller model},\
  }\href@noop {} {\bibfield  {journal} {\bibinfo  {journal} {arXiv preprint
  arXiv:2301.10609}\ } (\bibinfo {year} {2023})}\BibitemShut {NoStop}%
\bibitem [{\citenamefont {Binder}(2022)}]{binder2022monte}%
  \BibitemOpen
  \bibfield  {author} {\bibinfo {author} {\bibfnamefont {K.}~\bibnamefont
  {Binder}},\ }\bibfield  {title} {\bibinfo {title} {Monte carlo simulations in
  statistical physics},\ }in\ \href@noop {} {\emph {\bibinfo {booktitle}
  {Statistical and Nonlinear Physics}}}\ (\bibinfo  {publisher} {Springer},\
  \bibinfo {year} {2022})\ pp.\ \bibinfo {pages} {85--97}\BibitemShut {NoStop}%
\bibitem [{\citenamefont {Chayes}(1998)}]{chayes1998discontinuity}%
  \BibitemOpen
  \bibfield  {author} {\bibinfo {author} {\bibfnamefont {L.}~\bibnamefont
  {Chayes}},\ }\bibfield  {title} {\bibinfo {title} {Discontinuity of the
  spin-wave stiffness in the two-dimensional xy model},\ }\href@noop {}
  {\bibfield  {journal} {\bibinfo  {journal} {Communications in mathematical
  physics}\ }\textbf {\bibinfo {volume} {197}},\ \bibinfo {pages} {623}
  (\bibinfo {year} {1998})}\BibitemShut {NoStop}%
\bibitem [{\citenamefont {Simon}(1980)}]{simon1980correlation}%
  \BibitemOpen
  \bibfield  {author} {\bibinfo {author} {\bibfnamefont {B.}~\bibnamefont
  {Simon}},\ }\bibfield  {title} {\bibinfo {title} {Correlation inequalities
  and the decay of correlations in ferromagnets},\ }\href@noop {} {\bibfield
  {journal} {\bibinfo  {journal} {Communications in Mathematical Physics}\
  }\textbf {\bibinfo {volume} {77}},\ \bibinfo {pages} {111} (\bibinfo {year}
  {1980})}\BibitemShut {NoStop}%
\bibitem [{\citenamefont {Lieb}(1980)}]{lieb1980refinement}%
  \BibitemOpen
  \bibfield  {author} {\bibinfo {author} {\bibfnamefont {E.~H.}\ \bibnamefont
  {Lieb}},\ }\bibfield  {title} {\bibinfo {title} {A refinement of simon's
  correlation inequality},\ }\href@noop {} {\bibfield  {journal} {\bibinfo
  {journal} {Communications in Mathematical Physics}\ }\textbf {\bibinfo
  {volume} {77}},\ \bibinfo {pages} {127} (\bibinfo {year} {1980})}\BibitemShut
  {NoStop}%
\bibitem [{\citenamefont {Rivasseau}(1980)}]{rivasseau1980lieb}%
  \BibitemOpen
  \bibfield  {author} {\bibinfo {author} {\bibfnamefont {V.}~\bibnamefont
  {Rivasseau}},\ }\bibfield  {title} {\bibinfo {title} {{Lieb's correlation
  inequality for plane rotors}},\ }\href@noop {} {\bibfield  {journal}
  {\bibinfo  {journal} {Communications in Mathematical Physics}\ }\textbf
  {\bibinfo {volume} {77}},\ \bibinfo {pages} {145} (\bibinfo {year}
  {1980})}\BibitemShut {NoStop}%
\bibitem [{\citenamefont {Bauerschmidt}(2016)}]{bauerschmidt2016ferromagnetic}%
  \BibitemOpen
  \bibfield  {author} {\bibinfo {author} {\bibfnamefont {R.}~\bibnamefont
  {Bauerschmidt}},\ }\bibfield  {title} {\bibinfo {title} {Ferromagnetic spin
  systems},\ }\href@noop {} {\bibfield  {journal} {\bibinfo  {journal} {Lecture
  notes available at http://www. statslab. cam. ac. uk/~ rb812/doc/spin. pdf}\
  } (\bibinfo {year} {2016})}\BibitemShut {NoStop}%
\bibitem [{\citenamefont {Van~Enter}\ and\ \citenamefont
  {Shlosman}(2002)}]{van2002first}%
  \BibitemOpen
  \bibfield  {author} {\bibinfo {author} {\bibfnamefont {A.~C.}\ \bibnamefont
  {Van~Enter}}\ and\ \bibinfo {author} {\bibfnamefont {S.~B.}\ \bibnamefont
  {Shlosman}},\ }\bibfield  {title} {\bibinfo {title} {First-order transitions
  for n-vector models in two and more dimensions: Rigorous proof},\ }\href@noop
  {} {\bibfield  {journal} {\bibinfo  {journal} {Physical review letters}\
  }\textbf {\bibinfo {volume} {89}},\ \bibinfo {pages} {285702} (\bibinfo
  {year} {2002})}\BibitemShut {NoStop}%
\bibitem [{\citenamefont {van Enter}\ and\ \citenamefont
  {Shlosman}(2005)}]{van2005first}%
  \BibitemOpen
  \bibfield  {author} {\bibinfo {author} {\bibfnamefont {A.}~\bibnamefont {van
  Enter}}\ and\ \bibinfo {author} {\bibfnamefont {S.}~\bibnamefont
  {Shlosman}},\ }\bibfield  {title} {\bibinfo {title} {First-order transitions
  for very nonlinear sigma models},\ }\href@noop {} {\bibfield  {journal}
  {\bibinfo  {journal} {arXiv preprint cond-mat/0506730}\ } (\bibinfo {year}
  {2005})}\BibitemShut {NoStop}%
\bibitem [{\citenamefont {Van~Enter}\ \emph {et~al.}(2006)\citenamefont
  {Van~Enter}, \citenamefont {Romano},\ and\ \citenamefont
  {Zagrebnov}}]{van2006first}%
  \BibitemOpen
  \bibfield  {author} {\bibinfo {author} {\bibfnamefont {A.~C.}\ \bibnamefont
  {Van~Enter}}, \bibinfo {author} {\bibfnamefont {S.}~\bibnamefont {Romano}},\
  and\ \bibinfo {author} {\bibfnamefont {V.~A.}\ \bibnamefont {Zagrebnov}},\
  }\bibfield  {title} {\bibinfo {title} {First-order transitions for some
  generalized xy models},\ }\href@noop {} {\bibfield  {journal} {\bibinfo
  {journal} {Journal of Physics A: Mathematical and General}\ }\textbf
  {\bibinfo {volume} {39}},\ \bibinfo {pages} {L439} (\bibinfo {year}
  {2006})}\BibitemShut {NoStop}%
\bibitem [{fin()}]{finitecoupling}%
  \BibitemOpen
  \href@noop {} {}\bibinfo {note} {This is to be expected, since the edge
  (conditional) probabilities $p_e$ only depend on the edge weights $w_e$ of
  the spin configurations and not the priori distribution of each $\sigma_i$
  (i.e., what if we integrate over a non-uniform distribution of $\sigma_i \in
  \dS^1$, which occurs, for example, in the presence of a magnetic field). For
  example, suppose that the probability distribution of spin configurations
  (analogous to Eq. \eqref{eq:cluster-Is-spinprob}) is schematically written as
  \begin{equation} \dP[\sigma] \propto \prod_{e=ij} w_e (\sigma) \prod_{i}
  \nu_i (\sigma) \end{equation} Where $w_e(\sigma)$ depends on $\sigma$
  implicitly through the values of $\sigma_i,\sigma_j$, and $\nu_i(\sigma)$
  depends on $\sigma$ implicitly through its value $\sigma_i$ on lattice site
  $i$ ($\nu_i(\sigma) =1$ in the standard Ising or XY models). In this case,
  the corresponding edge probabilities $p_e$ will be independent of $\nu_i$. In
  the finite coupling strength system in Eq. \eqref{eq:H-J2}, the Josephson
  tunneling acts as a priori distribution $\nu_i$ and thus is not detected on
  the level of conditional probability $\sim \dP[\omega|\sigma]$; rather, one
  must consider the joint distribution $\sim \dP[(\omega,\sigma)]$, which is
  much more difficult.}\BibitemShut {Stop}%
\bibitem [{\citenamefont {Griffiths}\ \emph {et~al.}(1970)\citenamefont
  {Griffiths}, \citenamefont {Hurst},\ and\ \citenamefont
  {Sherman}}]{griffiths1970concavity}%
  \BibitemOpen
  \bibfield  {author} {\bibinfo {author} {\bibfnamefont {R.~B.}\ \bibnamefont
  {Griffiths}}, \bibinfo {author} {\bibfnamefont {C.~A.}\ \bibnamefont
  {Hurst}},\ and\ \bibinfo {author} {\bibfnamefont {S.}~\bibnamefont
  {Sherman}},\ }\bibfield  {title} {\bibinfo {title} {Concavity of
  magnetization of an ising ferromagnet in a positive external field},\
  }\href@noop {} {\bibfield  {journal} {\bibinfo  {journal} {Journal of
  Mathematical Physics}\ }\textbf {\bibinfo {volume} {11}},\ \bibinfo {pages}
  {790} (\bibinfo {year} {1970})}\BibitemShut {NoStop}%
\bibitem [{\citenamefont {Aizenman}(2005)}]{aizenman2005geometric}%
  \BibitemOpen
  \bibfield  {author} {\bibinfo {author} {\bibfnamefont {M.}~\bibnamefont
  {Aizenman}},\ }\bibfield  {title} {\bibinfo {title} {Geometric analysis of
  $\phi^4$ fields and ising models},\ }in\ \href@noop {} {\emph {\bibinfo
  {booktitle} {Mathematical Problems in Theoretical Physics: Proceedings of the
  VIth International Conference on Mathematical Physics Berlin (West), August
  11--20, 1981}}}\ (\bibinfo {organization} {Springer},\ \bibinfo {year}
  {2005})\ pp.\ \bibinfo {pages} {37--46}\BibitemShut {NoStop}%
\bibitem [{\citenamefont {Aizenman}\ \emph {et~al.}(1987)\citenamefont
  {Aizenman}, \citenamefont {Barsky},\ and\ \citenamefont
  {Fern{\'a}ndez}}]{aizenman1987phase}%
  \BibitemOpen
  \bibfield  {author} {\bibinfo {author} {\bibfnamefont {M.}~\bibnamefont
  {Aizenman}}, \bibinfo {author} {\bibfnamefont {D.~J.}\ \bibnamefont
  {Barsky}},\ and\ \bibinfo {author} {\bibfnamefont {R.}~\bibnamefont
  {Fern{\'a}ndez}},\ }\bibfield  {title} {\bibinfo {title} {The phase
  transition in a general class of ising-type models is sharp},\ }\href@noop {}
  {\bibfield  {journal} {\bibinfo  {journal} {Journal of Statistical Physics}\
  }\textbf {\bibinfo {volume} {47}},\ \bibinfo {pages} {343} (\bibinfo {year}
  {1987})}\BibitemShut {NoStop}%
\bibitem [{\citenamefont {Aizenman}\ \emph {et~al.}(2015)\citenamefont
  {Aizenman}, \citenamefont {Duminil-Copin},\ and\ \citenamefont
  {Sidoravicius}}]{aizenman2015random}%
  \BibitemOpen
  \bibfield  {author} {\bibinfo {author} {\bibfnamefont {M.}~\bibnamefont
  {Aizenman}}, \bibinfo {author} {\bibfnamefont {H.}~\bibnamefont
  {Duminil-Copin}},\ and\ \bibinfo {author} {\bibfnamefont {V.}~\bibnamefont
  {Sidoravicius}},\ }\bibfield  {title} {\bibinfo {title} {Random currents and
  continuity of ising model’s spontaneous magnetization},\ }\href@noop {}
  {\bibfield  {journal} {\bibinfo  {journal} {Communications in Mathematical
  Physics}\ }\textbf {\bibinfo {volume} {334}},\ \bibinfo {pages} {719}
  (\bibinfo {year} {2015})}\BibitemShut {NoStop}%
\bibitem [{\citenamefont {Aizenman}\ and\ \citenamefont
  {Duminil-Copin}(2021)}]{aizenman2021marginal}%
  \BibitemOpen
  \bibfield  {author} {\bibinfo {author} {\bibfnamefont {M.}~\bibnamefont
  {Aizenman}}\ and\ \bibinfo {author} {\bibfnamefont {H.}~\bibnamefont
  {Duminil-Copin}},\ }\bibfield  {title} {\bibinfo {title} {Marginal triviality
  of the scaling limits of critical 4d ising and
  $\lambda$$\backslash$phi\_4\^{}4 models},\ }\href@noop {} {\bibfield
  {journal} {\bibinfo  {journal} {Annals of Mathematics}\ }\textbf {\bibinfo
  {volume} {194}},\ \bibinfo {pages} {163} (\bibinfo {year}
  {2021})}\BibitemShut {NoStop}%
\bibitem [{\citenamefont {van Engelenburg}\ and\ \citenamefont
  {Lis}(2023)}]{van2023elementary}%
  \BibitemOpen
  \bibfield  {author} {\bibinfo {author} {\bibfnamefont {D.}~\bibnamefont {van
  Engelenburg}}\ and\ \bibinfo {author} {\bibfnamefont {M.}~\bibnamefont
  {Lis}},\ }\bibfield  {title} {\bibinfo {title} {An elementary proof of phase
  transition in the planar xy model},\ }\href@noop {} {\bibfield  {journal}
  {\bibinfo  {journal} {Communications in Mathematical Physics}\ }\textbf
  {\bibinfo {volume} {399}},\ \bibinfo {pages} {85} (\bibinfo {year}
  {2023})}\BibitemShut {NoStop}%
\bibitem [{\citenamefont {Diestel}\ and\ \citenamefont
  {Diestel}(2017)}]{diestel2017extremal}%
  \BibitemOpen
  \bibfield  {author} {\bibinfo {author} {\bibfnamefont {R.}~\bibnamefont
  {Diestel}}\ and\ \bibinfo {author} {\bibfnamefont {R.}~\bibnamefont
  {Diestel}},\ }\bibfield  {title} {\bibinfo {title} {{Extremal graph
  theory}},\ }\href@noop {} {\bibfield  {journal} {\bibinfo  {journal} {Graph
  theory}\ ,\ \bibinfo {pages} {173}} (\bibinfo {year} {2017})}\BibitemShut
  {NoStop}%
\bibitem [{\citenamefont {Friedli}\ and\ \citenamefont
  {Velenik}(2017)}]{velenik}%
  \BibitemOpen
  \bibfield  {author} {\bibinfo {author} {\bibfnamefont {S.}~\bibnamefont
  {Friedli}}\ and\ \bibinfo {author} {\bibfnamefont {Y.}~\bibnamefont
  {Velenik}},\ }\href {https://doi.org/10.1017/9781316882603} {\emph {\bibinfo
  {title} {{Statistical mechanics of lattice systems: a concrete mathematical
  introduction}}}}\ (\bibinfo  {publisher} {Cambridge University Press},\
  \bibinfo {year} {2017})\BibitemShut {NoStop}%
\end{thebibliography}%

\appendix
\onecolumngrid
\section{Change of Variables}
\label{app:bilayer}

\begin{lemma}
    \label{app-lem:var-map}
    Let $f:\dS^1 \times \dS^1 \mapsto \dR$ be bounded. Then
    \begin{equation}
        \int_{(-\pi,\pi)^2} \frac{d\phi_+}{2\pi}\frac{d\phi_-}{2\pi} f(\phi_+,\phi_-) =\int_{(-\pi,\pi)^2} \frac{d\theta}{2\pi}\frac{d\phi}{2\pi} f(\theta +\phi/2,\theta-\phi/2)
    \end{equation}
    Where $\theta= (\phi_+ +\phi_-)/2$ is the average phase and $\phi = \phi_+-\phi_-$ is the phase difference.
\end{lemma}
\begin{proof}
    For simplicity, we will drop the normalization $1/2\pi$. Indeed, notice that
    \begin{align}
        \int_{(-\pi,\pi)^2} d\phi_\pm f(\phi_+,\phi_-) &=\int_{-\pi}^\pi d\phi_- \int_{-\phi_-}^{2\pi -\phi_-} d\phi f(\phi_- +\phi, \phi_-), \quad \phi = \phi_+ -\phi_-\\
        &= \int_{-\pi}^\pi d\phi_- \int_{-\pi_-}^{\pi-} d\phi f(\phi_- +\phi, \phi_-) \\
        &= \int_{-\pi}^\pi d\phi \int_{-\pi}^\pi d\phi_- f(\phi_- +\phi,\phi_-) \\
        &= \int_{-\pi}^\pi d\phi \int_{-\pi+\phi/2}^{\pi+\phi/2} d\theta f(\theta +\phi/2,\theta-\phi/2), \quad \theta = \phi_- +\phi/2 \\
        &= \int_{(-\pi,\pi)^2} d\theta d\phi f(\theta +\phi/2,\theta-\phi/2)
    \end{align}
    Where the 2nd and 5th equality uses the fact that $\phi_\pm \mapsto f$ is $2\pi$-periodic, and thus the integration limits can be an arbitrarily chosen $2\pi$-interval.
\end{proof}

\section{(Random) Cluster Representation - Critical Regime}
\label{app:cluster}

In the main text, we have claimed that by choosing the edge probabilities $p_e^\tau, p_e^\sigma$ appropriately and defining the subgraphs $\hat{\tau},\hat{\sigma}$, the $\dZ_2$ and $U(1)$ spin-spin correlations can be mapped to percolation events.
By comparing the edge probabilities $p_e^\tau \ge p_e^\sigma$, the ordering of percolation events (and thus correlations) becomes apparent.
Therefore, in this section, we will follow Ref. \cite{dubedat2022random} and prove our claim of the correlation-percolation correspondence of the critical $U(1)\times  \dZ_2$ Hamiltonian in Eq. \eqref{eq:U1-Z2}.
As discussed in the main text, let $\dP[\sigma,\tau]$ be the probability distribution of the spin configurations defined by the $U(1)\times \dZ_2$ Hamiltonian in Eq. \eqref{eq:U1-Z2}, i.e.,
\begin{align}
    \dP[(\sigma,\tau)] &\propto \prod_{e=ij} w_e(\sigma,\tau) \\
    w_e(\sigma,\tau) &= e^{\beta \kappa_e (1+\tau_i\tau_j) \cos(\theta_i-\theta_j)}
\end{align}
\subsection{$\dZ_2$ Correlations}
\label{app:cluster-Z2}
Similar to that discussed in the main text for the standard Ising model (see Fig. \ref{fig:cluster}), for a fixed spin configuration $(\sigma,\tau)$, define a random variable of subgraphs $\hat{\tau}:E\to\{0,1\}$ via the edge probability
\begin{align}
    p_e^\tau (\sigma,\tau) &= \left[ 1-\frac{w_e(\sigma,\tau^i)}{w_e(\sigma,\tau)}\right] 1\{w_e(\sigma,\tau^i) < w_e(\sigma,\tau)\}\\
    &= \left[ 1-e^{-2\beta \kappa_e \tau_i \tau_j \cos (\theta_i -\theta_j)} \right] 1\{\tau_i \tau_j \cos (\theta_i -\theta_j)> 0\}
\end{align}
So that
\begin{align}
    \dP[\hat{\tau}|\sigma,\tau] &= \prod_{e\in \hat{\tau}} p_e^\tau(\sigma,\tau) \prod_{e\notin \hat{\tau}} (1-p_e^\tau (\sigma,\tau))\\
    \dP[\hat{\tau},\sigma,\tau] &= \dP[\hat{\tau}|\sigma,\tau] \dP[\sigma,\tau]
\end{align}

\begin{theorem}[\textbf{Cluster-Flip Invariance}]
    \label{app-thm:cluster-flip-Z2}
    Let $\dP_{G,\beta}$ be the probability distribution on $(\hat{\tau},\sigma,\tau)$ as befined previously on a finite graph $G$. Then $\dP$ is invariant under cluster flips with respect to $\tau$, i.e.,
    \begin{equation}
        \dP_{G,\beta}[ (\hat{\tau},\sigma,\tau)] =\dP_{G,\beta}[(\hat{\tau},\sigma,\tau^{C_0 (\hat{\tau})})]
    \end{equation}
    Where $C_0(\hat{\tau})$ is the cluster in $\hat{\tau}$ intersecting the lattice site $0$ and $\tau^{C_0(\hat{\tau})}$ denotes flipping the spins of $\tau$ only in $C_0(\hat{\tau})$.
    In particular, the $\dZ_2$ correlation-percolation is established, i.e.,
    \begin{equation}
        \bra \tau_0\tau_R \ket_{G,\beta} = \dP_{G,\beta}[0\lr_{\hat{\tau}} R]
    \end{equation}
\end{theorem}
\begin{proof}
    For notation simplicity, we shall omit the subscripts $G,\beta$.
    Notice that
    \begin{align}
        \frac{\dP [(\hat{\tau},\sigma,\tau^{C_0(\hat{\tau})})]}{\dP [(\hat{\tau},\sigma,\tau)]} &= \frac{\dP [\hat{\tau}|\sigma,\tau^{C_0(\hat{\tau})}]}{\dP [\hat{\tau}|\sigma,\tau]}\times  \frac{\dP [\sigma,\tau^{C_0(\hat{\tau})}]}{\dP [\sigma,\tau]}
    \end{align}
    Note that the ratios only depends on edge $e$ on the boundary of $C_0(\hat{\tau})$ since $w_e(\sigma,\tau)$ is invariant if both endpoints are flipped, i.e., $\tau_i,\tau_j \mapsto -\tau_i\tau_j$ where $e=ij$. 
    More specifically,
    \begin{align}
        \frac{\dP [(\hat{\tau},\sigma,\tau^{C_0(\hat{\tau})})]}{\dP [(\hat{\tau},\sigma,\tau)]} &= \prod_{e=ij\in \partial C_0(\hat{\tau}) \atop {i\in C_0(\hat{\tau}),j\notin C_0(\hat{\tau})}} \frac{p_e^\tau (\sigma,\tau^i) w_e (\sigma,\tau^i)}{p_e^\tau (\sigma,\tau) w_e (\sigma,\tau)}
    \end{align}
    By definition of $p_e^\tau$, it's straightforward to check that each term in the product is $=1$.
    The statement then follows.
\end{proof}

\subsection{$U(1)$ Correlations}
\label{app:cluster-U1}
The $U(1)$ correlation-percolation correspondence for the higher order term $\bra \cos 2(\theta_0 -\theta_R) \ket$ is much more difficult to establish than $\dZ_2$ correlation.; in fact, it relies on first establishing a correspondence for the conventional term $\bra \cos (\theta_0 -\theta_R)\ket$.
In this section, we shall follow Ref. \cite{dubedat2022random} (with slight modifications) in establishing the correspondence for the $U(1)\times \dZ_2$ model in Eq. \eqref{eq:U1-Z2}.
Similar to the $\dZ_2$ subgraph $\hat{\tau}$, let $\hat{\xi}:E \to\{0,1\}$ be constructed with edge probability
\begin{align}
    p_e^\xi (\sigma,\tau) &= \left[ 1-\frac{w_e(\sigma^{\xi,i},\tau)}{w_e(\sigma,\tau)}\right] 1\{w_e(\sigma^{\xi,i},\tau) < w_e(\sigma,\tau)\}\\
    &= \left[ 1-e^{-4\beta \kappa_e  \cos \theta_i \cos \theta_j} \right] 1\{\tau_i \tau_j =\xi_i \xi_j =1\}
\end{align}
Where $\xi$ is the sign of the $x$ component of $\sigma$ ($\cos\theta$), and $\sigma^{\xi,i}$ denotes flipping the spin only at the lattice site $i$ along the $x$ direction, i.e., $\xi_i\mapsto -\xi_i$ or $\theta_i \mapsto \pi -\theta_i$.
As before, this defines a probability distribution on $(\hat{\xi},\sigma,\tau)$.
Similarly, define the subgraph $\hat{\sigma}^{\eta}:E\to \{0,1\}$ by using the sign $\eta =\pm 1$ of the $y$ component of $\sigma$ ($\sin \theta$), i.e., the edge probability is
\begin{equation}
    p_e^\eta (\sigma,\tau) = \left[ 1-\frac{w_e(\sigma^{\eta,i},\tau)}{w_e(\sigma,\tau)}\right] 1\{w_e(\sigma^{\eta,i},\tau) < w_e(\sigma,\tau)\}
\end{equation}
And thus this defines a probability on $(\hat{\eta}, \sigma,\tau)$. 

As mentioned in the main text, to establish the correspondence for $\bra \cos 2(\theta_0-\theta_R)\ket$, we require both $\hat{\xi}, \hat{\eta}$ and thus the first question is then whether we can define a joint probability on $(\hat{\xi},\hat{\eta},\sigma,\tau)$.
Since each edge is independently established (conditional with respect to the fixed spin configuration $\sigma,\tau$), $\hat{\xi}_e,\hat{\eta}_e$ are Bernoulli random variables and thus whether we can define a joint distribution depends on choosing the appropriate correlation $c_e(\sigma,\tau)$ so that the following values are $\ge 0$, i.e.,
\begin{align}
    \label{eq:cluster-xi-eta-1}
    \dP[\hat{\xi}_e=1,\hat{\eta}_e =1|\sigma,\tau] &= c_e\\
    \label{eq:cluster-xi-eta-2}
    \dP[\hat{\xi}_e=0,\hat{\eta}_e =1|\sigma,\tau] &= p_e^\eta - c_e\\
    \label{eq:cluster-xi-eta-3}
    \dP[\hat{\xi}_e=1,\hat{\eta}_e =0|\sigma,\tau] &= p_e^\xi -c_e\\
    \label{eq:cluster-xi-eta-4}
    \dP[\hat{\xi}_e=0,\hat{\eta}_e =0|\sigma,\tau] &= 1-p_e^\xi -p_e^\eta +c_e
\end{align}
Moreover, we wish to choose the correlation $c_e$ so that a analogous cluster-flip invariance property is satisfied. 
Hence, it turns out that we should choose the correlation $c_e = p_e^\sigma$ as defined in the proof of the main result, Theorem \eqref{thm:cluster}, i.e.,
\begin{align}
p_e^{\sigma}(\sigma,\tau)&=\left[1 +\frac{w_e(\sigma^{\xi\eta,i},\tau)}{w_e(\sigma,\tau)}- \frac{w_e(\sigma^{\xi,i},\tau)}{w_e(\sigma,\tau)}-\frac{w_e(\sigma^{\eta,i},\tau)}{w_e(\sigma,\tau)} \right]\times  1\{w_e(\sigma^{\xi,i},\tau), w_e(\sigma^{\eta,i},\tau) < w_e(\sigma,\tau) \} 
\end{align}
Indeed, let us verify that this choice of correlation satisfies the necessary properties.

\begin{theorem}[\textbf{Cluster-Flip Invariance}]
    \label{app-thm:cluster-flip-U1}
    Let $\dP_{G,\beta}$ denote the joint probability on $(\hat{\xi},\hat{\eta},\sigma,\tau)$ on a finite graph $G$ defined previously. Then $\dP_{G,\beta}$ is well-defined and satisfies the cluster-flip invariance with respect to $\xi,\eta$, i.e.,
    \begin{align}
        \dP_{G,\beta}[ (\hat{\xi},\hat{\eta},\sigma,\tau)] &=\dP_{G,\beta}[ (\hat{\xi},\hat{\eta},\sigma^{\xi,C_0(\hat{\xi})},\tau)] \\
        \dP_{G,\beta}[ (\hat{\xi},\hat{\eta},\sigma,\tau)] &=\dP_{G,\beta}[ (\hat{\xi},\hat{\eta},\sigma^{\eta,C_0(\hat{\eta})},\tau)]
    \end{align}
    Where $C_0(\hat{\xi})$ is the cluster in $\hat{\xi}$ intersecting the lattice site $0$ and $\sigma^{\xi,C_0(\hat{\xi})}$ denotes flipping the spin configuration $\sigma$ only at lattice sites within the cluster $C_0(\hat{\xi})$ along the $x$ component, i.e., $\xi_i \mapsto -\xi_i$ or $\theta_i\mapsto \pi -\theta_i$ for $i\in C_0(\hat{\xi})$. The notation is similar for $\xi\mapsto \eta$.
\end{theorem}
\begin{proof}
    The proof follows that given in Ref. \cite{dubedat2022random}, though we simplify/modify some parts to illuminate some of the situation.
    For the $U(1)\times \dZ_2$ Hamiltonian in Eq. \eqref{eq:U1-Z2}, the appropriate probability $p_e^\sigma$ happens to satisfy
    \begin{equation}
        p_e^\sigma = p_e^\xi \times p_e^\eta
    \end{equation}
    And thus conditions \eqref{eq:cluster-xi-eta-1}-\eqref{eq:cluster-xi-eta-4} are easily seen to satisfy. For general weights $w_e$ (as we will see for the general regime in Appendix \eqref{app:cluster-general}), this is not as simple \cite{dubedat2022random}.
    Let us now show that the cluster flip invariance holds true along $\xi$. The proof is similarly applied to that along $\eta$.
    Notice that
    \begin{align}
        \frac{\dP [(\hat{\xi},\hat{\eta},\sigma^{\xi,C_0(\hat{\xi})},\tau)]}{\dP [(\hat{\xi},\hat{\eta},\sigma,\tau)]} &= \frac{\dP [\hat{\xi},\hat{\eta}|\sigma^{\xi,C_0(\hat{\xi})},\tau]}{\dP [\hat{\xi},\hat{\eta}|\sigma,\tau]}\times  \frac{\dP [\sigma^{\xi,C_0(\hat{\xi})},\tau]}{\dP [\sigma,\tau]}
    \end{align}
    Similar to the proof in Theorem \eqref{app-thm:cluster-flip-Z2}, we see that if an edge $e=ij$ has endpoints $i,j$ entirely in $C_0(\hat{\xi})$ or entirely outside of the cluster, then the edge weights $w_e(\sigma,\tau),w_e(\sigma^{\xi,i},\tau),w_e(\sigma^{\eta,i},\tau)$ and $w_e(\sigma^{\xi\eta,i},\tau)$ are all invariant under $\sigma \mapsto \sigma^{\xi,C_0(\hat{\xi})}$. 
    Therefore, we only need to consider edges in the boundary of $C_0(\hat{\xi})$ so that $e=ij$ with $i\in C_0(\hat{\xi})$ and $j\notin C_0(\hat{\xi})$, i.e.,
    \begin{align}
        \frac{\dP [(\hat{\xi},\hat{\eta},\sigma^{\xi,C_0(\hat{\xi})},\tau)]}{\dP [(\hat{\xi},\hat{\eta},\sigma,\tau)]} &= \prod_{e=ij\in \partial C_0(\hat{\xi})\atop i\in C_0(\hat{\xi}),j\notin C_0(\hat{\xi})} \frac{\dP[\hat{\xi}_e=0,\hat{\eta}_e|\sigma^{\xi, C_0(\hat{\xi})},\tau]}{\dP[\hat{\xi}_e=0,\hat{\eta}_e|\sigma,\tau]} \frac{w_e (\sigma^{\xi,i},\tau)}{w_e (\sigma,\tau)}
    \end{align}
    Where we used the fact that since $e\in \partial C_0(\hat{\xi})$, it cannot be in  $\hat{\xi}$ (the cluster $C_0(\hat{\xi})$ stops growing at its boundary). 
    However, the value of $\hat{\eta}_e$ could be either $0$ or $1$ and thus we must consider all possible cases, i.e., those described by Eqs. \eqref{eq:cluster-xi-eta-2} and \eqref{eq:cluster-xi-eta-4}.
    More specifically, it is sufficient to prove that the following term
    \begin{equation}
        w_e (\sigma,\tau)\dP[\hat{\xi}_e=0,\hat{\eta}_e|\sigma,\tau]
    \end{equation}
    Is invariant under spin-flip $\sigma \mapsto \sigma^{\xi,i}$ where $i$ is an endpoint of the edge $e\in \partial C_0(\hat{\xi})$.
    
    For notation simplicity, let me fix an edge $e=ij$ in the boundary for the remainder of this proof so that we can omit the subscript. 
    We shall also write
    \begin{align}
        w &\equiv w_e (\sigma,\tau) \\
        w^\xi &\equiv w_e(\sigma^{\xi,i},\tau) \\
        w^\eta &\equiv w_e(\sigma^{\xi,i},\tau) \\
        w^{\xi\eta} &\equiv w_e(\sigma^{\xi\eta,i},\tau)
    \end{align}
    And similarly for $p^\xi, p^\eta,p^\sigma$ so that
    \begin{align}
        w p^\xi &= \left[w-w^\xi \right]  1\{w^\xi < w\} \\
        w p^\eta &= \left[w-w^\eta \right]  1\{w^\eta < w\} \\
        w p^\sigma &= \left[w+w^{\xi\eta}-w^\xi -w^\eta \right]  1\{w^\xi,w^\eta < w\}
    \end{align}
    In fact, it is more illuminating if we write the previous equations in the following form
    \begin{align}
        w (1-p^\xi) &= \min (w,w^\xi) \\
        w (1-p^\eta) &= \min (w,w^{\eta}) 
    \end{align}
    Note that under spin-flip $\sigma \mapsto \sigma^{\xi,i}$ (where $i$ is one endpoint of the fixed edge $e$), we have
    \begin{align}
        w,w^\xi, w^\eta, w^{\xi\eta} &\mapsto w^\xi, w, w^{\xi\eta}, w^{\eta}
    \end{align}
    And thus $w(1-p^\xi)$ is invariant under the spin-flip. 
    If we compare the 2 possible cases (corresponding to $\hat{\eta}=1,0$) in Eq. \eqref{eq:cluster-xi-eta-2} and \eqref{eq:cluster-xi-eta-4}, we see that they differ by the invariant variable $w(1-p^\xi)$ and thus we only need to consider one of the two cases, say that corresponding to Eq. \eqref{eq:cluster-xi-eta-2}, i.e.,
    \begin{equation}
        w(p^\eta -p^\sigma) = [w-w^\eta] 1\{w^\eta <w <w^\xi\} +[w^{\xi} - w^{\xi\eta}] 1\{w^{\xi},w^{\eta} <w\}
    \end{equation}
    Under spin flip $\sigma \mapsto \sigma^{\xi,i}$, we have
    \begin{equation}
        w(p^\eta -p^\sigma) \mapsto  [w^{\xi}-w^{\xi\eta}] 1\{w^{\xi\eta} <w^{\xi} <w\} +[w - w^{\eta}] 1\{w,w^{\xi\eta} <w^\xi\}
    \end{equation}
    Note that the condition $w^{\xi \eta} <w^\xi$ is equivalent to $w^{\eta} <w$ and thus we can write
    \begin{align}
        w(p^\eta -p^\sigma) &\mapsto [w^{\xi}-w^{\xi\eta}] 1\{w^{\xi},w^{\eta} <w \} +[w - w^{\eta}] 1\{w^{\eta} <w <w^\xi\} \\
        &= w(p^\eta -p^\sigma)
    \end{align}
    Hence, the term $w(p^\eta -p^\sigma)$ is invariant under spin-flip $\sigma \mapsto \sigma^{\xi,i}$, and thus the statement follows.
    
\end{proof}

\begin{theorem}
    \label{app-thm:corr-U1}
    Let $\dP_{G,\beta}$ denote the joint probability on $(\hat{\xi},\hat{\eta},\sigma,\tau)$ on a finite graph $G$ defined previously. Then there exists constants $c>0$ depending only on $\beta$ and the degree (number of nearest neighbors) of the lattice sites $0,R$ in $G$ such that
    \begin{equation}
        c\dP_{G,\beta}[ 0\lr_{\hat{\xi}\hat{\eta}} R] \le \bra \cos 2(\theta_0 -\theta_R)\ket_{G,\beta} \le 2\dP_{G,\beta}[ 0\lr_{\hat{\xi}\hat{\eta}} R]
    \end{equation}
    Where $\hat{\xi}\hat{\eta}$ is the intersection of the two subgraphs $\hat{\xi},\hat{\eta}$ (edge-wise multiplication when viewed as a map $E\mapsto \{0,1\}$). Moreover, if $\xi,\eta$ denote the sign of the $x,y$ components of the XY spins $\sigma =e^{i\theta}$, then
    \begin{equation}
        \bra \xi_0 \eta_0 \xi_R \eta_R \ket_{G,\beta} = \dP_{G,\beta}[ 0\lr_{\hat{\xi}\hat{\eta}} R]
    \end{equation}
\end{theorem}

Before starting the proof, we remark that the percolation event $\{0\lr_{\hat{\xi}\hat{\eta}} R \}$ only depends on each subgraph $\hat{\xi},\hat{\eta}$ implicitly through their intersection, and thus if we were to fix a given subgraph $\hat{\sigma}$ and integrate the probability $\dP$ over all $\hat{\xi},\hat{\eta}$ which have an intersection $\hat{\xi}\hat{\eta}=\hat{\sigma}$, we would obtain a probability distribution on $(\hat{\sigma},\sigma,\tau)$. 
This is equal to that constructed by using the edge probability $p_e^\sigma$ as done in the main text. 
The ``hidden" parameters $\hat{\xi},\hat{\eta}$ were necessary to establish the cluster-flip property (and thus the correlation-percolation correspondence) but is not necessary in defining $\dP [0\lr_{\hat{\sigma}} R]$.
Hence, 
\begin{equation}
    \bra \cos 2(\theta_0 -\theta_R)\ket_{G,\beta} \cong \dP_{G,\beta}[ 0\lr_{\hat{\sigma}} R]
\end{equation}

We also note that for the intent of this paper (in which we prove the absence of a floating phase) the upper bound is sufficient and much simpler. 
However, for the sake of completeness, we will also prove the lower bound.
\begin{proof}[Proof of Upper Bound]
    The proof follows that given in Ref. \cite{dubedat2022random}, though we modify it to fit our consideration of free boundary conditions (where as the proof in Ref. \cite{dubedat2022random} considered fixed boundary conditions).
    For notation simplicity, we shall omit the subscripts $G,\beta$ unless otherwise stated.
    For the $U(1)\times \dZ_2$ Hamiltonian in Eq. \eqref{eq:U1-Z2}, we have
    \begin{align}
        \bra \cos 2(\theta_0 -\theta_R) \ket &= \bra \cos 2 \theta_0 \cos2\theta_R \ket +\bra \sin 2\theta_0 \sin 2\theta_R\ket \\
        &= 2\bra \sin 2\theta_0 \sin 2\theta_R\ket \\
        &= 2^3 \bra \sin \theta_0 \cos \theta_0 \sin \theta_R \cos \theta_R \ket
    \end{align}
    Using our constructed joint probability $\dP$, we find that
    \begin{align}
        \bra \sin \theta_0 \cos \theta_0 \sin \theta_R \cos \theta_R \ket &= \dE [\sin \theta_0 \cos \theta_0 \sin \theta_R \cos \theta_R (1\{0\lr_{\hat{\xi}} R\} +1\{0\not\lr_{\hat{\xi}} R\})] \\
        &= \dE [\sin \theta_0 \cos \theta_0 \sin \theta_R \cos \theta_R 1\{0\lr_{\hat{\xi}} R\}]
    \end{align}
    Where we used the fact that $\dP$ is invariant under cluster-spin flip and thus the second term is invariant under $\cos_0 \mapsto -\cos \theta_0$ (and keep the other terms invariant).
    Hence, the second term must be $=0$.
    Similarly, we have
    \begin{align}
        \dE [\sin \theta_0 \cos \theta_0 \sin \theta_R \cos \theta_R 1\{0\lr_{\hat{\xi}} R\}] &= \dE [\sin \theta_0 \cos \theta_0 \sin \theta_R \cos \theta_R 1\{0\lr_{\hat{\xi}} R\} (1\{0\lr_{\hat{\eta}} R\} +1\{0\not\lr_{\hat{\eta}} R\})] \\
        &= \dE [\sin \theta_0 \cos \theta_0 \sin \theta_R \cos \theta_R 1\{0\lr_{\hat{\xi}\hat{\eta}} R\}]
    \end{align}
    Since the term $\sin \theta_0 \cos \theta_0 \sin \theta_R \cos \theta_R \le 1/4$, the upper bounded with the extra factor of $2$ follows.
    Note that the same argument shows that
    \begin{equation}
        \bra \xi_0 \eta_0 \xi_R \eta_R \ket = \dP[ 0\lr_{\hat{\xi}\hat{\eta}} R]
    \end{equation}
\end{proof}

\begin{proof}[Proof of Lower Bound]
    Continuing the process in the proof of the upper bound, let $N_0$ denote the neighboring lattice sites of site 0 and let $S_0 \subseteq N_0$ denote the subset of lattice sites $i\in N_0$ which are connected to lattice site $R$ in $\hat{\sigma}$, if we remove all the edge adjacent to 0.
    Since $S_0$ only depends on $\hat{\sigma}$ implicitly through its value on edges not adjacent to 0, we write $S_0(\hat{\sigma}_{\not\sim 0})$.
    Note that $0\lr_{\hat{\sigma}} R$ if and only if there exist nonempty $\emptyset \ne S \subseteq N_0$ such that $S_0(\hat{\sigma}) = S$ and that $0$ is connected to $S$ within $\hat{\sigma}$, restricted on edges $e$ adjacent to $\sim 0$ (which we denote by $0\lr_{\hat{\sigma}|\sim 0} R$).
    Therefore,
    \begin{align}
        \dE [\sin \theta_0 \cos \theta_0 \sin \theta_R \cos \theta_R 1\{0\lr_{\hat{\sigma}} R\}] &=  \sum_{\emptyset \ne S\subseteq N_0} \dE[\sin \theta_0 \cos \theta_0 \sin \theta_R \cos \theta_R 1\{S_0(\hat{\sigma}_{\not\sim 0}) = S\} 1\{0\lr_{\hat{\sigma}|\sim 0} S\}]\\
        &\ge \sum_{\emptyset \ne S\subseteq N_0} \dE[\sin \theta_0 \cos \theta_0 \sin \theta_R \cos \theta_R 1\{S_0(\hat{\sigma}_{\not\sim 0}) = S\} 1\{0\lr_{\hat{\sigma}|\sim 0} S\} \\
        &\quad\quad\quad \times 1\{\pi/8 < \theta_0,\theta_R < 3\pi/8\}] \\
        &\ge 1/8 \sum_{\emptyset \ne S\subseteq N_0}\dP[S_0(\hat{\sigma}_{\not\sim 0}) = S, 0\lr_{\hat{\sigma}|\sim 0} S, \pi/8 < \theta_0,\theta_R < 3\pi/8]
    \end{align}
    Notice that if we fix the spin values on $N_0$, then we have effectively partitioned the graph structure into decoupled systems consisting of edges adjacent to $\sim 0$ and those not adjacent $\sim 0$ (the fixed spin value act as boundary conditions of the two partitions). 
    Therefore, the conditional probability of fixing spins on $N_x$ is given by
    \begin{align}
        \dP[S_0(\hat{\sigma}_{\not\sim 0}) = S, 0\lr_{\hat{\sigma}|\sim 0} S, \pi/8 < \theta_0,\theta_R < 3\pi/8|\theta_{N_0}] &= \dP[S_0(\hat{\sigma}_{\not\sim 0}) = S, \pi/8 < \theta_R < 3\pi/8|\theta_{N_0}]\\
        &\quad\quad\times \dP[0\lr_{\hat{\sigma}|\sim 0} S, \pi/8 < \theta_0 < 3\pi/8|\theta_{N_0}]
    \end{align}
    It's straightforward \cite{dubedat2022random} to check that there exists some constant $\delta(\beta, \deg 0)>0$ depending on $\beta$ and the degree of site $0$ such that
    \begin{equation}
        \dP[0\lr_{\hat{\sigma}|\sim 0} S, \pi/8 < \theta_0 < 3\pi/8|\theta_{N_0}] \ge \delta (\beta, \deg 0)
    \end{equation}
    Indeed, this probability corresponds to the finite system consisting of edges $e$ adjacent to site 0 with boundary conditions $\theta_{N_0}$. 
    Therefore,
    \begin{align}
        \dP[S_0(\hat{\sigma}_{\not\sim 0}) = S, 0\lr_{\hat{\sigma}|\sim 0} S, \pi/8 < \theta_0,\theta_R < 3\pi/8|\theta_{N_0}] &\ge \delta (\beta, \deg 0) \times \dP[S_0(\hat{\sigma}_{\not\sim 0}) = S, \pi/8 < \theta_R < 3\pi/8|\theta_{N_0}]
    \end{align}
    Integrating over all fixed spins $\theta_{N_0}$ and substituting back, we find that
    \begin{equation}
        \dE [\sin \theta_0 \cos \theta_0 \sin \theta_R \cos \theta_R 1\{0\lr_{\hat{\sigma}} R\}] \ge \frac{1}{8} \delta (\beta, \deg 0) \times\dP[0\lr_{\hat{\sigma}} R, \pi/8<\theta_R <3\pi/8]
    \end{equation}
    Repeat the argument for the lattice site $R$ to obtain
    \begin{equation}
        \dE [\sin \theta_0 \cos \theta_0 \sin \theta_R \cos \theta_R 1\{0\lr_{\hat{\sigma}} R\}] \ge \frac{1}{8} \delta (\beta, \deg 0)\delta (\beta, \deg R)\times \dP[0\lr_{\hat{\sigma}} R]
    \end{equation}
    Hence, the lower bound follows.
\end{proof}
\subsection{Relation between Percolation Events}
In the proof of Theorem \eqref{thm:cluster}, we noted that the key observation is that $p^\tau_e (\sigma,\tau) \ge p^\sigma_e (\sigma,\tau)$ for all spin configurations $(\sigma,\tau)$. Here we provide a short proof of the statement.
\begin{theorem}
    \label{app-thm:cluster-relation}
    Let $p_e^\tau,p_e^\sigma$ be defined with respect to the critical $U(1)\times \dZ_2$ Hamiltonian in Eq. \eqref{eq:U1-Z2} as shown in Appendices \eqref{app:cluster-U1} and \eqref{app:cluster-Z2}. Then
    \begin{equation}
        p_e^\tau (\sigma,\tau) \ge p_e^\sigma (\sigma,\tau)
    \end{equation}
    For all spin configurations $(\sigma,\tau)$.
\end{theorem}
\begin{proof}
    Note we can rewrite the edge probabilities as
    \begin{align}
        \label{eq:cluster-edge-tau-explicit}
        p^\tau_e &= [1-e^{-2\beta \kappa_e \cos (\theta_i -\theta_j)}] 1\{\tau_i\tau_j  \cos (\theta_i-\theta_j) > 0\} \\
        \label{eq:cluster-edge-sigma-explicit}
        p_e^\sigma &= [1+e^{4\beta \kappa_e \cos (\theta_i-\theta_j)} -e^{-4\beta \kappa_e \sin \theta_i \sin \theta_j}-e^{-4\beta \kappa_e \cos \theta_i \cos \theta_j}] 1\{\tau_i\tau_j = \xi_i\xi_j =\eta_i\eta_j = +1\}
    \end{align}
    Where $\xi,\eta = \pm 1$ are the signs of the $x,y$ components of the spin $\sigma$. 
    Notice that if the condition in $p_e^\sigma$ is not satisfied, i.e., we do not have $\tau_i\tau_j =\xi_i\xi_j =\eta_i\eta_j =1$, then $p_e^\sigma=0$ and must be trivially $\le p_e^\tau$.
    Hence, we shall consider the case where the condition is satisfied.
    In this case, we see that the condition for $p_e^\tau$ is also satisfied, i.e., $\tau_i \tau_j \cos (\theta_i -\theta_j) >0$.
    Hence,
    \begin{align}
        p_e^\tau -p_e^\sigma &= [e^{-4\beta \kappa_e \sin \theta_i \sin \theta_j}+e^{-4\beta \kappa_e \cos \theta_i \cos \theta_j}] -e^{-2\beta \kappa_e \cos (\theta_i -\theta_j)} [1+e^{-2\beta \kappa_e \cos (\theta_i -\theta_j)}] \\
        &=e^{-2\beta \kappa_e \cos (\theta_i -\theta_j)} [2 \cosh (2\beta \kappa_e \cos (\theta_i+\theta_j)) -1-e^{-2\beta \kappa_e \cos (\theta_i -\theta_j)}] \\
        &\ge 0
    \end{align}
    Where we used the fact that $\cos (\theta_i-\theta_j) >0$ and that $\cosh \ge 1$.
\end{proof}

\section{(Random) Cluster Representation - General Regime}
\label{app:cluster-general}

In the section, we will extend the arguments of Appendix \eqref{app:cluster} to the general $U(1)\times \dZ_2$ Hamiltonian in Eq. \eqref{eq:U1-Z2-general}.
As discussed in the main text, let $\dP[\sigma,\tau]$ be the probability distribution of the spin configurations defined by the $U(1)\times \dZ_2$ Hamiltonian in Eq. \eqref{eq:U1-Z2}, i.e.,
\begin{align}
    \dP[(\sigma,\tau)] &\propto \prod_{e=ij} w_e(\sigma,\tau) \\
    \label{eq:edge-weight-general}
    w_e(\sigma,\tau) &= e^{\beta [\kappa_e (1+\tau_i\tau_j) \cos(\nabla_e \theta) +\lambda_e \tau_i\tau_j +\alpha_e \cos (2\nabla_e \theta)]}
\end{align}

\subsection{$\dZ_2$ Correlations}
\label{app:cluster-Z2-general}
Similar to that discussed in the main text for the standard Ising model (see Fig. \ref{fig:cluster}) and the critical regime in Appendix \eqref{app:cluster}, for a fixed spin configuration $(\sigma,\tau)$, define a random variable of subgraphs $\hat{\tau}:E\to\{0,1\}$ via the edge probability
\begin{align}
    p_e^\tau (\sigma,\tau) &= \left[ 1-\frac{w_e(\sigma,\tau^i)}{w_e(\sigma,\tau)}\right] 1\{w_e(\sigma,\tau^i) < w_e(\sigma,\tau)\}\\
    &= \left[ 1-e^{-2\beta \tau_i \tau_j (\kappa_e \cos (\nabla_e \theta) +\lambda_e)} \right] 1\{\tau_i \tau_j (\kappa_e \cos (\nabla_e \theta) +\lambda_e) > 0\}
\end{align}
So that
\begin{align}
    \dP[\hat{\tau}|\sigma,\tau] &= \prod_{e\in \hat{\tau}} p_e^\tau(\sigma,\tau) \prod_{e\notin \hat{\tau}} (1-p_e^\tau (\sigma,\tau))\\
    \dP[\hat{\tau},\sigma,\tau] &= \dP[\hat{\tau}|\sigma,\tau] \dP[\sigma,\tau]
\end{align}

\begin{theorem}[\textbf{Cluster-Flip Invariance}]
    \label{app-thm:cluster-flip-Z2-general}
    Let $\dP_{G,\beta}$ be the probability distribution on $(\hat{\tau},\sigma,\tau)$ as befined previously on a finite graph $G$. Then $\dP$ is invariant under cluster flips with respect to $\tau$, i.e.,
    \begin{equation}
        \dP_{G,\beta}[ (\hat{\tau},\sigma,\tau)] =\dP_{G,\beta}[(\hat{\tau},\sigma,\tau^{C_0 (\hat{\tau})})]
    \end{equation}
    Where $C_0(\hat{\tau})$ is the cluster in $\hat{\tau}$ intersecting the lattice site $0$ and $\tau^{C_0(\hat{\tau})}$ denotes flipping the spins of $\tau$ only in $C_0(\hat{\tau})$.
    In particular, the $\dZ_2$ correlation-percolation is established, i.e.,
    \begin{equation}
        \bra \tau_0\tau_R \ket_{G,\beta} = \dP_{G,\beta}[0\lr_{\hat{\tau}} R]
    \end{equation}
\end{theorem}
\begin{proof}
    Same as that in Theorem \eqref{app-thm:cluster-flip-Z2}.
\end{proof}
\subsection{$U(1)$ Correlations}
\label{app:cluster-U1-general}

Similar to Appendix \eqref{app:cluster-U1}, we can define $p_e^\xi, p_e^\eta,p_e^\sigma$. However, this time (as discussed in the main text in Sec. \eqref{sec:general}), we shall choose the sector where $\tau_i\tau_j = +1$, so that effectively, the edge weights in Eq. \eqref{eq:edge-weight-general} become
\begin{equation}
    \tilde{w}_e(\sigma) = \exp{\left[2\beta \alpha_e \left(\cos \nabla_e \theta + \frac{\kappa_e}{2\alpha_e}\right)^2 + \text{const}_e\right]} \equiv \rho_e (\cos \nabla_e \theta)
\end{equation}
In the parameter regime where $\kappa_e\ge 2\alpha_e$, we see that $x\mapsto \rho_e (x)$ is convex and increasing within the interval $[-1,1]$.
Note that this is the exact condition used in Ref. \cite{dubedat2022random} (albeit they are considering a model with only XY spins), and thus theorems analogous to those in Appendix \eqref{app:cluster} of the critical regime can also be proven.

To be concrete, let us define
\begin{align}
p_e^\xi (\sigma,\tau) &= \left[ 1-\frac{\tilde{w}_e(\sigma^{\xi,i})}{\tilde{w}_e(\sigma)}\right] 1\{\tilde{w}_e(\sigma^{\xi,i}) < \tilde{w}_e(\sigma), \tau_i\tau_j = +1\}\\
p_e^\eta (\sigma,\tau) &= \left[ 1-\frac{\tilde{w}_e(\eta^{\xi,i})}{\tilde{w}_e(\sigma)}\right] 1\{\tilde{w}_e(\sigma^{\eta,i}) < \tilde{w}_e(\sigma), \tau_i\tau_j = +1\}\\
p_e^{\sigma}(\sigma,\tau)&=\left[1 +\frac{\tilde{w}_e(\sigma^{\xi\eta,i})}{\tilde{w}_e(\sigma)}- \frac{\tilde{w}_e(\sigma^{\xi,i})}{\tilde{w}_e(\sigma)}-\frac{\tilde{w}_e(\sigma^{\eta,i})}{\tilde{w}_e(\sigma)} \right]\times  1\{\tilde{w}_e(\sigma^{\xi,i}), \tilde{w}_e(\sigma^{\eta,i}) < \tilde{w}_e(\sigma), \tau_i\tau_j= +1 \} 
\end{align}
Then define the probability distribution $\dP[(\hat{\xi},\hat{\eta},\sigma,\tau)]$ as done in Eq. \eqref{eq:cluster-xi-eta-1}-\eqref{eq:cluster-xi-eta-4}.

\begin{theorem}[\textbf{Cluster-Flip Invariance}]
    \label{{app-thm:cluster-flip-U1}-general}
    Let $\dP_{G,\beta}$ denote the joint probability on $(\hat{\xi},\hat{\eta},\sigma,\tau)$ on a finite graph $G$ defined previously. Then $\dP_{G,\beta}$ is well-defined and satisfies the cluster-flip invariance with respect to $\xi,\eta$, i.e.,
    \begin{align}
        \dP_{G,\beta}[ (\hat{\xi},\hat{\eta},\sigma,\tau)] &=\dP_{G,\beta}[ (\hat{\xi},\hat{\eta},\sigma^{\xi,C_0(\hat{\xi})},\tau)] \\
        \dP_{G,\beta}[ (\hat{\xi},\hat{\eta},\sigma,\tau)] &=\dP_{G,\beta}[ (\hat{\xi},\hat{\eta},\sigma^{\eta,C_0(\hat{\eta})},\tau)]
    \end{align}
    Where $C_0(\hat{\xi})$ is the cluster in $\hat{\xi}$ intersecting the lattice site $0$ and $\sigma^{\xi,C_0(\hat{\xi})}$ denotes flipping the spin configuration $\sigma$ only at lattice sites within the cluster $C_0(\hat{\xi})$ along the $x$ component, i.e., $\xi_i \mapsto -\xi_i$ or $\theta_i\mapsto \pi -\theta_i$ for $i\in C_0(\hat{\xi})$. The notation is similar for $\xi\mapsto \eta$.
\end{theorem}
\begin{proof}
    The proof is similar to that of Theorem \eqref{app-thm:cluster-flip-U1}, and thus we will only outline the differences.
    Using the fact that $x\mapsto \rho_e(x)$ is convex and increasing in the parameter regime $\kappa_e \ge 2\alpha_e$, we see that $\dP[(\hat{\xi},\hat{\eta},\sigma,\tau)]$ is well-defined \cite{dubedat2022random}.
    As before, it is sufficient to prove that the following term
    \begin{equation}
        w_e(\sigma,\tau) \dP[\hat{\xi}_e=0,\hat{\eta}_e|\sigma,\tau]
    \end{equation}
    Is invariant under spin-flip $\sigma \mapsto \sigma^{\xi,i}$ where $i$ is an endpoint of the edge $e\in \partial C_0(\hat{\xi})$.
    Note that since we defined $p_e^\xi,p_e^\eta,p_e^\sigma$ to be nonzero only in the sector $\tau_i\tau_j=+1$, we see that we can assume $\tau_i\tau_j= +1$ without loss of generality.
    Hence, it's sufficient to show that
    \begin{equation}
        \tilde{w}_e(\sigma) \dP[\hat{\xi}_e=0,\hat{\eta}_e|\sigma,\tau_i\tau_j=+1]
    \end{equation}
    Is invariant under spin-flip $\sigma \mapsto \sigma^{\xi,i}$ where $i$ is an endpoint of the edge $e=ij\in \partial C_0(\hat{\xi})$.

    We also note that in the proof of Theorem \eqref{app-thm:cluster-flip-U1}, we used the fact that $w^{\xi \eta} < w^\xi$ is equivalent to $w^{\eta} <w$. 
    In our case, this is equivalent to stating that the following 2 are equivalent, i.e.,
    \begin{equation}
        \rho_e (-X-Y) < \rho_e (-X+Y) \Leftrightarrow \rho_e (X-Y) < \rho_e (X+Y)
    \end{equation}
    This follows from the fact that $\rho_e$ is a strictly increasing function. Indeed, if $\rho_e (-X-Y) < \rho_e (-X+Y)$, then $-X-Y < -X+Y$ and thus $X-Y<X+Y$, which further implies $\rho_e (X-Y) <\rho_e(X+Y)$.
    The converse is similarly shown.
    The remainder of the proof is exactly the same as in Theorem \eqref{app-thm:cluster-flip-U1}.
\end{proof}

\begin{theorem}
    Let $\dP_{G,\beta}$ denote the joint probability on $(\hat{\xi},\hat{\eta},\sigma,\tau)$ on a finite graph $G$ defined previously. Then there exists constants $c>0$ depending only on $\beta$ and the degree (number of nearest neighbors) of the lattice sites $0,R$ in $G$ such that
    \begin{equation}
        c\dP_{G,\beta}[ 0\lr_{\hat{\xi}\hat{\eta}} R] \le \bra \cos 2(\theta_0 -\theta_R)\ket_{G,\beta} \le 2\dP_{G,\beta}[ 0\lr_{\hat{\xi}\hat{\eta}} R]
    \end{equation}
    Where $\hat{\xi}\hat{\eta}$ is the intersection of the two subgraphs $\hat{\xi},\hat{\eta}$ (edge-wise multiplication when viewed as a map $E\mapsto \{0,1\}$). Moreover, if $\xi,\eta$ denote the sign of the $x,y$ components of the XY spins $\sigma =e^{i\theta}$, then
    \begin{equation}
        \bra \xi_0 \eta_0 \xi_R \eta_R \ket_{G,\beta} = \dP_{G,\beta}[ 0\lr_{\hat{\xi}\hat{\eta}} R]
    \end{equation}
\end{theorem}
\begin{proof}
    The proof is exactly the same as Theorem \eqref{app-thm:corr-U1}
\end{proof}
\subsection{Relation between Percolation Events}
In the proof of Theorem \eqref{thm:cluster-general}, we noted that the key observation is that $p^\tau_e (\sigma,\tau) \ge p^\sigma_e (\sigma,\tau)$ for all spin configurations $(\sigma,\tau)$. Here we provide a short proof of the statement.

\begin{theorem}
    \label{app-thm:cluster-relation-general}
    Let $p_e^\tau,p_e^\sigma$ be defined with respect to the general $U(1)\times \dZ_2$ Hamiltonian in Eq. \eqref{eq:U1-Z2-general} as shown in Appendices \eqref{app:cluster-U1-general} and \eqref{app:cluster-Z2-general}. Then
    \begin{equation}
        p_e^\tau (\sigma,\tau) \ge p_e^\sigma (\sigma,\tau)
    \end{equation}
    For all spin configurations $(\sigma,\tau)$.
\end{theorem}
\begin{proof}
    Note we can rewrite the edge probabilities as
    \begin{align}
        \label{eq:cluster-edge-tau-explicit-general}
        p^\tau_e &= [1-e^{-2\beta (\kappa_e \cos (\nabla_e \theta) +\lambda_e)}] 1\{\tau_i\tau_j  (\kappa_e \cos (\nabla_e \theta) +\lambda_e) > 0\} \\
        \label{eq:cluster-edge-sigma-explicit-general}
        p_e^\sigma &= [1+e^{-4\beta \kappa_e \cos (\nabla_e \theta)} -(e^{-4\beta \kappa_e \sin \theta_i \sin \theta_j}+e^{-4\beta \kappa_e \cos \theta_i \cos \theta_j})e^{-2\beta \alpha \sin 2\theta_i \sin 2\theta_j}] \nonumber\\
        &\quad\times 1\{\tau_i\tau_j = \xi_i\xi_j =\eta_i\eta_j = +1\}
    \end{align}
    Where $\xi,\eta = \pm 1$ are the signs of the $x,y$ components of the spin $\sigma$. 
    Notice that if the condition in $p_e^\sigma$ is not satisfied, i.e., we do not have $\tau_i\tau_j =\xi_i\xi_j =\eta_i\eta_j =1$, then $p_e^\sigma=0$ and must be trivially $\le p_e^\tau$.
    Hence, we shall consider the case where the condition is satisfied.
    In this case, we see that the condition for $p_e^\tau$ is also satisfied, i.e., $\tau_i (\kappa_e \cos (\nabla_e \theta) +\lambda_e) >0$.
    
    For notation simplicity, let us omit the edge $e=ij$ subscript and also write $|X|=\cos \theta_i \cos \theta_j$ and $|Y| =\sin \theta_i \sin \theta_j$ (where $1\ge |X|,|Y|\ge 0$ due to the condition $\xi_i\xi_j=\eta_i\eta_j=+1$).  
    Then we have
    \begin{align}
        p_e^\tau -p_e^\sigma &= (e^{-4\beta \kappa |X|}+e^{-4\beta \kappa |Y|})e^{-2\beta \alpha 4|X||Y|} -e^{-2\beta \kappa (|X|+|Y|)} [e^{ -2\beta \lambda}+e^{-2\beta \kappa (|X|+|Y|)}] \\
        &=e^{-2\beta \kappa(|X|+|Y|)} [2 \cosh (2\beta \kappa (|X|-|Y|))e^{-2\beta \alpha 4|X||Y|} -e^{ -2\beta \lambda}-e^{-2\beta \kappa (|X|+|Y|)}] \\
    \end{align}
    Note that $\alpha \le \lambda$ and thus we have
    \begin{align}
        \cosh (2\beta \kappa (|X|-|Y|))e^{-2\beta \alpha 4|X||Y|} &\ge  e^{-2\beta \alpha 4|X||Y|} \\
        &\ge e^{-2\beta \lambda},  \quad 4|X||Y| = \sin 2\theta_i \sin 2\theta_j  \le 1
    \end{align}
    Note that $\kappa \ge 2\alpha$ and thus we have
    \begin{align}
        \cosh (2\beta \kappa (|X|-|Y|))e^{-2\beta \alpha 4|X||Y|} &\ge  e^{-2\beta \alpha 4|X||Y|} \\
        &\ge e^{-2\beta \kappa (|X|+|Y|)}
    \end{align}
    Where we used the fact that
    \begin{align}
        |X|+|Y|-\frac{2\alpha}{\kappa} 2|X||Y| &\ge |X|^2+|Y|^2-2|X||Y|,\quad |X|,|Y|\le 1 \\
        &\ge 0
    \end{align}
    Combining everything together, we get $p_e^\tau \ge p_e^\sigma$.
\end{proof}
\section{(Random) Current Representation - Critical Regime}
\label{app:current}

The current representation dates back to Griffith \textit{et. al.} \cite{griffiths1970concavity} and Aizenman \cite{aizenman2005geometric}, where it was developed to study the Ising model.
Similar to the cluster representation, the current representation establishes a correlation-percolation correspondence in the Ising model. 
This representation has proven useful in numerous occasions in obtaining rigorous results regarding the standard Ising model \cite{duminil2017lectures,aizenman1987phase,aizenman2015random,aizenman2021marginal}. 
However, unlike the cluster representation, the current representation has been mostly limited to the Ising model.
Only until recently has the representation been extended to the standard XY model \cite{van2023elementary}.
Therefore, in this section, we will provide a short review by considering the standard Ising and XY models.
The developed techniques will be used towards improving the correlation inequality in Theorem \eqref{thm:cluster} and proving Theorem \eqref{app-thm:current} for the critical Hamiltonian in Eq. \eqref{eq:U1-Z2}.
\subsection{A short review of the Ising model}

\begin{figure}[ht]
\subfloat[\label{fig:current-Is-part}]{%
  \centering
  \includegraphics[width=0.25\columnwidth]{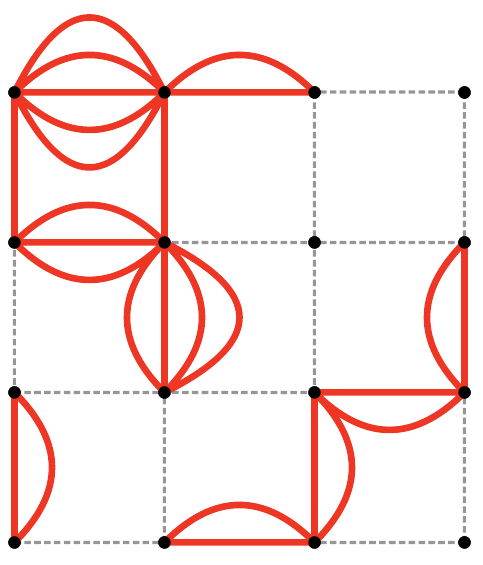}
}
\hspace{0.05\columnwidth}
\subfloat[\label{fig:current-Is-corr}]{%
  \centering
  \includegraphics[width=0.25\columnwidth]{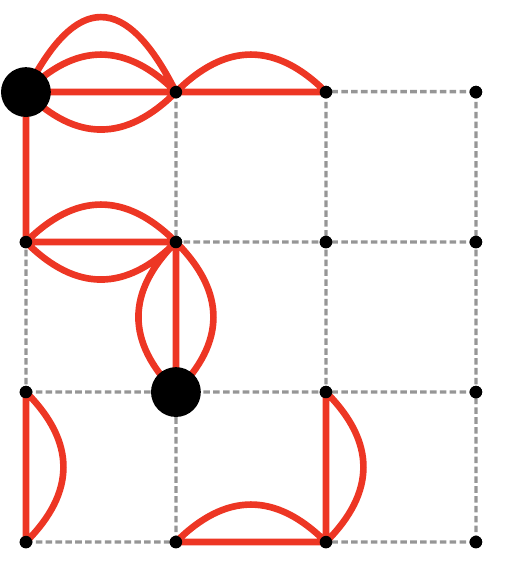}
}
\\
\subfloat[\label{fig:current-Is-double}]{%
  \centering
  \includegraphics[width=0.25\columnwidth]{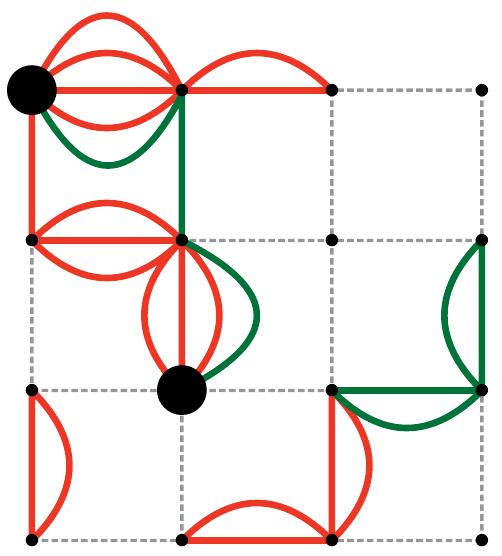}
}
\hspace{0.05\columnwidth}
\subfloat[\label{fig:current-Is-perc}]{%
  \centering
  \includegraphics[width=0.25\columnwidth]{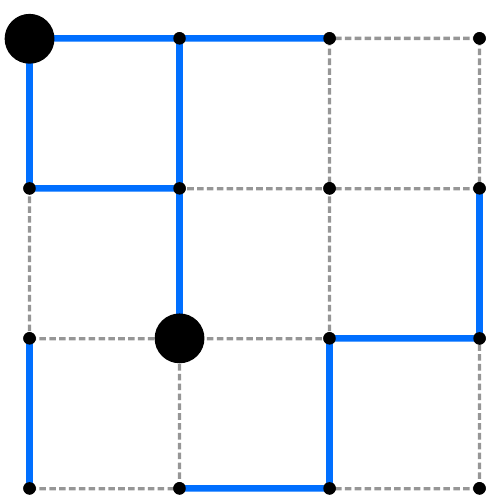}
}
\caption{Current representation of the Ising model on $\dZ^2$. 
(a) A current configuration $n:E\to \dN$ contributing to the partition function $Z^\text{Is}$. 
(b) A current configuration contributing to the unnormalized correlation function $Z^\text{Is}\bra \tau_0 \tau_R\ket^\text{Is}$ where the large dots denote the lattice sites $0,R$. 
(c). A double current representation, where the red and green lines representation distinct configurations $n_1,n_2$ contributing to the unnormalized correlation function $(Z^\text{Is}\bra \tau_0 \tau_R\ket^\text{Is})^2$. 
The large dots denote the lattice sites $0,R$.
(d). A double current representation, where only the trace $\hat{n}$ (compared to Fig. \ref{fig:current-Is-double}) is shown.}
\label{fig:current-Is}
\end{figure}

Consider the standard Ising model on a finite graph $G$,
\begin{equation}
    H^\text{Is}(\tau) = -\sum_{e=ij} \tau_i\tau_j
\end{equation}
So that the partition function is given by
\begin{equation}
    Z^\text{Is} = \sum_{\tau} \prod_{e=ij} e^{\beta \tau_i\tau_j}
\end{equation}
For each edge $e=ij$, we can expand the Boltzmann weight $e^{\beta \tau_i \tau_J}$ via Taylor series. 
This introduces an extra $\dN$ degree of freedom on each edge $e$ corresponding to the order of the Taylor series expansion, and thus a \textit{current} configuration on the edges $n:E \to \dN$ \cite{duminil2017lectures}.
Integrating over all spin configurations $\tau$ induces an interaction between the currents on distinct edges. 
More specifically, we have
\begin{equation}
    Z^\text{Is} = \sum_{n} \frac{\beta^n}{n!} \times 1\{\partial n=\emptyset \}
\end{equation}
Where $\beta^n \equiv  \prod_e \beta^{n_e}$ and $n! \equiv \prod_e n_e!$. 
Note that $\partial n$ denote the \textit{endpoints} of the current configuration $n$, that is, the set of vertices $i$ with an odd parity of neighboring edges in $n$, i.e., $\delta_i n = \sum_{e\sim i} n_e$.
A typical current configuration is shown in Fig. \ref{fig:current-Is-part}. 
Due to the restriction that $\partial n =\emptyset$, we see that each vertex must be adjacent to an even (possibly 0) number of edges in $n$.

One can perform a similar expansion for the unnormalized correlation function
\begin{equation}
    Z^\text{Is} \bra \tau_0\tau_R \ket^\text{Is} = \sum_{n} \frac{\beta^n}{n!} \times 1\{\partial n=0,R \}
\end{equation}
Where the only difference is that the summation only includes current configurations with endpoints at lattice sites $0,R$ (see Fig. \ref{fig:current-Is-corr} for a typical configuration).

Notice that the partition function and unnormalized correlation functions are written over different sums of current configurations, and thus cannot define a probability distribution.
Mathematicians have realized that one can circumvent this problem by applying a \textit{double} current representation \cite{duminil2017lectures,aizenman2015random}: create two identical copies of the same Hamiltonian so that one can compute the square of the unnormalized correlation functions, i.e.,
\begin{equation}
    ( Z^\text{Is} \bra \tau_0\tau_R \ket^\text{Is} )^2= \sum_{n_1,n_2} \prod_{i=1}^2 \frac{\beta^{n_i} }{n_i!} 1\{\partial n_i= 0,R \}
\end{equation}
The key observation is to note that while each (duplicated) current configuration $n_1,n_2$ must have endpoints at lattice site $0,R$, their sum\footnote{Performed over each edge $e$} $n\equiv n_1+n_2$ does not have any endpoints, as shown in Fig. \ref{fig:current-Is-double}.
The switching lemma \cite{duminil2017lectures,aizenman2005geometric} formalizes this intuition so that one can rewrite
\begin{align}
    ( Z^\text{Is} \bra \tau_0\tau_R \ket^\text{Is} )^2&= \sum_{n_1,n_2} 1\{0\underset{n_1+n_2}{\lr} R\}(\mu^\text{Is})^{\otimes 2}[n_1,n_2] \\
    (\mu^\text{Is})^{\otimes 2}[n_1,n_2] &=\prod_{i=1}^2 \frac{\beta^{n_i} }{n_i!} 1\{\partial n_i= \emptyset \} 
\end{align}
Where $\{0\lr_{n_1+n_2} R\}$ denotes the percolation event where lattice sites $0,R$ are connected in a cluster of $n_1+n_2$.
Since the partition functions $(Z^\text{Is})^2$ is the normalization factor of the weights $(\mu^\text{Is})^{\otimes 2}[n_1,n_2]$, we find that the correlation function can be written as the probability of a percolation event, i.e.,
\begin{align}
    (\bra \tau_0\tau_R \ket^\text{Is} )^2 &= (\dP^\text{Is})^{\otimes 2} [0\underset{n_1+n_2}{\lr}  R] \\
    (\dP^\text{Is})^{\otimes 2}[n_1,n_2] &= \frac{(\mu^\text{Is})^{\otimes 2}[n_1,n_2]}{(Z^\text{Is})^2}
\end{align}

We emphasize that the percolation event $\{0\lr_{n_1+n_2} R\}$ is only dependent on the sum $n \equiv n_1 +n_2$ of the two individual configurations; in fact, it only depends on the \textit{trace} (see Fig. \ref{fig:current-Is-perc}), which is a subgraph $\hat{n}:E\to \{0,1\}$ defined by including an edge $e$ in $\hat{n}$ if $n_e >0$, i.e., $\hat{n} = 1\{n >0\}$.
Therefore, we simplify our notation and write
\begin{equation}
    (\bra \tau_0\tau_R \ket^\text{Is} )^2 = (\dP^\text{Is})^{\otimes 2} [0\lr_{\hat{n}} R]
\end{equation}

\subsection{A short review of the XY model}

\begin{figure}[ht]
\subfloat[\label{fig:current-XY-part}]{%
  \centering
  \includegraphics[width=0.25\columnwidth]{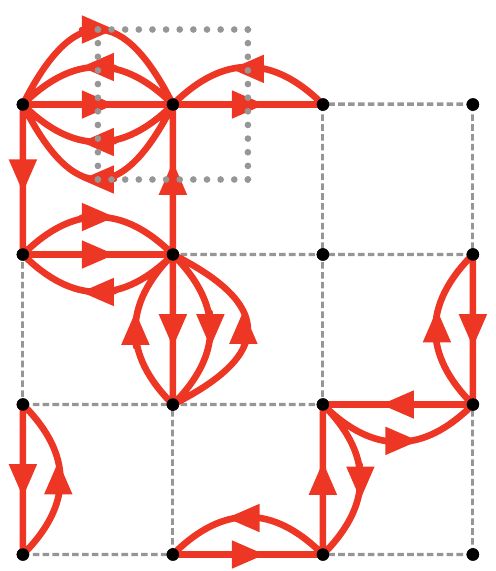}
}
\hspace{0.1\columnwidth}
\subfloat[\label{fig:current-XY-pair}]{%
  \centering
  \includegraphics[width=0.15\columnwidth]{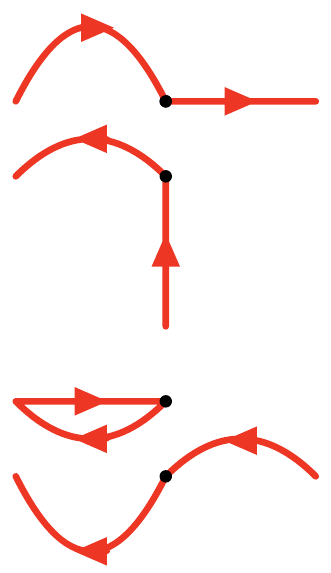}
}
\\
\subfloat[\label{fig:current-XY-loop}]{%
  \centering
  \includegraphics[width=0.3\columnwidth]{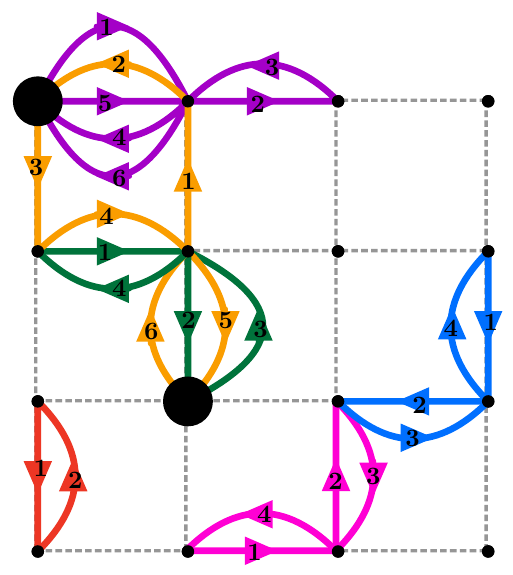}
}
\caption{Current representation of the XY model on $\dZ^2$. 
(a). A directed current configuration $\nn$ contributing to the partition function, i.e., with zero flow $\delta \nn =0$.
(b). A possible pairing of incoming and outgoing edges of the boxed vertex in Fig. \ref{fig:current-XY-part}.
(c). A possible cyclic decomposition $\bm{\sC}$ of the direct current $\nn$ in Fig. \ref{fig:current-XY-part}. Distinct colors denote distinct loops and the numbers denote the ordering, e.g., $1\to 2\to 3\to 4\to 5 \to 6\to 1$ for the orange cycle. Large black dots denote lattice site $0,R$.}
\label{fig:current-XY}
\end{figure}

Similar to the Ising model in the previous subsection, we can write the correlation functions of the standard XY model as percolation events in corresponding current representation.
Here, we shall outline the main argument and refer the reader to Ref. \cite{van2023elementary} for details (see Appendix \eqref{app:current-technical} for its application to the $U(1)\times \dZ_2$ Hamiltonian).
Consider the standard XY model on a finite graph $G$,
\begin{align}
    H^\text{XY} = -\sum_{e=ij} \sigma_i \cdot \sigma_j = -\frac{1}{2} \sum_{\ee=i\to j} \bar{\sigma}_i \sigma_j
\end{align}
Where $\sigma_i = e^{i\theta_i} \in \dS^1 \cong \dC$. 
Notice that in contrast to the Ising model, the summation is over directed edges $\ee \in \EE$ and thus when computing the partition function $Z^\text{XY}$, the induced current is a directed configuration, i.e., $\nn:\EE \to \dN$ and
\begin{equation}
    Z^\text{XY} = \sum_{\nn} \cdots \times 1\{ \delta \nn = 0\}
\end{equation}
Where we have omitted the corresponding weights $\cdots$ and chosen to focus on the interaction term induced by integrating over all XY spins.
More concretely, $\delta_i \nn$ is the \textit{flow} of the current out of the lattice site $i$, defined by $\delta_i \nn =\sum_{j\sim i} (\nn_{i\to j} - \nn_{j\to i})$ so that the condition $\delta \nn =0$ requires $\delta_i \nn =0$ at all lattice sites $i$.
A typical directed configuration satisfying this condition\footnote{
Note there is a subtlety regarding which edge copies should be pointing from $i\to j$ and which $j\to i$. 
Indeed, the directed current $\nn$ only specifies the number of directed edge copies pointing from $i\to j$ (and $j\to i$), but not the \textit{stacking} order.
For example, in Fig. \ref{fig:current-XY-part}, there are 5 edge copies on the upper-left most horizontal edge, of which 2 are point to the right and 3 to the left. 
Out of the 5 edge copies, we could've arranged the right and left arrows in any order.
See Appendix \eqref{app:redundancies} for explicit treatment on $U(1)\times \dZ_2$ model.
} is shown in Fig. \ref{fig:current-XY-part}.

It would then appear that we can repeat the procedure outlined for the Ising model and write the correlation functions as a percolation event of the direct currents $\nn$.
However, as first noticed by Ref. \cite{van2023elementary}, a further degree of freedom is required.
As shown in Fig. \ref{fig:current-XY-part}, the zero flow requirement implies that we can decompose the directed current configuration into distinct cycles.
More concretely, as shown in Fig. \ref{fig:current-XY-pair} (which corresponds to the boxed vertex in Fig. \ref{fig:current-Is-part}), we can pair up the incoming- and outgoing-edges in any manner.
By performing this decomposition at each vertex, we obtain a collection of directed cycles\footnote{
Since there are multiple ways of decomposition a given directed current $\nn$, the mapping $\nn \mapsto \bm{\sC}$ is multi-valued, i.e., the well-defined map $\bm{\sC} \mapsto \nn$ is not 1-1.} $\bm{\sC}$.
A typical configuration of $\bm{\sC}$ is shown in Fig. \ref{fig:current-XY-loop}, where each distinct color denotes a distinct directed loop and the numbers specify the ordering, e.g., $1\to 2\to 3\to 4\to 5 \to 6\to 1$ for the orange cycle.

By using an analogue of the switching lemma \cite{duminil2017lectures,aizenman2005geometric} for the Ising model, the following correspondence\footnote{A similar correspondence for the conventional correlations $\bra \cos(\theta_0-\theta_R)\ket^\text{XY}$ was also established using a corresponding \textit{double} current representation \cite{van2023elementary}.} was established
\begin{equation}
    \bra \cos 2(\theta_0 -\theta_R)\ket^\text{XY} \cong \dP^\text{XY}[ 0\lr_{\sC} R]
\end{equation}
Where $\cong$ denotes that the ratio of the two values are bounded above and below by constants\footnote{In this case, the ratio of the left-hand-side and the right-hand-side is bounded between $1/2 \le \cdots \le 1$.}, and $\{0\lr_{\sC} R\}$ denotes the percolation event in which there exists a directed loop in the collection $\bm{\sC}$ which connects lattice sites $0,R$ (e.g., orange cycle in Fig. \ref{fig:current-XY-loop}).
Similar to the Ising model, we emphasize that the percolation event $\{0\lr_{\sC} R\}$ is independent of the direction of each cycle\footnote{Also independent of the stacking order as defined previously.} in $\bm{\sC}$ and thus warrants the notation $\sC$ instead of $\bm{\sC}$.
\subsection{Application to the Critical Regime}

Similar to the cluster representation, the correlation-percolation correspondence in the current representation extends to the critical subclass in Eq. \eqref{eq:U1-Z2} and that the difficulty lies in finding the relation between two percolation event so that
\begin{theorem}[see Appendix \eqref{app:current-technical}]
    \label{app-thm:current}
    Let the $U(1)\times \dZ_2$ Hamiltonian $H$ in Eq. \eqref{eq:U1-Z2} be defined on a finite graph $G$. Then for any temperature,
    \begin{equation}
        \bra \cos 2(\theta_0-\theta_R) \ket_{G,\beta} \le \bra \tau_0\tau_R\ket_{G,\beta}
    \end{equation}
\end{theorem}
\begin{proof}[Sketch of Proof]
If we repeat the procedure for the standard XY models and define directed cyclic decomposition $\bm{\sC}$, we see that the $U(1)$ correlations are related to the percolation events $\{0\lr_{\sC} R\}$ (see Appendix \eqref{app:current-technical-U1}).
Since every directed cyclic decomposition corresponds to a directed current, i.e., $\bm{\sC} \to \nn$, and each directed current corresponds to an undirected current, i.e., $\nn \to n$ (remove the directions), we see that we can also establish a correspondence between the $\dZ_2$ correlations and the percolation event $\{0 \lr_{\hat{n}} R\}$ (see Appendix \eqref{app:current-technical-Z2}).
It's then clear that if $\bm{\sC}$ percolates between lattice sites $0,R$, then its induced undirected current configuration $n$ must also percolate, i.e.,
\begin{align}
    \{0\lr_{\sC} R\} &\subseteq \{0 \lr_{\hat{n}} R\} \\
    \dP [0\lr_{\sC} R] &\le \dP[0 \lr_{\hat{n}} R]
\end{align}
The statement then follows.
\end{proof}

One natural question that may arise is why not compare the conventional $U(1)$ correlations $\bra \cos (\theta_0 -\theta_R)\ket$ with the $\dZ_2$ correlations $\bra \tau_0\tau_R\ket$ since, schematically, both require a double current representation.
In short, the $U(1)\times \dZ_2$ model is (heuristically) ``already" a double current representation, since it is physically motivated by identical layers (e.g., twisted bilayer BSCCO).
More concretely, note that in the standard Ising model, the correlations are obtained by considering two identical copies of the same Hamiltonian, i.e.,
\begin{equation}
    H^\text{Is}(\sigma) +H^\text{Is}(\tau) = -\sum_{e=ij} (1+(\tau_i\sigma_i)(\tau_j \sigma_j)) \sigma_i \sigma_j
\end{equation}
Where we need to sum over all Ising spin configuration $(\sigma,\tau)$. Note, however, that $(\sigma,\tau)\mapsto (\sigma,\tau \sigma)$ (where $\tau \sigma$ denotes site-wise multiplication), is a bijective mapping between spin configurations, and thus the duplicated Ising Hamiltonian is equivalent to
\begin{equation}
    -\sum_{e=ij} (1+\tau_i\tau_j) \sigma_i \sigma_j
\end{equation}
This has the same form as the $U(1)\times \dZ_2$ Hamiltonian in Eq. \eqref{eq:U1-Z2} with constant $\kappa_e =1$, except that the $U(1)$ spins in the latter are replaced by the simpler $\dZ_2$ spins, i.e., $U(1)\mapsto \dZ_2$.
Therefore, heuristically, the $U(1)\times \dZ_2$ model is ``already" duplicated with respect to the Ising spins $\tau$.
\subsection{Suggestive Evidence of a Single Phase Transition}

Apart from providing a stronger inequality over Theorem \eqref{thm:cluster}, the current representation in Theorem \eqref{app-thm:current} also provides suggestive insight into why only a single phase transition is observed for the critical regime $\ell_\infty$ (at least for the constant $\kappa_e =1$ scenario on a regular lattice such as $\dZ^d$), as suggested by mean-field theory \cite{can2021high,yuan2023inhomogeneity} and numerics \cite{song2022phase,bojesen2014phase,maccari2022effects}.
More specifically, we have shown that
\begin{align}
    \bra \cos 2(\theta_0 -\theta_R) \ket \cong \dP[ 0\lr_{\sC} R] \\
    \bra \tau_0\tau_R \ket \cong \dP[ 0\lr_{\hat{n}} R]
\end{align}
Where the cyclic decomposition $\bm{\sC}$ induces the undirected current $n$ as argued previously, i.e., $\bm{\sC} \to \nn \to n$.
It is thus worth mentioning that the constructed current representation $\dP$ for the $U(1)\times \dZ_2$ Hamiltonian with constant $\kappa_e=1$ is related to that of the standard XY model $\dP^\text{XY}$ via the following relation
\begin{equation}
    \dP[\bm{\sC}] \propto 2^{C(\hat{n})}\dP^\text{XY}[\bm{\sC}]
\end{equation}
Where $C(\hat{n})$ denote the number of clusters in the induced trace $\hat{n}$ (e.g., $C(\hat{n})=6$ in the example of Fig. \ref{fig:current-Is-perc}, where isolated vertices are also counted).
Note the two percolation events $\{0 \lr_\sC R\}$ and $\{0 \lr_{\hat{n}} R\}$ also occur in the standard XY model.
However, since it's well-known that the standard XY model only has a single phase transition on $\dZ^d$, the two percolation events are \textit{expected} to coincide in the XY model\footnote{
This is not yet proven. 
However, such a conjecture was made in Ref. \cite{van2023elementary}, by comparing the cyclic decomposition $\bm{\sC}$ to simple random walks. More concretely, they consider the dimensional dependence of the XY phase transition (BKT or long-range order) and relate it to that of recurrence/transient property of the simple random walk on $\dZ^d$.}.
The extra factor $2^{C(\hat{n})}$ \textit{is not expected} to change this behavior, though it may change the critical temperature\footnote{One can compare this situation to the cluster representation of the $q$ states Potts model \cite{duminil2017lectures}, in which each probability distribution only differ by $q^{C(\omega)}$.}.
\subsection{Technical Details of the Current Representation}
\label{app:current-technical}
In this section, we shall derive some of the technical details for the random current representation given in previous sections of Sec. \eqref{app:current} for the $U(1)\times \dZ_2$ Hamiltonian in Eq. \eqref{eq:U1-Z2}.
\subsubsection{$\dZ_2$ Correlations}
\label{app:current-technical-Z2}
\begin{theorem}
    \label{app-app-thm:current-part-n}
    Let $H$ be that given in Eq. \eqref{eq:U1-Z2}. 
    Then the partition function $Z_{G,\beta}$ on a finite graph $G$ is given by
    \begin{equation}
        Z_{G,\beta} = \frac{1}{2^{V}}\sum_{\nn:\EE\to \dN} 2^{C(\hat{n})} \frac{(\beta \kappa)^{|\nn|}}{\nn !}  1\{\delta \nn=0\} 
    \end{equation}
    Where $|\nn|$ is the (undirected) current induced by $\nn$, i.e., $|\nn|_{ij}= \nn_{i\to j} +\nn_{j\to i}$, and 
    \begin{equation}
        \left(\beta \kappa\right)^{|\nn|} \equiv \prod_{e} \left(\beta \kappa_e\right)^{|\nn|_e}, \quad \nn! \equiv \prod_{\ee} \nn_{\ee} 
    \end{equation}
    $C(\hat{|\nn|})$ is the number of clusters in the trace $\hat{|\nn|}$ and $\delta \nn$ is the flow of $\nn$ as defined in the main text.
\end{theorem}
\begin{proof}
    Notice that
    \begin{align}
        Z &= \sum_{\tau}\int_\theta \prod_{\ee=i\to j\in \EE} \exp\left(\frac{\beta \kappa_e}{2} (1+\tau_i\tau_j) \bar{\sigma}_i\sigma_j \right) \\
        &=\sum_{\tau}\int_\theta \prod_{\ee=i\to j\in \EE} \sum_{\nn_{\ee}=0}^\infty \frac{1}{\nn_{\ee}!} \left(\frac{\beta \kappa_e}{2}\right)^{\nn_{\ee}} (1+\tau_i\tau_j)^{\nn_{\ee}} e^{-i(\theta_i -\theta_j)} \\
        &= \sum_{\nn:\EE\to \dN} \frac{1}{\nn!} \left(\frac{\beta \kappa}{2}\right)^{|\nn|} \left[\sum_{\tau} \prod_{e=ij\in E} (1+\tau_i\tau_j)^{|\nn|_e} \right] \left[\int_\theta \prod_{i\in V} e^{-i\theta_i \delta_i \nn} \right]
    \end{align}
    Where $|\nn|$ is the (undirected) current induced by $\nn$, i.e., $|\nn|_{ij}= \nn_{i\to j} +\nn_{j\to i}$, and 
    \begin{equation}
        \left(\frac{\beta \kappa}{2}\right)^{|\nn|} = \prod_{e} \left(\frac{\beta \kappa_e}{2}\right)^{|\nn|_e} 
    \end{equation}
    Note that by Lemma \eqref{app-lem:num-subloops}, we have
    \begin{align}
        \sum_{\tau} \prod_{e=ij} (1+\tau_i\tau_j)^{|\nn|_e} &= \sum_{m:E\to \dN,m\le |\nn|} \left( |\nn| \atop m \right) 1\{ \partial m=\emptyset \} \\
        &= 2^{|\nn|-V+C(\hat{|\nn|})}
    \end{align}
    Where 
    \begin{equation}
        \left( |\nn| \atop m \right) = \prod_{e} \left( |\nn|_e \atop m_e \right)
    \end{equation}
    And $C(\hat{|\nn|})$ is the number of clusters in the trace $\hat{|\nn|}$ as defined in the main text.
    Also notice that
    \begin{equation}
        \int_\theta \prod_{i\in V} e^{-i\theta_i \delta_i \nn} = 1\{\delta \nn =0\}
    \end{equation}
    Where $\delta \nn$ is the flow of the current as defined in the main text. 
    Hence, the statement follows.
\end{proof}

\begin{theorem}
    Let $\dP_{G,\beta}$ be the probability distribution on directed currents $\nn$ defined on a finite graph $G$ with weights
    \begin{equation}
        \dP_{G,\beta}[\nn] \propto 2^{C(\hat{n})} \frac{(\beta \kappa)^{|\nn|}}{\nn !}   1\{\delta \nn=0\}
    \end{equation}
    Where the notation is as in the previous Theorem \eqref{app-app-thm:current-part-n}, and let $\bra\cdots \ket_{G,\beta}$ be the thermal average with respect to the $U(1)\times \dZ_2$ Hamiltonian in Eq. \eqref{eq:U1-Z2}. Then
    \begin{equation}
        \bra \tau_0 \tau_R \ket_{G,\beta} = \dP_{G,\beta}[0\lr_{\hat{|\nn|}} R]
    \end{equation}
\end{theorem}
\begin{proof}
    For notation simplicity, let us omit the subscripts $G,\beta$.
    Notice that we can repeat the proof in obtaining the previous Theorem \eqref{app-app-thm:current-part-n} so that
    \begin{align}
        Z \bra \tau_0 \tau_R\ket &= \sum_{\nn:\EE\to \dN} \frac{1}{\nn!} \left(\frac{\beta \kappa}{2}\right)^{|\nn|} \left[\sum_{\tau} \tau_0 \tau_R\prod_{e=ij\in E}  (1+\tau_i\tau_j)^{|\nn|_e} \right] \left[\int_\theta \prod_{i\in V} e^{-i\theta_i \delta_i \nn} \right] \\
        &= \sum_{\nn:\EE\to \dN} \frac{1}{\nn!} \left(\frac{\beta \kappa}{2}\right)^{|\nn|} \left[\sum_{m:E\to \dN, m\le n} \left( |\nn| \atop m\right)  1\{\partial m = 0,R\} \right] 1\{\delta \nn=0\}
    \end{align}
    By the switching lemma \cite{duminil2017lectures} states that
    \begin{equation}
        \sum_{m:E\to \dN, m\le n} \left( |\nn| \atop m\right)  1\{\partial m = 0,R\}= 1\{0 \lr_{\hat{|\nn|} } R\}\sum_{m:E\to \dN, m\le n} \left( |\nn| \atop m\right)  1\{\partial m = \emptyset\}
    \end{equation}
    Therefore,
    \begin{align}
        Z \bra \tau_0 \tau_R\ket &= \sum_{\nn:\EE\to \dN} 1\{0 \lr_{\hat{|\nn|}} R \} \times 2^{C(\hat{n})} \frac{(\beta \kappa)^{|\nn|}}{\nn !}   1\{\delta \nn=0\} \\
        \bra \tau_0 \tau_R\ket_G &= \dP[ 0 \lr_{\hat{|\nn|}} R]
    \end{align}
\end{proof}

\begin{lemma}
    \label{app-lem:num-subloops}
    Let $G=(V,E)$ denote a finite graph with vertices $V$ and edges $E$. 
    Let $n:E\to \dN$ be a current configuration on $G$ with trace $\hat{n}$ as defined in the main text. Then
    \begin{equation}
        \sum_{m:E\to \dN,m\le n} \left( n \atop m \right) 1\{ \partial m=\emptyset \} = 2^{n-V+C(\hat{n})}
    \end{equation}
    Where $C(\hat{n})$ denotes the number of clusters in the subgraph $\hat{n}$, and $m\le n$ denote $m_e\le n_e$ for all edges $e$ and
    \begin{equation}
        \left( n \atop m \right) = \prod_{e\in E} \left( n_e \atop m_e \right)
    \end{equation}
\end{lemma}
\begin{proof}
    For a given $n:E\to \dN$, construct a multigraph $\sN$ as in Fig. \ref{fig:current-Is-part}, i.e., each edge $e$ is duplicated $n_e$ times.
    Note that if $m:E\to \dN$ is such that $m \le n$, then we can construct a sub-multigraph $\sM$ of $\sN$ by choosing $m_e$ edge copies of the total $n_e$ edge copies. 
    Since for each edge $e$, there are exactly
    \begin{equation}
        \left( n_e \atop m_e \right)
    \end{equation}
    many ways to select $m_e$ edges, we see that the summation computes the number of sub-multigraphs $\sM$ of $\sN$ without endpoints, i.e.,
    \begin{equation}
        \sum_{m:E\to \dN,m\le n} \left( n \atop m \right) 1\{ \partial m=\emptyset \} = \sum_{\sM\subseteq \sN} 1\{ \partial \sM=\emptyset \}
    \end{equation}
    Where $\partial \sM$ is defined similarly as $\partial m$.
    Since the summation counts the number of loops, by Theorem 1.9.5. of Ref. \cite{diestel2017extremal}, the statement follows.
\end{proof}
\subsubsection{$U(1)$ Correlations}
\label{app:current-technical-U1}
Similar to the standard XY model as discussed in the main text, the directed currents $\nn$ are insufficient to establish a correspondence between the $U(1)$ spin-spin correlation and percolation events in the corresponding current representation.
Extra degrees of freedom correspond to cyclic decompositions $\bm{\sC}$ is required. 
More specifically, we have
\begin{theorem}
    \label{app-app-thm:current-part-C}
    Let $H$ be that given in Eq. \eqref{eq:U1-Z2}. 
    Then the partition function $Z_{G,\beta}$ on a finite graph $G$ is given by
    \begin{equation}
        Z_{G,\beta} = \frac{1}{2^{V}}\sum_{\bm{\sC} } 2^{C(\hat{n})} \frac{(\beta  \kappa)^{n}}{n !}   \frac{1}{\lambda_n}1\{\bm{\sC} \mapsto n\} 
    \end{equation}
    Where $\nn$ is the directed current induced by $\bm{\sC}$ (unpair all the directed edges at each vertex as discussed in the main text. See Fig. \ref{fig:current-XY}) and $n$ is the undirected current induced by $\nn$, i.e., $\bm{\sC} \mapsto \nn \mapsto n$ and 
    \begin{equation}
        \lambda_n = \prod_{i\in V} (\delta_i n/2)!
    \end{equation}
\end{theorem}
\begin{proof}
    For notation simplicity, we shall omit the subscripts $G,\beta$.
    From Theorem \eqref{app-app-thm:current-part-n}, we see that
    \begin{equation}
        Z= \frac{1}{2^{V}}\sum_{\nn:\EE\to \dN} 2^{C(\hat{n})} \frac{(\beta \kappa)^{|\nn|}}{\nn !}  1\{\delta \nn=0\} 
    \end{equation}
    Let us rewrite this as
    \begin{align}
        Z &= \sum_{n:E\to \dN} 2^{C(\hat{n})} \frac{\left(\beta \kappa\right)^n }{n!}\underbrace{\left[ \sum_{\nn:\EE \to \dN} \frac{n!}{\nn!} 1\{|\nn| = n\}1\{\delta \nn = 0\}\right]}_{\sI_n}
    \end{align}
    Let us now attempt to compute $\sI_n$ in terms of cycle collections $\bm{\sC}$. 
    Indeed, given an undirected current $n:E\to \dN$ (see Fig. \eqref{fig:current-Is-part}) and a directed current $\nn$ which induces $n$, for each edge $e=ij$, there are
    \begin{equation}
        \frac{n_e!}{\nn_{i\to j}! \nn_{j\to i}}
    \end{equation}
    Many ways to assign $\nn_{i\to j}$ edge copies with the direction $i\to j$ and the rest with $j\to i$ (see. Fig. \eqref{fig:current-XY-part}. 
    Since this can be done independently for each edge, there are exactly $n!/\nn!$ many ways to assign directions to the undirected current $n$ so that it is consistent with the directed current $\nn$ (we refer to this as the \textit{stacking order}).
    As discussed in the main text, at each vertex, we can then pair up incoming and outgoing edges in any manner (see Fig. \eqref{fig:current-XY-pair}).
    Since the flow $\delta \nn =0$ for all vertices, we see that there are an equal number of incoming and outgoing edges and thus $(\delta_i n/2)!$ many ways to pair up edges.
    Since each vertex is independent, we see that there are exactly $\lambda_n$ many ways to pair up edges.
    Then the mapping from $\bm{\sC} \to n$ (reversing the decomposition from $n$ to $\nn$ and then to $\bm{\sC}$) must be $(n!/\nn! \times \lambda_n)$-to-1.
    Hence, we have
    \begin{equation}
        Z= \sum_{n:E\to \dN} 2^{C(\hat{n})} \frac{\left(\beta \kappa\right)^n }{n!}\times \frac{1}{\lambda_n} \left[ \sum_{\bm{\sC}} 1\{\bm{\sC} \mapsto n\}\right]
    \end{equation}
    Hence, the statement follows.
\end{proof}

\begin{theorem}
    Let $\dP_{G,\beta}$ be the probability distribution on directed cycle collections $\bm{\sC}$ defined on a finite graph $G$ with weights
    \begin{equation}
        \dP_{G,\beta}[\bm{\sC}] \propto 2^{C(\hat{n})} \frac{(\beta  \kappa)^{n}}{n !}   \frac{1}{\lambda_n}1\{\bm{\sC} \mapsto n\} 
    \end{equation}
    Where the notation is as in the previous Theorem \eqref{app-app-thm:current-part-C}, and let $\bra\cdots \ket_{G,\beta}$ be the thermal average with respect to the $U(1)\times \dZ_2$ Hamiltonian in Eq. \eqref{eq:U1-Z2}. Then
    \begin{equation}
        \frac{1}{2} \dP_{G,\beta}[0\lr_{\sC} R] \le \bra \cos 2(\theta_0 -\theta_R) \ket_{G,\beta} \le \dP_{G,\beta}[0\lr_{\sC} R]
    \end{equation}
\end{theorem}
\begin{figure}[ht]
\subfloat[\label{fig:current-cut}]{%
  \centering
  \includegraphics[width=0.4\columnwidth]{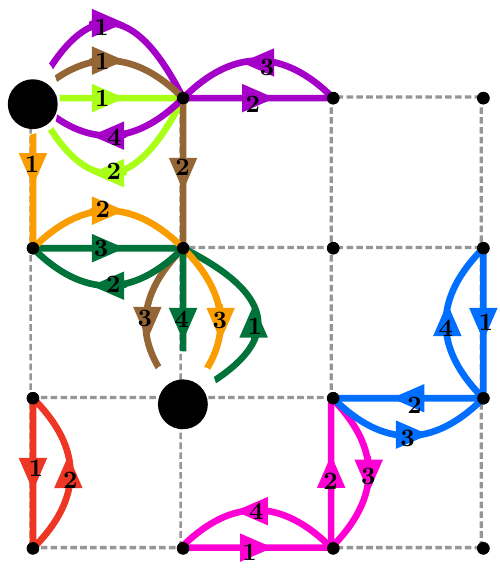}
}
\hspace{0.1\columnwidth}
\subfloat[\label{fig:cluster-cut-reverse}]{%
  \centering
  \includegraphics[width=0.4\columnwidth]{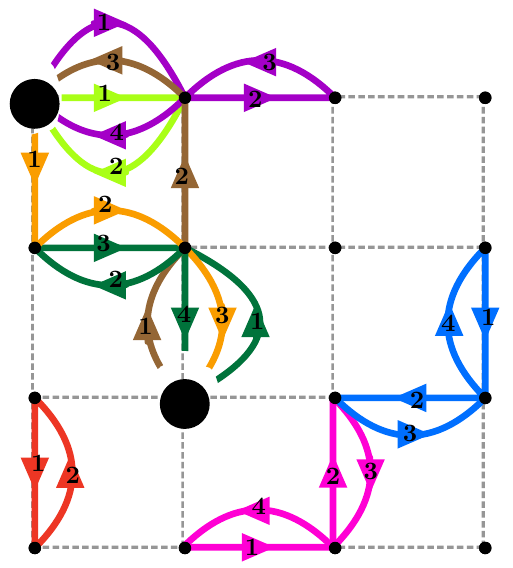}
}
\caption{Current Representation on $\dZ^2$. 
(a). Some cycle/path decomposition with cut at lattice sites $0,R$, wich are denoted by the large black dots. Distinct colors denote distinct loops/paths, while the numbers indicate the order of edges. 
Note, for example, the purple path start and ends at one of the cuts, and thus should be regarded as a path rather than a loop (the numbers are not modulo the length of the path, e.g., 4).
(b). The cycle/path decomposition of Fig. \ref{fig:current-cut}, but with the brown path reversed, i.e., $(\bm{\sC},\bm{\gamma})\mapsto (\bm{\sC}',\bar{\bm{\gamma}})$ where $\bm{\sC}'=\bm{\sC} \Delta \{\bm{\gamma},\bar{\bm{\gamma}}\}$ (where $\Delta$ denotes the symmetric difference).}
\label{fig:current-U1-Z2}
\end{figure}
\begin{proof}
    The proof follows that given in Ref. \cite{van2023elementary}, in which the $U(1)$ spin-spin correlation (but for the standard XY model) was related to a percolation event by reversing one of the 2 ``paths" from $0\to R$, and thus forming a new cycle. 
    Since path reversal does not change the undirected current $n$, the proof can be applied to our model, which only differs from the XY model by a extra weight of $2^{C(\hat{n})}$. 
    Indeed, we provide the details here for the $U(1)\times \dZ_2$ Hamiltonian, in a way that (the author feels) is more physically motivated (though ultimately the same as in Ref. \cite{van2023elementary}, which focuses a bit more on rigorous definitions).

    As before, we shall omit the subscripts $G,\beta$ for notation simplicity.
    Notice that we can repeat the proof of Theorem \eqref{app-app-thm:current-part-n} and obtain
    \begin{align}
        Z \bra \sigma_R^2 \bar{\sigma}_0^2 \ket &= \sum_{\nn:\EE\to \dN} 2^{C(\hat{|\nn|})} \frac{(\beta \kappa)^{|\nn|}}{\nn!} 1\{\delta \nn = 2\delta_0-2\delta_R\}
    \end{align}
    Where the flow $\delta \nn$ of the current is zero everywhere except at the latice sites 0,R, at which $\delta_i \nn =\pm 2$, respectively.
    It's then evident that the unnormalized correlation $Z \bra \sigma_R^2 \bar{\sigma}_0^2 \ket$ is equal to a summation over currents $\nn$ which have two ``paths" from $0\to R$; instead of loop configurations as in the partition function $Z$.
    Therefore, the intuition is to reverse one of the two ``paths" from $0\to R$ so that the resulting configuration is a loop configuration.
    The definition of a ``path" in its current form, however, is a bit ambiguous to achieve this, and thus warrants us to pair up incoming \& outgoing edges at each lattice site as done for the partition function in Theorem \eqref{app-app-thm:current-part-C}. 
    More concretely, rewrite the unnormalized correlation as such
    \begin{align}
        Z \bra \sigma_R^2 \bar{\sigma}_0^2 \ket &= \sum_{n:E\to \dN} 2^{C(\hat{n})}\frac{1}{n!} \left(\frac{\beta}{2}\right)^n \underbrace{\left[ \sum_{\nn:\EE \to \dN} \frac{n!}{\nn!} 1\{|\nn| = n\}1\{\delta \nn = f^{(2)}\}\right]}_{\sI_n^{(2)}}
    \end{align}
    Where we have introduced an extra degree of freedom corresponding to undirected currents $n:E\to \dN$ and $f^{(2)} = 2\delta_0-2\delta_R$.
    The summation $\sI_n^{(2)}$ is regulated by the condition $1\{|\nn| = n\}$, i.e., only sum over direct currents $\nn$ which induce $n$.
    As before in Theorem \eqref{app-app-thm:current-part-C}, given a fixed undirect current $n$, for each edge $e=ij\in E$, there are exactly $n!/\nn!$ ways to assign directions so that it is consistent with $\nn$ (see Fig. \eqref{fig:current-XY-part}).
    To generalize the notion of cycle decompositions $\bm{\sC}$ used in the previous theorem, let us define
    \begin{enumerate}
        \item A \textit{cycle/path decomposition} $\bm{\sC}$ of the undirected current $n$ is a partition of $n$ into directed paths and directed loops (loosely speaking, uses up all the edge copies in $n$ exactly once). See Fig. \ref{fig:current-cut}
        \item $\bm{\sC}$ has a \textit{cut} at lattice sites $S\subseteq V$ if every directed path/loop in $\bm{\sC}$ is segmented at lattice sites in $S$. See Fig. \ref{fig:current-cut}. 
        \item Denote $\bm{\sL}^{S}_n(f)$ to be the collection of all cycle/path decompositions $\bm{\sC}$ on $n$ with cuts $S$ and satisfies the flow equation $\delta \nn = f$ where $\bm{\sC}\mapsto \nn$. 
        In the case where $f=0$ everywhere and $S=\emptyset$, we denote $\bm{\sL}_n\equiv \bm{\sL}_n^\emptyset (f\equiv 0)$, and denote $\bm{\sL} \equiv \bigsqcup_{n} \bm{\sL}_n$ be the union of all possible cycle/path decompositions with no cuts and zero flow. The elements in $\bm{\sL}$ are exactly the cycle decomposition we defined previously for the partition function.
    \end{enumerate}
    Based on our previous observation, it's then clear that $\sI_n^{(2)} = |\bm{\sL}_n^V(f^{(2)})|$. 
    Moreover, notice that given cycle/path decomposition $\bm{\sC} \in \bm{\sL}_n^V(f^{(2)})$ (since we cut at every lattice site, $\bm{\sC}$ can be regarded as the collection of directed edge (copies)), and a lattice site $i\in V,i\ne 0,R$, there are exactly $\delta_i n/2$ incoming \& outgoing edges, respectively.
    Therefore, there are exactly $(\delta_i n/2)!$ many ways to pair up incoming \& outgoing edges at $i$, in which each pairing induces a distinct cycle decomposition $\bm{\sC}'\in \bm{\sL}_n^{V-i} (f^{(2)})$.
    Conversely, every cycle decomposition $\bm{\sC}'\in \bm{\sL}_n^{V-i} (f^{(2)})$ (with $i\ne 0,R)$ induces a cycle decomposition in $\bm{\sL}_n^{V} (f^{(2)})$ (by segmenting at $i$).
    Therefore, the mapping is $(\delta_i n/2)!$-to-$1$ from $\bm{\sL}_n^{V-i} (f^{(2)}) \to \bm{\sL}_n^{V} (f^{(2)})$, i.e., given fixed $\bm{\sC} \in \bm{\sL}_n^V(f^{(2)})$, we have
    \begin{equation}
        \sum_{\bm{\sC}'\in \bm{\sL}_n^{V-i} (f^{(2)})} 1\{\bm{\sC}' \mapsto \bm{\sC}\} = (\delta_i n/2)!
    \end{equation}
    And in particular,
    \begin{equation}
        \frac{|\bm{\sL}_n^{V-i}(f^{(2)})|}{|\bm{\sL}_n^V(f^{(2)})|} = (\delta_i n/2)!
    \end{equation}
    We can repeat this argument inductively on all lattice sites $i\ne 0,R$ to obtain
    \begin{equation}
        \sum_{\bm{\sC}'\in \bm{\sL}_n^{0,R} (f^{(2)})} 1\{\bm{\sC}' \mapsto \bm{\sC}\} =
        \frac{|\bm{\sL}_n^{0,R}(f^{(2)})|}{|\bm{\sL}_n^V(f^{(2)})|} =  \lambda^{0,R}_n
    \end{equation}
    Where we use $\lambda^S_n$ to denote
    \begin{equation}
        \lambda^S_n \equiv \prod_{i\notin S} (\delta_i n/2)!, \quad \lambda_n = \lambda^\emptyset_n
    \end{equation}
    Therefore, 
    \begin{equation}
        \sI_n^{(2)} = \frac{1}{\lambda^{0,R}_n} |\bm{\sL}_n^{0,R}(f^{(2)})| =\frac{1}{\lambda^{0,R}_n} \sum_{\bm{\sC}\in \bm{\sL}_n^{0,R}(f^{(2)})} 1
    \end{equation}
    Notice that for every $\bm{\sC}\in \bm{\sL}_n^{0,R}(f^{(2)})$ (see Fig. \eqref{fig:current-cut}), the number of directed paths from $0 \to R$ in the collection $\bm{\sC}$ must be 2 more than that from $R\to 0$, i.e., $|\bm{\sP}_{0\to R}(\bm{\sC})| =|\bm{\sP}_{R\to 0}(\bm{\sC})| + 2 \ge 2$, where $\bm{\sP}_{a\to b}(\bm{\sC})$ is the subset of $\bm{\sC}$ consisting of directed path from $a\to b$.
    For every directed path $\bm{\gamma}\in \bm{\sP}_{0\to R}(\bm{\sC})$, we can reverse the direction to obtain $\bar{\bm{\gamma}}$, and replace $\bm{\gamma}\mapsto \bar{\bm{\gamma}}$ within the collection $\bm{\sC}$ to obtain the new decomposition, i.e., $\bm{\sC}' =\bm{\sC}\Delta \{\bm{\gamma},\bar{\bm{\gamma}}\}$ (where $\Delta$ is the symmetric difference).
    See Fig. \ref{fig:cluster-cut-reverse} for an example.
    The induced cycle decomposition satisfies $\bm{\sC}'\in \bm{\sL}_n^{0,R}(f\equiv 0)$ and $\bar{\bm{\gamma}}\in \bm{\sP}_{R\to 0} (\bm{\sC}')$, and thus we obtain a mapping $(\bm{\sC},\bm{\gamma}) \mapsto (\bm{\sC}',\bar{\bm{\gamma}})$.
    A similar argument shows that the mapping is injective (1-to-1), and thus
    \begin{align}
         \sum_{\bm{\sC}\in \bm{\sL}_n^{0,R}(f^{(2)})} 1 &=  \sum_{\bm{\sC}\in \bm{\sL}_n^{0,R}(f^{(2)})} \sum_{\bm{\gamma}\in \bm{\sP}_{0\to R}(\bm{\sC})} \frac{1}{|\bm{\sP}_{0\to R}(\bm{\sC})|} \\
         &= \sum_{(\bm{\sC},\bm{\gamma})} \frac{1}{|\bm{\sP}_{0\to R}(\bm{\sC})|} \\
         &=  \sum_{(\bm{\sC}',\bm{\gamma}')} \frac{1}{|\bm{\sP}_{0\to R}(\bm{\sC'})| +1} \\
         &= \sum_{\bm{\sC}' \in \bm{\sL}_n^{0,R}(f\equiv 0)} \frac{|\bm{\sP}_{0\to R}(\bm{\sC'})|}{|\bm{\sP}_{0\to R}(\bm{\sC'})| +1}
    \end{align}
    Notice that by a similar argument, we have
    \begin{equation}
        \sum_{\bm{\sC}\in \bm{\sL}_n^{0,R} (f\equiv 0)} 1\{\bm{\sC} \mapsto \bm{\sC}'\} =
        \frac{|\bm{\sL}_n|}{|\bm{\sL}_n^{0,R}(f^{(2)})|} = \prod_{i= 0,R} (\delta_i n/2)! 
    \end{equation}
    Therefore,
    \begin{align}
        \frac{1}{\lambda_n^{0,R}}\sum_{\bm{\sC}\in \bm{\sL}_n^{0,R}(f^{(2)})} 1 &= \frac{1}{\lambda_n}\sum_{\bm{\sC} \in \bm{\sL}_n} \frac{|\bm{\sP}_{0\to R}(\bm{\sC})|}{|\bm{\sP}_{0\to R}(\bm{\sC})| +1}
    \end{align}
    Where we have abused notation and also use $\bm{\sP}_{0\to R}(\bm{\sC})$ to denote the collection of directed paths from $0 \to R$ after cutting $\bm{\sC}$ at $0,R$.
    In particular, we find that
    \begin{align}
        Z \bra \sigma_R^2 \bar{\sigma}_0^2 \ket &= \sum_{n:E\to \dN} 2^{C(\hat{n})}\frac{1}{n!} \left(\frac{\beta}{2}\right)^n  \frac{1}{\lambda_n}\sum_{\bm{\sC}\in \bm{\sL}_n} \frac{|\bm{\sP}_{0\to R}(\bm{\sC})|}{|\bm{\sP}_{0\to R}(\bm{\sC})| +1} \\
        &= \sum_{\bm{\sC}\in \bm{\sL}} 2^{C(\hat{n})} \frac{1}{n!} \left(\frac{\beta}{2}\right)^n  \frac{1}{\lambda_n} \times  \frac{|\bm{\sP}_{0\to R}(\bm{\sC})|}{|\bm{\sP}_{0\to R}(\bm{\sC})| +1}
    \end{align}
    Where it's understood that $n:E\to \dN$ the current obtained from $\bm{\sC}\in \bm{\sL}$, i.e., $\bm{\sC}\mapsto \nn \mapsto n$.
    Therefore,
    \begin{equation}
        \bra \sigma_R^2 \bar{\sigma}_0^2 \ket= \dE \left[ \frac{|\bm{\sP}_{0\to R}(\bm{\sC})|}{|\bm{\sP}_{0\to R}(\bm{\sC})| +1}\right]
    \end{equation}
    Notice, however, that
    \begin{equation}
         \frac{1}{2}  1\{ |\bm{\sP}_{0\to R}(\bm{\sC})| \ge 1\} \le \frac{|\bm{\sP}_{0\to R}(\bm{\sC})|}{|\bm{\sP}_{0\to R}(\bm{\sC})| +1}\le 1\{ |\bm{\sP}_{0\to R}(\bm{\sC})| \ge 1\}
    \end{equation}
    Therefore,
    \begin{equation}
        \frac{1}{2} \dP[|\bm{\sP}_{0\to R}(\bm{\sC})| \ge 1] \le \bra \sigma_R^2 \bar{\sigma}_0^2 \ket \le \dP [|\bm{\sP}_{0\to R}(\bm{\sC})| \ge 1]
    \end{equation}
    Notice that the event $\{ |\bm{\sP}_{0\to R}(\bm{\sC})| \ge 1\}$ is exactly the event $\{0\lr_{\sC} R\}$. Therefore, the statement follows.
\end{proof}
\subsubsection{Redundancies in the Percolation Event $\{0\lr_{\sC} R\}$}
\label{app:redundancies}
As discussed in the main text, there are redundancies when considering the percolation event $\{0 \lr_{\bm{\sC}} R\}$. 
\begin{enumerate}
    \item The event $\{0 \lr_{\bm{\sC}} R\}$ does not depend on the direction of each directed cycle $\bm{\gamma}_i\in \bm{\sC}$ in the cycle decomposition $\bm{\sC} = \{\bm{\gamma}_1,...,\bm{\gamma}_\ell \}$, and thus there exists a $2^\ell$ degeneracy.
    \item Since $\bm{\sC}$ is a cycle decomposition of a unique current $n:E\to \dN$, there is a $n!\equiv \prod_{e\in E} n_e!$ redundancy of which edge (copy) of $e=ij \in E$ the cycle traverses (referred to as the \textit{stacking order}).
    Since the percolation event $\{0 \lr_{\sC} R\}$ does not depend on the stacking order, we can integrate over this redundancy to obtain a factor of $n!$.
\end{enumerate}
After removing the direction and stacking order, the resulting equivalence class is a collection $\sC =\{\gamma_1,...,\gamma_\ell\}$ of undirected closed random walks $C_i$ on $G$, i.e., $\bm{\sC} \mapsto \sC$ is a $2^\ell n!$-to-1 mapping.
In particular, we find that
\begin{equation}
    \dP[\sC] \propto 2^{C(\omega_{\sC})}\left[\prod_{\gamma\in \sC} 2 \left(\frac{\beta}{2}\right)^{|\gamma|} \right] \frac{1}{\lambda_{\sC}}1\{\sC\in \sL\}  \propto 2^{C(\omega_\sC)} \dP_G^{\mathrm{XY}}[\sC]
\end{equation}
Where $|\gamma|$ is the length of the closed random walk $\gamma \in \sC$, while $\sL$ is the collection of all possible sets of closed random walks. 
We also denote $\omega_{\sC}:E\to \{0,1\}$ to be the subgraph consisting of edges traversed by the random walks $\gamma\in \sC$, so that $C(\omega_{\sC})$ is the number of clusters in $\omega_{\sC}$.
Similarly, $\lambda_{\sC}$ is defined to be
\begin{equation}
    \lambda_{\sC} = \prod_{i\in V} N_i (\sC)!, \quad N_i(\sC) \equiv \sum_{\gamma\in \sC} N_i(\gamma)
\end{equation}
Where $N_i(\gamma)$ is the number of times the closed random walk visits site $i\in V$. Therefore, we find that
\begin{align}
    \dP_G[0\lr_{\hat{n}} R] &= \dP_G[0\lr_{\omega_\sC} R] \\
    \dP_G[0\lr_{\bm{\sC}} R] &= \dP_G[0\lr_{\sC} R]
\end{align}
Where the percolation event $\{0\lr_{\sC} R\}$ corresponds to the event in which a random walk $\gamma\in \sC$ visits both lattice sites $0,R$.

\section{Simplified Models - Critical Regime}
\label{app:simple}

So far, we have shown that the system in Eq. \eqref{eq:U1-Z2} cannot have a floating phase, i.e., must satisfy $T_\text{TRSB} \ge T_c$ on any lattice. 
Since it corresponds to the critical regime $\ell_\infty$ in the strong coupling limit, it is expected to have a single phase transition and thus the converse inequality $T_\text{TRSB} \le T_c$ should also hold, albeit being much more difficult to prove.
Therefore, in this section, we will attempt to provide insight by studying two simplified models, i.e., replacing the $U(1)$ degree of freedom $\sigma$ in Eq. \eqref{eq:U1-Z2} with $\dZ_4$ and $\dZ_2$ clock models.
In fact, with this simplification, we are capable of studying a larger class of $U(1)\times \dZ_2$ Hamiltonian, i.e., before simplification,
\begin{equation}
    \label{eq:U1-Z2-rho}
    H_\rho (\sigma,\tau)= -\sum_{e=ij} \kappa_e (\rho_+ +\rho_- \tau_i\tau_j)\sigma_i \cdot\sigma_j
\end{equation}
Where $\rho \in [0,1]$ and $\rho_\pm =1\pm \rho$. Note that $\rho =0$ corresponds to the Hamiltonian in Eq. \eqref{eq:U1-Z2}.
The simplifications $U(1)\mapsto \dZ_2,\dZ_4$ can then be regarded as adding an arbitrary interaction $-\lambda_2 \sum_i \cos (2\theta_i)$ and $-\lambda_4 \sum_i \cos (4\theta_i)$ and taking the limit $\lambda_2,\lambda_4 \to \infty$, respectively.

Similar to Eq. \eqref{eq:U1-Z2}, the larger class $H_\rho$ with constant $\kappa_e = 1$ can be mapped to the strong coupling limit of twisted bilayer BSCCO and frustrated $n\ge 3$-band superconductors. 
Indeed, the latter was derived in Eq. (8) of Ref. \cite{bojesen2014phase}.
The former corresponds to the case where there is nonzero 1\ts{st} order Josephson coupling $J_1$ between the two layers, i.e., in addition to Eq. \eqref{eq:H-J2}, we have
\begin{equation}
    \label{eq:H-J1J2}
    H_{J_1,J_2} (\phi^\pm) = H_{J_2}(\phi^\pm) - J_1\sum_i \cos \phi_i
\end{equation}
In this case, $H_\rho$ corresponds to the Hamiltonian $H_{J_1,J_2}$ with $\rho = J_1/4J_2$ kept fixed while we take $J_2 \to \infty$.

Due to the correspondence, it is conjectured that the two transitions $T_c, T_\text{TRSB}$ coincide at $\rho=0$ and split for $\rho\ne 0$, i.e., $T_\text{TRSB} <T_c$. 
Indeed, an exact understanding of $H_{J_1,J_2}$ is known within the context of mean-field theory \cite{yuan2023inhomogeneity} without the need for simplifications.
Therefore, the simplifications $U(1)\mapsto \dZ_4, \dZ_2$ are to probe the behavior of the system in low dimensions (though we will prove the statements for any lattice).
\subsection{$U(1) \mapsto \dZ_2$}
Consider the simplification
\begin{equation}
    \label{eq:Z2-Z2-rho}
    H_\rho^{\dZ_2}(\sigma,\tau) = -\sum_{e=ij} \kappa_e (\rho_+ +\rho_- \tau_i\tau_j)\sigma_i \sigma_j
\end{equation}
Where $\sigma_i =\pm 1$.
Since $(\sigma,\tau)\mapsto (\sigma,\tau\sigma)$ (where $\tau \sigma$ denotes site-wise multiplication) is a bijective mapping between spin configurations, we see that $H_\rho^{\dZ_2}$ is equivalent to decoupled Ising models, and thus we have the following statement.
\begin{theorem}
    \label{thm:Z2-map}
    Let $G$ be any finite graph, and $\bra\cdots \ket^{\mathrm{Is}}_{\beta}$ denotes the thermal average with respect to the Ising model (with $s_i =\pm 1$) on $G$ with edge-coupling $\kappa_e$ and inverse temperature $\beta$. Then
    \begin{align}
        \bra \sigma_0\sigma_R \ket_{\rho,\beta}^{\dZ_2} &= \bra s_0 s_R\ket^\mathrm{Is}_{\beta\rho_+} \\
        \bra \tau_0\tau_R \ket_{\rho,\beta}^{\dZ_2} &= \bra s_0 s_R\ket^\mathrm{Is}_{\beta\rho_+} \bra s_0 s_R\ket^\mathrm{Is}_{\beta\rho_-}
    \end{align}
\end{theorem}
\begin{proof}
    Notice that the Hamiltonian is given by
    \begin{equation}
        H^{\dZ_2}_\rho = -\sum_{e=ij\in E} \kappa_e (\rho_+ +\rho_- \tau_i \tau_j) \sigma_i \sigma_j
    \end{equation}
    When computing the partition function, we need to sum over all spin configurations $(\tau,\sigma)$. 
    Notice, however, that $(\tau,\sigma)\mapsto (\tau \sigma,\sigma)$ is a bijective map, and thus we have
    \begin{align}
        Z_{\rho,\beta}^{\dZ_2} &= \sum_{\tau,\sigma} \exp\left[ \beta\sum_{e=ij\in E} \kappa_e (\rho_+ +\rho_- \tau_i \tau_j) \sigma_i \sigma_j \right] \\
        &= \sum_{\tau',\sigma} \exp\left[ \beta\sum_{e=ij\in E} \kappa_e (\rho_+ +\rho_- \tau'_i \tau'_j \sigma_i \sigma_j) \sigma_i \sigma_j \right], \quad \tau'=\tau \sigma\\
        &= \left[ \sum_{\sigma} \exp{\left(\beta\rho_+ \sum_{e=ij}\kappa_e \sigma_i \sigma_j\right)}\right]\left[ \sum_{\tau'} \exp{\left(\beta\rho_- \sum_{e=ij}\kappa_e \tau'_i \tau'_j\right)}\right]\\
        &= Z^{\mathrm{Is}}_{\beta \rho_+} \times Z^{\mathrm{Is}}_{\beta \rho_-,}
    \end{align}
    A similar transform can be applied to the correlation functions, and thus the statement follows.
\end{proof}
The previous correspondence then implies that
\begin{equation}
    ``T_c" = \rho_+ T^\mathrm{Is}_c \quad\ge \quad T_{\text{TRSB}} = \rho_- T^\mathrm{Is}_c
\end{equation}
Where equality only holds at $\rho =0$. 
Here, $T_c^\text{Is}$ is the critical temperature of the corresponding Ising model and we write $``T_c"$ since the original $U(1)$ spins are replaced by $\dZ_2$ spins. 

\subsection{$U(1) \mapsto \dZ_4$}
Note that the previous simplification $U(1)\mapsto \dZ_2$ was trivial in the sense that the system is can be exactly mapped to decoupled Ising models.
In this section, we shall consider the slightly more general case where $U(1)\mapsto \dZ_4$, i.e., $\sigma_i = (1,0),(0,1),(-1,0),(0,-1)$ at each lattice site $i$.
It should be noted that the $\dZ_4$ clock model is special in the sense that the $(X,Y)$ degree of freedom can be replaced by 2 independent Ising degrees of freedom, i.e., $\xi,\eta=X\pm Y$.
Therefore, the Hamiltonian can be rewritten as
\begin{equation}
    \label{eq:Z4-Z2-rho}
    H^{\dZ_4}_\rho = -\frac{1}{2}\sum_{e = ij } \kappa_e (\rho_+ +\rho_- \tau_i\tau_j) (\xi_i \xi_j +\eta_i \eta_j)
\end{equation}

To simplify notation, let $\bra \psi \ket$ be shorthand for the correlation $\bra \psi_0 \psi_R \ket^{\dZ_4}_{\rho,\beta} $ where $\psi$ can be any of the 7 options $\psi = \tau, \xi, \eta, \tau \xi, \tau \eta, \xi \eta, \tau \xi \eta$, each of which can (in principle) define a separate transition temperature $T_\psi$. 
However, using the symmetries of the Hamiltonian in Eq. \eqref{eq:Z4-Z2-rho} and the corresponding FKG inequalities \cite{velenik,duminil2017lectures,ginibre1970general}, the following (in)equalities can be shown.

\begin{theorem}
    For any $0\le \rho \le 1$,
    \begin{align}
        \bra \xi\ket &=\bra \eta\ket, \quad    \bra \tau \xi \ket = \bra\tau \eta \ket \\
        \bra \tau \ket &= \bra \tau \xi \eta\ket \ge \bra \tau \xi\ket \bra \xi\ket\\
        \bra \xi\ket  &\ge \bra \tau \xi \ket \ge \bra \tau\ket \bra \xi\ket \\
        \bra \xi\ket & \ge \bra \xi\eta\ket \ge \bra \xi\ket^2
    \end{align}
    Moreover,
    \begin{equation}
        \bra \xi \eta \ket \ge \bra \tau\ket
    \end{equation}
    Where equality holds if $\rho =0$.
\end{theorem}
\begin{proof}
    From the Hamiltonian in Eq. \eqref{eq:Z4-Z2-rho}, it's  clear that $\bra \xi\ket =\bra \eta\ket$ and $\bra \tau \xi\ket =\bra \tau \eta\ket$ due to the symmetry $\xi \lr \eta$.
    Also notice that when summing over the configuration $(\tau, \xi, \eta)$, the mapping $\tau \mapsto \tau\xi \eta$ is a bijective transformation which keeps the Hamiltonian invariant. 
    Therefore, $\bra \tau \ket=\bra \tau \xi\eta\ket$.
    By Ginibre's/FKG inequality (Prop 3, 5 and Example 4 of Ref. \cite{ginibre1970general}), we see that
    \begin{align}
        \bra \tau \xi \eta\ket &\ge \bra \tau \xi\ket \bra \eta\ket =\bra \tau \xi\ket \bra \xi\ket \\
        \bra \tau \xi\ket &\ge\bra \tau\ket \bra \xi\ket \\
        \bra \xi\eta \ket &\ge \bra \xi\ket \bra \eta \ket =\bra \xi\ket^2
    \end{align}

    It should be noted that the conditional expectation $\bra \xi\ket^\tau$ (i.e., computing the thermal average of $\xi_0 \xi_R$ if the $\dZ_2$ configuration $\tau$ were fixed) is equal to that of an Ising model $-\sum_e \kappa_e' \xi_i \xi_j$ with edge coupling
    \begin{equation}
        \kappa_e' = \frac{1}{2} \kappa_e (\rho_+ +\rho_- \tau_i\tau_j) \ge 0, \quad e=ij
    \end{equation}
    Therefore, by Griffiths 1\ts{st} inequality \cite{ginibre1970general,velenik,duminil2017lectures}, we see that $\bra \xi\ket^{\tau} =\bra \xi\ket^{\mathrm{Is}'} \ge 0$, and thus
    \begin{align}
        \bra \tau \xi\ket =\bra \tau \bra \xi\ket^{\tau} \ket \le \bra \bra \xi\ket^\tau \ket =\bra\xi\ket
    \end{align}
    Where the first equality can be understood as first averaging over $\xi, \eta$, then averaging over $\tau$. 
    The inequality uses the trivial fact $\tau_0 \tau_R \le 1$. 
    Similarly, we have
    \begin{equation}
        \bra \xi\eta\ket = \bra \xi \bra \eta\ket^{\xi,\tau} \ket \le \bra \bra \eta\ket^{\xi,\tau}\ket =\bra \eta\ket =\bra \xi\ket
    \end{equation}
    Where $\bra \eta\ket^{\xi,\tau} $ is the conditional expectation with respect to fixing $\xi, \tau$.

    For the last inequality, note that $\bra \tau \ket =\bra \tau \xi \eta\ket$ and as before, the conditional expectation $\bra \xi \eta \ket^\tau =\bra \xi\ket^{\mathrm{Is}'} \bra \eta \ket^{\mathrm{Is}'}$ (fixing the $\dZ_2$ spin configuration $\tau$) is equal to two independent Ising models with edge couplings $\kappa_e'$.
    Therefore, by Griffiths 1st inequality \cite{ginibre1970general,velenik,duminil2017lectures}, we see that $\bra \xi\ket^{\mathrm{Is}'}, \bra \eta\ket^{\mathrm{Is}'}\ge 0$ and thus
    \begin{align}
        \bra \tau \ket &=\bra\tau \xi\eta\ket \\
        &= \bra \tau \bra \xi \eta\ket^\tau\ket \\
        &\le \bra  \bra \xi \eta\ket^\tau\ket, \quad \tau_0\tau_R \le 1 \\
        &= \bra \xi\eta \ket
    \end{align}
    In the case where $\rho =0$, we see that $\tau \lr \xi \eta$ is a bijective transform which leaves the Hamiltonian invariant and thus we obtain equality,
    \begin{equation}
        \bra \tau\ket =\bra\xi\eta\ket
    \end{equation}
\end{proof}
Using the inequalities, it's straightforward to check that $T_\text{TRSB} \equiv  T_\tau$ and $``T_c" \equiv T_{\xi \eta}$ are the only (possibly) independent transition temperatures\footnote{
Due to the transform $\xi,\eta = X\pm Y$, we see that $\bra \cos 2(\theta_0 -\theta_R) \ket \sim 2 \bra \cos (2\theta_0) \cos (2\theta_R) \ket \sim \bra \xi \eta \ket$, and that $\bra \cos (\theta_0 -\theta_R )\ket \sim \bra \xi \ket$.
}, i.e.,
\begin{align}
    T_\xi &= T_\eta =T_{\xi\eta}\\
    T_\tau &= T_{\tau \xi \eta} \\
    T_{\tau \xi} &= T_{\tau \eta} = \min (T_\tau, T_{\xi\eta}) 
\end{align}
Moreover, the last inequality implies that
\begin{equation}
    T_\text{TRSB} \le ``T_c"
\end{equation}
Where equality holds\footnote{In contrast to the $U(1)\to \dZ_2$ simplification, we cannot tell whether the transitions split for $\rho \ne 0$.} at $\rho =0$. 
Hence, we see that the system cannot exhibit vestigial order on any lattice despite having short-range interactions (in comparison to claims of Ref. \cite{zeng2021phase}); at most, the two transitions coincide at the critical point $\rho =0$.

\end{document}